\documentclass[12pt,twoside,openright]{report}
\usepackage{epsfig,amsfonts,amsthm}
\usepackage{amsmath,amssymb}

\newcommand{\be}{\begin{equation}}
\newcommand{\ee}{\end{equation}}
\newcommand{\bea}{\begin{eqnarray}}
\newcommand{\eea}{\end{eqnarray}}
\newcommand{\dst}{\displaystyle}
\newcommand{\fr}[2]{\frac{{\dst #1}}{{\dst #2}}}

\newcommand{\f}{\phi}
\newcommand{\fd}{\phi^\dagger}
\renewcommand{\Re}{\mathrm{Re }}
\renewcommand{\Im}{\mathrm{Im }}

\newcommand{\doublet}[2]{ \left( \begin{array}{c}#1 \\ #2 \end{array}\right) }

\newcommand{\lr}[1]{ \langle #1 \rangle}
\newcommand{\Tr}{\mathrm{Tr}}
\newcommand{\Z}{\mathbb{Z}}

\newcommand{\stolbik}[2]{ \left( \begin{array}{c}#1 \\ #2 \end{array}\right) }

\newtheorem{theorem}{Theorem}
\newtheorem{proposition}[theorem]{Proposition}
\newtheorem{conjecture}[theorem]{Conjecture}

\def\lsim{\mathrel{\rlap{\lower4pt\hbox{\hskip1pt$\sim$}}
    \raise1pt\hbox{$<$}}}         
\def\gsim{\mathrel{\rlap{\lower4pt\hbox{\hskip1pt$\sim$}}
    \raise1pt\hbox{$>$}}}         

\begin{document}
\thispagestyle{empty}

\begin{center}

{Universit\'{e} de Li\`{e}ge \\
Interactions Fondamentales en Physique et Astrophysique \\
Facult\'{e} des Sciences}

\begin{figure} [ht]
\centering
\includegraphics[height=3cm]{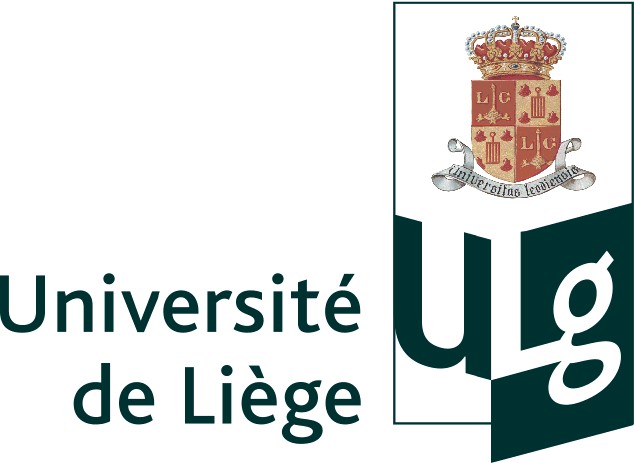}
\end{figure}

\vspace{2cm} 
{\Large \bf Symmetries of the Scalar Sector of Multi-Higgs-Doublet Models}\\

\vspace{3cm} {\large  Thesis presented in fulfilment of the requirements to \\
Obtain a PhD degree in Physics}

\vspace{3cm}
{presented by}

\vspace{1cm}
\large {\bf Venus Ebrahimi-Keus} \\  \normalsize 

\vspace{1cm}
September 2012

\end{center}

\clearpage{\pagestyle{empty}\cleardoublepage}


\thispagestyle{empty}
\begin{center}

{Universit\'{e} de Li\`{e}ge \\
Interactions Fondamentales en Physique et Astrophysique \\
Facult\'{e} des Sciences}

\begin{figure} [ht]
\centering
\includegraphics[height=3cm]{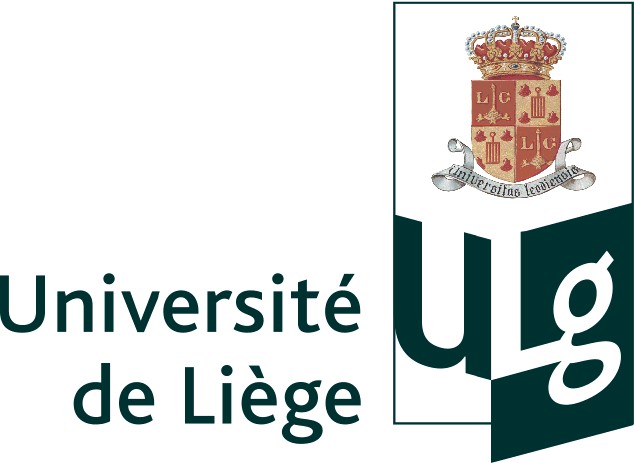}
\end{figure}

\vspace{2cm} 
{\Large \bf Symmetries of the Scalar Sector of Multi-Higgs-Doublet Models}\\

\vspace{2cm} {Thesis presented in fulfilment of the requirements to \\
Obtain a PhD degree in Physics}

\vspace{1cm}
{presented by}

\vspace{1cm}
\large {\bf Venus Ebrahimi-Keus} \\  \normalsize

\vspace{1cm}
Members of the jury: \\
1. Pierre Magain (Pr\'{e}sident) \\
2. Igor Ivanov (Promoteur)\\
3. Jean-Ren\'{e} Cudell\\
4. Diego Aristizabal\\
5. Maria Krawczyk \\
6. Pedro Ferreira \\

\vspace{1cm}
\noindent Day of defence: 03 September 2012 
\end{center}

\clearpage{\pagestyle{empty}\cleardoublepage}


\thispagestyle{empty}
\begin{center}
~\\
\vspace{7.5cm}
{\large to whom I owe the person I am today \\
\vspace{1cm}
my tireless parents }
\end{center}

\clearpage{\pagestyle{empty}\cleardoublepage}

\thispagestyle{empty}
\begin{center}
{\large\bf Abstract}
\end{center}

In pursuit of the origin of the masses of particles in the Standard Model, the Brout-Englert-Higgs Mechanism predicts the existence of a scalar boson, whose discovery has been the main goal of the Large Hadron Collider. However, introducing this scalar boson does not provide the answer to all inadequacies of the Standard Model. Many different Beyond Standard Model theories have been proposed in order to explain these anomalies. Multi-Higgs-doublet models are among these models. 

We restrict ourselves to the scalar sector of multi-Higgs-doublet-models. The huge number of free parameters in the scalar potential in these models, makes it impossible to study the most general case for any N. In order to reduce the number of free parameters, one could impose symmetries on the scalar potential. Therefore, it is important to explore which symmetries can be implemented in the scalar sector, how these symmetries are broken, how they could be encoded in the Yukawa sector and what the resulting properties of the fermions are. 

Classifying these symmetries in the scalar sector is the main focus of this thesis. We have found certain symmetries that are always broken in models with more than two doublets, which we name "frustrated symmetries". Examples of such symmetries are presented in 3HDM, and one particular symmetry, the octahedral symmetry, is studied further. This symmetry seems to be the largest realizable discrete symmetry that can be imposed on the scalar sector in 3HDM, and interestingly results in a 2HDM-like mass spectrum. 

In the attempt towards the classification of possible symmetries in the scalar sector of the NHDM, we find that these symmetry groups are either subgroups of the maximal torus, or certain finite Abelian groups which are not subgroups of maximal tori. For the subgroups of the maximal torus, we present an algorithmic strategy that gives the full list of possible realizable Abelian symmetries for any given $N$. We extend this strategy to include Abelian antiunitary symmetries (with generalized $CP$ transformations) in NHDM. 

We also show that multi-Higgs-doublet models can naturally accommodate scalar dark matter candidates protected by the group $\Z_p$, since these groups are realizable in NHDM. These models do not require any significant fine-tuning and can lead to a variety of forms of microscopic dynamics among the dark matter candidates.

The results of this thesis are published in \cite{frustrated-paper,Abelian2012,Zp-scalar-dm}.

\clearpage{\pagestyle{empty}\cleardoublepage}

\thispagestyle{empty}

\begin{center}
{\large\bf Acknowledgements}
\end{center}

First of all I would like to thank the members of the jury for taking their time to read this thesis. Secondly, I would like to express my appreciation to my supervisor, Igor who was abundantly helpful throughout this project and without whose knowledge and assistance this study would not have been possible.

Special thanks go to Jean-Ren\'{e} for his guidance and motivational chats throughout this journey, and to the members of IFPA group, especially Diego, Audrey, Athina and Alex for their invaluable assistance and moral support.

I would also like to convey thanks to the Fonds de la Recherche Scientifique (FNRS) for providing the financial means for this project.

My deepest gratitude goes to my family for their love and support. I am thankful to my brother for the encouraging conversations. I am indebted to my parents for always believing in me, and teaching me the true meaning of "where there is a will, there is a way".

Saving the best for last, I would like to thank my wonderful Martyn. Baby, I would not have been able to get here without you. You cannot imagine how grateful I am for your support in every step of the way. Thank you for being my rock through all the ups and downs of my PhD.

\clearpage{\pagestyle{empty}\cleardoublepage}


\tableofcontents


\chapter{Introduction}\label{chapter1}

This is the most exciting time to be a physicist; With the recent and upcoming results of the Large Hadron Collider (LHC) physicist are gaining better knowledge on the properties of the basic constituents of matter, and the fundamental interactions in nature. The current best-suited-with-experiment framework describing high energy phenomena is the Standard Model (SM) of electroweak and strong interactions. 

The first blocks towards the construction of the Standard Model were put down by Glashow in 1960, where he combined the $U(1)$-symmetric theory of electromagnetic interactions and the $SU(2)$-symmetric theory of weak interactions into an $SU(2) \times U(1)$ symmetric electroweak theory \cite{Glashow}. Later on, in 1967 Weinberg \cite{Weinberg} and Salam \cite{Salam} incorporated the Brout-Englert-Higgs mechanism \cite{Englert}\cite{Higgs}\cite{Guralnik} into Glashow's electroweak theory, and gave the model its modern form. And finally, in 1970's the $SU(3)$-symmetric strong interactions were described by quantum chromodynamics (QCD), shaping up the Standard Model. The symmetry group of the Standard Model is denoted by $SU(3)_C \times SU(2)_L \times U(1)_Y$.

Despite its great success in explaining all available data, the Standard Model has several inadequacies. It requires the existence of a massive scalar particle, the Brout-Englert-Higgs boson, whose discovery has been the main challenge of the LHC. Moreover, it does not give an explanation for the fermion mass spectrum, in particular the small neutrino masses (and their oscillations). The theory does not contain any viable dark matter particle that possesses all of the required properties deduced from observational cosmology. Also, this model does not include gravitation\footnote{This is partly because at the level of elementary particles the gravitational force is many orders of magnitude weaker and can be ignored. The quantum effects of gravity become important in describing particle interactions at the Planck scale. This refers to either a very large energy scale ($1.22 \times 10^{19}$ GeV) or a very tiny size scale ($1.616 \times 10^{-35}$ meters).}. This theory also falls short in explaining the strong $CP$ problem, and matter-antimatter asymmetry.

These deficiencies led physicists to explore Beyond the Standard Model theories (BSM). A rich spectrum of theories has been proposed, such as supersymmetric extensions, string theory, extra dimension theories and grand unified theories, etc. The question of which theory is the one chosen by the Nature, can only be answered through experiments.

In this Chapter we study the Standard Model Lagrangian and present the particle content of the SM in Section \ref{SM}. Section \ref{BEH} presents the Brout-Englert-Higgs mechanism through which the gauge bosons acquire mass. We end this Chapter with motivations for introducing more than one doublet, and give the example of a two-Higgs-doublet model.

\section{ The Standard Model of particle physics}\label{SM}

The particle content of the Standard Model is shown in Figure (\ref{Families}). The three families of fermions (with spin $1/2$), which obey the Fermi-Dirac statistics, form matter. Each family consists of two quarks and two leptons. These families are a perfect replica of each other, where the heavier families eventually decay into the first family particles. Each fermion is associated with an anti-particle with the same mass and decay width but opposite charges. 

The bosons (with spin $1$), which obey the Bose-Einstein statistics, mediate the fundamental interactions. The photon $\gamma$ is exchanged in electromagnetic interactions between charged particles. The gauge bosons $W^+, W^-$ and $Z^0$ mediate the weak interactions of radioactive decay. The eight gluons are responsible for strong interactions inside the atomic nuclei. The Feynman diagrams in Figure (\ref{Diagrams}) illustrate these three fundamental interactions.

There is also the two theoretical bosons, the graviton (with spin $2$) which is responsible for the gravitational interactions, and the Higgs boson (with spin $0$), which is responsible for the masses of the SM particles.


\begin{figure} [ht]
\centering
\includegraphics[height=11cm]{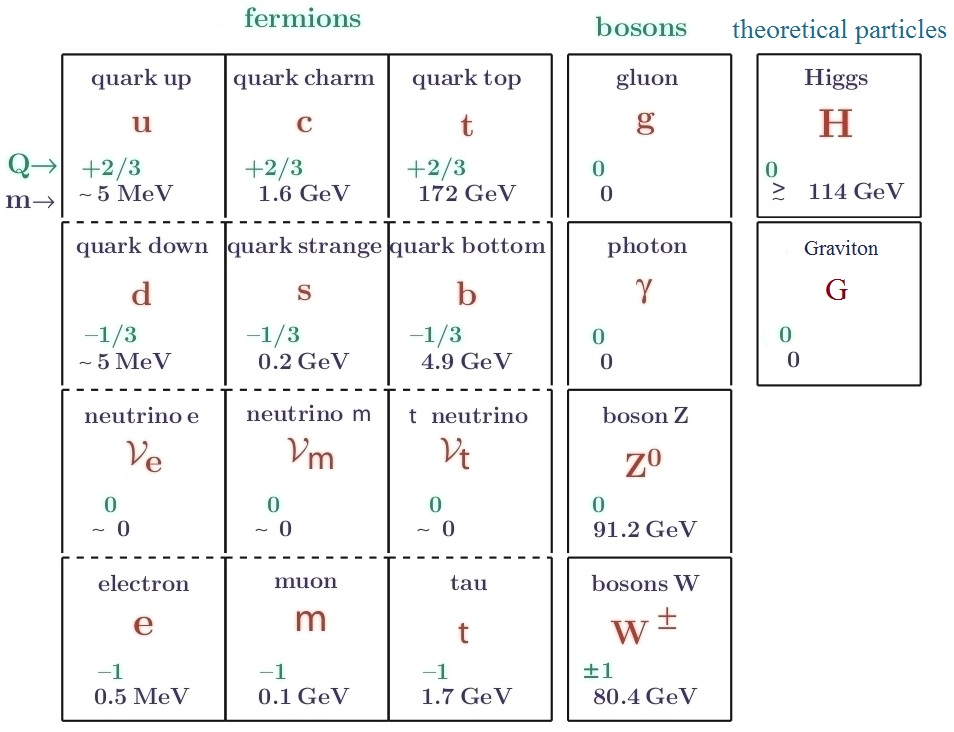}
\caption{The elementary particles of the Standard Model, and their characteristics \cite{Djouadi}.}
\label{Families}
\end{figure}

\begin{figure} [ht]
\centering
\includegraphics[height=4cm]{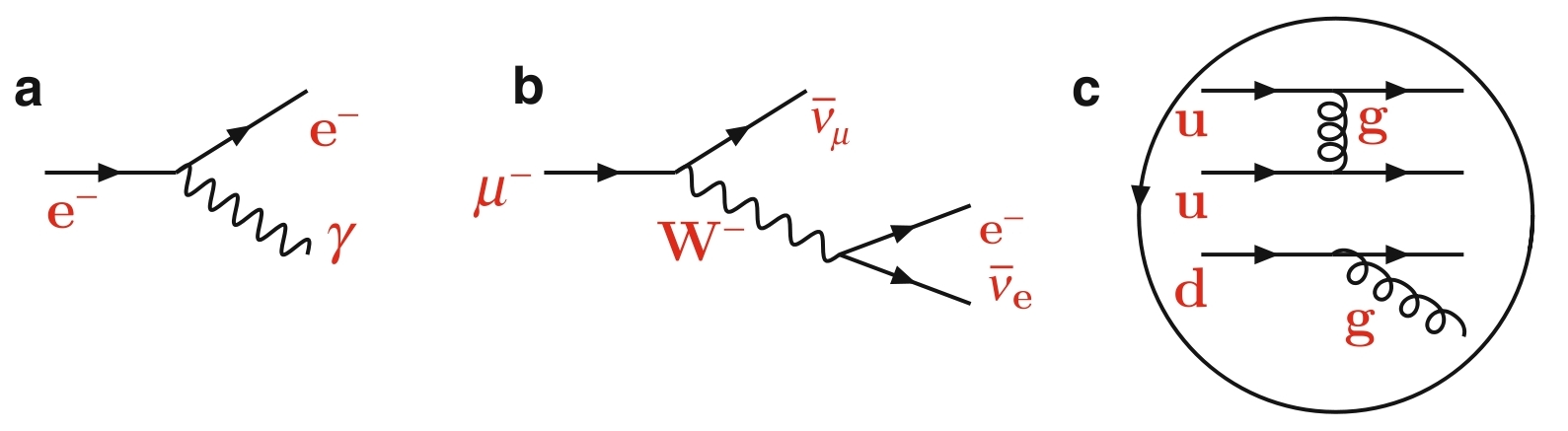}
\caption{The three fundamental interactions of the Standard Model \cite{Djouadi}; (a) The electromagnetic interaction with the exchange of a photon, (b) the weak interaction with the exchange of a $W$ boson, (c) the strong interaction with the exchange of gluons \cite{Djouadi}.}
\label{Diagrams}
\end{figure}

\subsection{The Standard Model Lagrangian}

The Standard Model is based on a very powerful principle, the gauge symmetry; the Lagrangian remains unchanged under gauge transformations of an internal symmetry group, $SU(3)_C \times SU(2)_L \times U(1)_Y$.

The Standard Model consists of the electroweak theory describing the electroweak interactions which is based on a $SU(2)_L \times U(1)_Y$ symmetry, and the theory of Quantum Chromodynamics (QCD) describing strong interactions with an $SU(3)_C$ symmetry. 

In this Chapter we only study the QED and electroweak Lagrangian.

\subsubsection{QED Lagrangian}
The Lagrangian of QED theory involves three parts;
\be
L_{QED} = L_{EM} + L_{Dirac} + L_{int} 
\ee
The electromagnetic field Lagrangian describes the propagation of free photons;
\bea
&& L_{EM} = -\frac{1}{4}F_{\mu\nu}F^{\mu\nu} \nonumber \\
&& \mbox{with} \quad F_{\mu\nu}= \partial_\nu A_\mu - \partial_\mu A_\nu
\eea
where $A_\mu$ is the four-vector electromagnetic potential.
The kinetic and mass terms in the Lagrangian are:
\be 
L_{Dirac}= i\overline{\psi}\gamma^\mu \partial_\mu \psi - m\overline{\psi} \psi 
\ee
And the Lagrangian for interaction of a fermion with charge $q$ and the photon field is described with:
\be 
L_{int}= -q \overline{\psi} \gamma^\mu \psi A_\mu 
\ee
This Lagrangian respects the $U(1)$ gauge symmetry;
\bea
\psi & \to & e^{-iq\alpha(x)}\psi \nonumber \\ 
A_\mu & \to & A_\mu+\frac{1}{q}\partial_\mu \alpha(x) 
\eea
which is equivalent to replacing the normal derivative $\partial_\mu$ by the covariant derivative $D_\mu$;
\be 
\partial_\mu \longrightarrow D_\mu \quad , \quad D_\mu = \partial_\mu + iqA_\mu
\ee
This replacement preserves the local symmetry. The electric charge $q$ acts as a generator for the local group transformation $U(x)=e^{-iq\alpha(x)}$. According to Noether's theorem, the continuous symmetry results in a conserved quantity that is the electric charge. The massless gauge boson associated with the local $U(1)$ symmetry is the photon with couples to particles with the electric charge. So, QED is an Abelian theory with the symmetry group $U(1)$ and one gauge field, the electromagnetic field. 

\subsubsection{The electroweak Lagrangian}

The electroweak Lagrangian is an $SU(2)_L\times U(1)_Y$ gauge invariant theory. There are three gauge bosons associated with the $SU(2)$; $W^1_\mu$, $W^2_\mu$, $W^3_\mu$. There is an additional gauge field associated with the $U(1)$ symmetry $B_\mu$. 
The electroweak Lagrangian has the following form; 
\be 
L = L_{gauge} + L_{Yukawa} + L_{f} + L{\phi} 
\ee
The kinetic part of the Lagrangian involving the leptons consists of:
\be
L_{lepton}= i\overline{\psi}_R \gamma^\mu (D_\mu)\psi_R + i\overline{\psi}_L \gamma^\mu (D_\mu)\psi_L 
\ee
And the gauge part is:
\be 
L_{gauge}= -\frac{1}{4}f_{\mu\nu}f^{\mu\nu} - \frac{1}{8} Tr(F_{\mu\nu}F^{\mu\nu})
\ee
where the strength tensors of $W$ and $B$ fields are;
\bea
& f_{\mu\nu} = \partial_\mu B_\nu - \partial_\nu B_\mu & \nonumber \\
& F^l_{\mu\nu} = \partial_\mu W^l_\nu - \partial_\nu W^l_\mu - g_W \epsilon_{lmn} W^m_\mu W^n_\nu & 
\eea
and the covariant derivative preserves the local gauge invariance;
\be 
D_\mu = \partial_\mu + \frac{ig_BY}{2}B_\mu + ig_W \frac{\overrightarrow{\sigma}}{2} \cdot \overrightarrow{W}_\mu 
\ee
Here $\gamma^\mu$ are the gamma matrices (Dirac matrices), and $\sigma_i$ are the Pauli matrices
\be 
\sigma_1 = \left( \begin{array}{cc} 0 & 1 \\ 1 & 0 \\ \end{array} \right) , \quad \sigma_2 = \left( \begin{array}{cc} 0 & -i \\ i & 0 \\ \end{array} \right) , \quad \sigma_3 = \left( \begin{array}{cc} 1 & 0 \\ 0 & -1 \\ \end{array} \right)
\ee
$g_W$ and $g_B$ are the $SU(2)$ and $U(1)$ gauge couplings respectively and $\epsilon_{lmn}$ is the totally antisymmetric symbol.

Note that the electron and its neutrino field are combined together into a doublet for the left-handed components, and a singlet for the right-handed component;
\be 
\psi_L = \stolbik{\nu_{e_L}}{e_L} , \quad \psi_R = e_R
\ee
since the mass of the neutrino is $0$ in SM and there is only a left-handed component of the neutrino;
\be 
\left( \frac{1-\gamma_5}{2} \right) \nu_e = \nu_e
\ee
The left-handed and right-handed electron fields, $e_L$ and $e_R$, are related to the electron in the following way;
\be 
e_L = \left( \frac{1-\gamma_5}{e} \right) e , \quad e_R = \left( \frac{1+\gamma_5}{e} \right) e
\ee

The three types of charges in the electroweak theory are the electric charge $Q$, the weak isospin $I$, and the weak hypercharge $Y$. These operators are related by the Gell-Mann$-$Nishijima relation;
\be 
Q = I^3 + \frac{Y}{2}
\ee
where $I^3$ is the third component of the weak isospin. The neutrino is assigned an isospin of $I^3_\nu = +1/2$. A left-handed electron has $I^3_e = -1/2$, and a right-handed electron field has $I_R =0$. For a left-handed spinor $Y=-1$, while for a right-handed spinor $Y=-2$. So, the overall charges are:
\bea
\mbox{Neutrino} \quad &: \quad I^3_L = +\frac{1}{2} \quad &Y_L=-1 \quad Q_L =0  \nonumber  \\
\mbox{Left-handed electron} \quad &: \quad I^3_L = -\frac{1}{2} \quad &Y_L=-1 \quad Q_L =-1   \nonumber \\
\mbox{Right-handed electron} \quad &: \quad I^3_L = 0 \quad &Y_L=-2 \quad Q_L =-1 
\eea
The three gauge fields $W^1_\mu$, $W^2_\mu$, $W^3_\mu$ correspond to the weak isospin charge $I$ and the gauge field $B_\mu$ corresponds to the weak hypercharge $Y$. A particle with a given type of charge, interacts with the field associated with that charge, and the value of the charge determines the strength of that interaction.

The left-handed fermionic fields with weak isospin $I =1/2$ ($I^3 = \pm 1/2$) and weak hypercharge $Y(Q_L)=1/3$ can be written as:
\be 
Q_L^1= \stolbik{u}{d}_L , \quad Q_L^2= \stolbik{c}{s}_L , \quad Q_L^3= \stolbik{t}{b}_L 
\ee
For the fermions of right helicity, the weak isospin is zero ($I_3=0$) which corresponds to singlets of fermions:
\be 
U_R^{1,2,3}= u_R, c_R, t_R \quad , \quad D_R^{1,2,3}= d_R, s_R, b_R
\ee



\section{Spontaneous symmetry breaking}\label{BEH}

The local gauge symmetry causes the $W$ and $Z$ boson masses in the electroweak Lagrangian to be zero, which is far from the experimentally observed values, $M_W=80$ $GeV$ and $M_Z=91$ $GeV$. Introducing a mass term, by hand, to the Lagrangian is not the solution, since it breaks the gauge symmetry. The mechanism that insures massive $W$ and $Z$ bosons and massless photons while preserving the gauge invariance, is called the Brout-Englert-Higgs mechanism, and is presented in this Section.

This mechanism also explains the masses of the quarks and leptons through Yukawa interactions with the scalar fields and their conjugates.

\subsection{Abelian Higgs model}

Let's start with an Abelian Higgs model, where we introduce spontaneous symmetry breaking in a model with a $U(1)$ symmetry, say electromagnetism theory.

Recall that the $U(1)$-symmetric electromagnetism Lagrangian;
\bea 
&& L_{EM}=-\frac{1}{4}F_{\mu\nu}F^{\mu\nu} \nonumber  \\
&& \mbox{with} \quad F_{\mu\nu}= \partial_\nu A_\mu - \partial_\mu A_\nu
\eea  
is invariant under the transformations $A_\mu  \longrightarrow A_\mu + \frac{1}{q}\partial_\mu \alpha (x)$ and $\phi\longrightarrow \phi e^{-iq \alpha (x)}$. Let's add a mass term to the Lagrangian:
\be 
L=-\frac{1}{4}F_{\mu\nu}F^{\mu\nu} + \frac{1}{2}m^2 A_\mu A^\mu
\ee 
The mass term clearly violates the gauge symmetry. However, we can extend the model by adding a single complex scalar field;
\be 
\phi= \frac{1}{\sqrt{2}}(\phi_1 +i\phi_2)
\ee
with real fields $\phi_1$ and $\phi_2$. The Lagrangian describing this new field's kinetic and potential term has the form 
\bea
&&L=-\frac{1}{4}F_{\mu\nu}F^{\mu\nu} +|D_\mu \phi|^2 -V(\phi) \nonumber \\
&& \mbox{with} \quad V(\phi)= {\mu}^2|\phi|^2 + \lambda (|\phi|^2)^2 \nonumber \\
&& \mbox{and} \quad D_\mu = \partial_\mu + iqA_\mu  
\eea
where $V(\phi)$ is the most general $U(1)$-symmetric renormalizable potential. 
For $\lambda >0 $, this theory has two different forms:

\begin{itemize}
\item  $\mu^2 > 0$ \\
In this case the potential has a shape as shown in Figure (\ref{Pos-Pot}), and it preserves the symmetry of the Lagrangian. The vacuum appears at $\phi =0$. This Lagrangian describes a massless photon and a charged field $\phi$, with mass $\mu$.

\begin{figure} [ht]
\centering
\includegraphics[height=3.5cm]{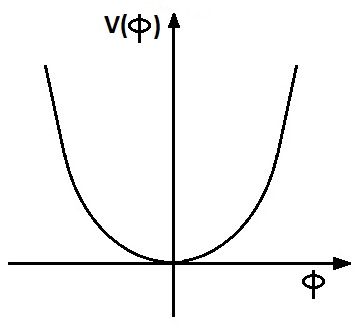}
\caption{The scalar potential with $\mu^2 > 0$}
\label{Pos-Pot}
\end{figure}

\item  $\mu^2 < 0$ \\
The more interesting potential in this case has the shape shown in Figure (\ref{Neg-Pot}). The vacuum appears at $ \langle \phi \rangle = v = \pm \sqrt{-\frac{\mu^2}{2\lambda}}$, which breaks the symmetry. It is conventional to choose the vacuum to lie along the real axis of $\phi$. Note that the symmetry in the potential is spontaneously broken with acquiring the vacuum expectation value (v.e.v).

\begin{figure} [ht]
\centering
\includegraphics[height=3.5cm]{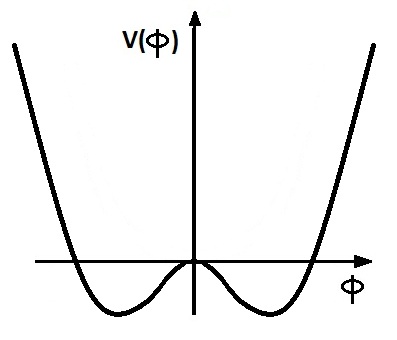}
\caption{The scalar potential with $\mu^2 < 0$}
\label{Neg-Pot}
\end{figure}

\end{itemize}

\subsubsection{Goldstone bosons}

It is common to parametrize the field $\phi$ in terms of two real fields  $\chi$ and $H$,
\be 
\phi = \frac{1}{\sqrt{2}} e^{i\frac{\chi}{v}} (v + H)
\ee
The Lagrangian in terms of $\chi$ and $H$, will have the form;
\bea
L &=& -\frac{1}{4}F_{\mu\nu}F^{\mu\nu} +\frac{1}{2}\partial_\mu \chi \partial^\mu \chi + \frac{1}{2} \partial_\mu H \partial^\mu H  \nonumber \\
 && + \mu^2H^2  + \frac{1}{2}q^2v^2A_\mu A^\mu  \nonumber \\
&& - qvA_\mu \partial^\mu \chi + \quad \mbox{higher order terms} 
\eea
This Lagrangian describes a photon with mass $M_A = qv$, a scalar field with mass $ \sqrt{-2\mu^2} > 0$, and a scalar field with mass zero. The confusing mixed term $- qvA_\mu \partial^\mu \chi $ can be removed by making a gauge transformation:
\be 
A_\mu \longrightarrow A_\mu - \frac{1}{qv}\partial_\mu \chi
\ee
This transformation removes the field $\chi$ from the Lagrangian. The field $\chi$ is called a Nambu-Goldstone boson \cite{Nambu,Goldstone}, which has been ''eaten" to give mass to the photon. The field $H$ is called a Higgs boson field. 

Let's count the number of degrees of freedom (d.o.f). Before the spontaneous symmetry breaking we had a massless photon (2 d.o.f) and a complex scalar field (2 d.o.f), making the total of 4 degrees of freedom. After the spontaneous symmetry breaking we have a massive photon (3 d.o.f) and a real scalar $H$ (1 d.o.f), making the same total number of degrees of freedom.

\subsection{Non-Abelian Higgs model (Brout-Englert-Higgs mechanism)}

Here we present the spontaneous symmetry breaking method which forces the gauge bosons to acquire mass in the $SU(2) \times U(1)$-symmetric electroweak theory. This method is an extension of the Abelian Higgs model we just studied.

This mechanism postulates the existence of a doublet of complex scalar fields: 
\be 
\Phi = \frac{1}{\sqrt{2}}\stolbik{\phi_1 + i\phi_2}{\phi_3 + i\phi_4}
\ee
with real fields $\phi_1$, $\phi_2$, $\phi_3$ and $\phi_4$.
The most general $SU(2)$-symmetric renormalizable potential would be: 
\be 
V(\Phi) = \mu^2 |\Phi^\dagger \Phi| + \lambda (|\Phi^\dagger \Phi|)^2  \label{potential-one-doublet}
\ee
with $\lambda >0$ (to ensure that the potential is bounded from below). 

Similar to the Abelian case, the vacuum acquires a v.e.v for $\mu^2 <0$; 
\be 
\langle\Phi \rangle = \frac{1}{\sqrt{2}} \stolbik{0}{v}
\ee
shown in Figure (\ref{Sym-Break}).
\begin{figure} [ht]
\centering
\includegraphics[height=4cm]{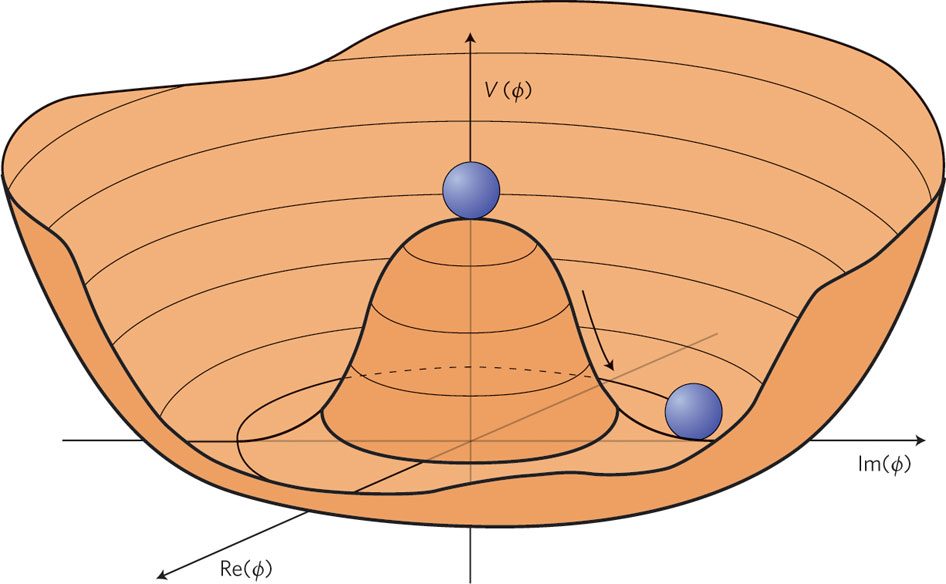}
\caption{The scalar potential with a v.e.v. leading to spontaneous symmetry breaking \cite{Gaume-Ellis}.}
\label{Sym-Break}
\end{figure}
$\Phi$ has the hypercharge $Y_\Phi =1/2$, and electromagnetic charge $Q= I_3 + Y/2 = 0$ for the lower component, which means that the electromagnetic $U(1)$ symmetry is preserved. So, this mechanism yields the desired symmetry breaking $SU(2)_L \times U(1)_Y \longrightarrow U(1)_{EM} $.

The scalar field $\Phi$ contributes to the Lagrangian with:
\be 
L =(D^\mu \Phi)^\dagger (D_\mu \Phi)-V(\Phi)
\ee
with the covariant derivative:
\be 
D_\mu=\partial_\mu +i{g_BY\over 2} B_\mu +i {g_W\over 2}\overrightarrow{\sigma} \cdot \overrightarrow{W_\mu} 
\ee
We develop the theory around this minimum;
\be 
\Phi=  \stolbik{\xi_1+i\xi_2}{\frac{1}{\sqrt{2}}(v+H+i\xi_3)}
\ee
We substitute this $\Phi$ in the Lagrangian, and study the second order terms to calculate the masses. The three fields $\xi_i$ are three Goldstone bosons which give mass to the three weak gauge fields. The spontaneous symmetry breaking also leads to a massive boson from the field $H$, the Higgs boson. 

These mass terms result from the following terms in the Lagrangian:
\bea 
\Delta L&=& \frac{1}{2}(0 \quad v)(i{g_BY\over 2} B_\mu +i {g_W\over 2}\overrightarrow{\sigma} \cdot \overrightarrow{W_\mu})(i{g_BY\over 2} B_\mu +i {g_W\over 2}\overrightarrow{\sigma} \cdot \overrightarrow{W_\mu})\stolbik{0}{v} \nonumber \\
&=& \frac{v^2}{8}\left[g^2_W(W^1_\mu)^2 + (W^2_\mu)^2 + (-g_WW^3_\mu + g_BB_\mu)^2  \right] \nonumber
\eea
Three massive vector bosons appear;
\be
W^{\pm}_\mu = \frac{1}{\sqrt{2}}(W^1_\mu \mp i W^2_\mu) \quad, \quad 
Z_{\mu} = \frac{-g_B B_\mu + g_W W^3_\mu}{\sqrt{g^2_W + g^2_B}} 
\ee
The fourth vector field remains massless;
\be 
A_{\mu} = \frac{g_W B_\mu + g_B W^3_\mu}{\sqrt{g_W^2 + g^2_B}}
\ee
with the weak mixing angle $\theta_W$; 
\be
\sin\theta_W =\frac{g_B}{\sqrt{g^2_W +g^2_B}}  \quad  , \quad \cos\theta_W = \frac{g_W}{\sqrt{g^2_W +g^2_B}}
\ee
The value for masses appears to be:
\be 
M_W^2 = {1\over 4} g^2_W v^2, \quad M_Z^2 = {1\over 4} (g^2_W + g^2_B)v^2, \quad M_\gamma= 0
\ee
Therefore, introducing the Brout-Englert-Higgs boson has managed to spontaneously break the electroweak symmetry and give mass to the gauge bosons while leaving the photon massless.

Let's count the degrees of freedom in this case. We began with a complex scalar $SU(2)_L$ doublet $\Phi$ with four degrees of freedom, a massless $SU(2)_L$ gauge field $W_i$ with six degrees of freedom, and a massless $U(1)_Y$ gauge field $B$ with 2 degrees of freedom, making the total of 12. After the spontaneous symmetry breaking we have a physical real scalar field $H$ (1 degree of freedom), massive $W's$ and $Z$ fields (9 degrees of freedom), and a massless photon (2 degrees of freedom), adding up to the same 12 degrees of freedom.

\section{Motivation for introducing more than one doublet}

There is no fundamental reason to assume that the scalar sector is minimal and contains only one Higgs doublet. One of the simplest extension of the BEH Mechanism is the two-Higgs-doublet model (2HDM) \cite{Lee}, with one extra Higgs doublet with the same quantum numbers as the first one. There are many motivations for introducing two-Higgs doublet models, and they have been studied extensively.  

One of the motivations for introducing the second doublet is stemmed from the hierarchy of the Yukawa couplings in the third generation quarks, the top and the bottom quarks. In the Standard Model the masses of both these quarks comes from their couplings to the one Higgs doublet, which does not explain the ratio $m_t / m_b \approx 35$. However, this problem could be fixed if they received their masses from the two different doublets, providing the free parameters of the theory acquired appropriate values.

One other reason to extend the scalar sector is to implement $CP$ violation \cite{Lee} in the model. In two-Higgs-doublet models the $CP$ symmetry can be violated explicitly by choosing complex coefficients.

Another motivation for introducing the second doublet is its application in supersymmetry \cite{Haber-Kane}. The scalars appearing in supersymmetric theories are chiral multiplets, and their conjugates are of opposite chirality, therefore they cannot couple together. As a result, two Higgs doublets are needed to simultaneously give mass to the quarks with different charges, and therefore explaining the mass hierarchy in the fermion sector.

Moreover, two-Higgs-doublet models are a possible source of Flavour Changing Neutral Currents (FCNC). These processes are strongly suppressed by experimental data, even though they do not seem to violate any fundamental law of nature. The Standard Model is compatible with experimental constraints on FCNC so far \cite{Sher,Cheng}, with the remarkable exception of neutrino oscillations \cite{Fukuda}. These processes are absent in the lepton sector, and in the quark sector they are prohibited. The simplest model that could explain FCNC processes is the two-Higgs-doublet model. The Yukawa interactions arising from the second doublet lead to tree level FCNC's. 

Generating baryon asymmetry in the Universe of sufficient size is another motivation for introducing models with more than one doublet \cite{Trodden}, since they have a flexible scalar mass spectrum and they provide additional sources for $CP$ violation.

Models with more than one doublet, such as two-Higgs-doublet inert model \cite{Ginzburg-Inert}, could also accommodate dark matter candidates. In such models a $Z_2$ symmetry which remains unbroken, requires one of the doublets to have zero vacuum expectation value, and therefore not to couple to quarks and leptons. As a result this doublet, the inert Higgs doublet, could contain a dark matter candidate.

\section{Two-Higgs-doublet model potential}\label{section-2hdm}

The Two-Higgs-Doublet Model (2HDM) is one of the simplest BSM extensions of the Higgs mechanism. In this model, one introduces two complex Higgs doublets $ \phi_1$ and $\phi_2 $ with:
\be 
 \phi_{a} = \stolbik{\phi_{a}^+}{\phi_{a}^0}, \quad a=1,2
\ee
The most general Higgs potential of 2HDM which is gauge invariant and renormalizable, is usually divided into quadratic and quartic parts:
\bea
V &=& V_2 + V_4  \\
V_2&=&-{1\over 2}\left[m_{11}^2(\phi_1^\dagger\phi_1) + m_{22}^2(\phi_2^\dagger\phi_2) + m_{12}^2 (\phi_1^\dagger\phi_2) + m_{12}^{2\ *} (\phi_2^\dagger\phi_1)\right] \nonumber\\
V_4&=&\fr{\lambda_1}{2}(\phi_1^\dagger\phi_1)^2 +\fr{\lambda_2}{2}(\phi_2^\dagger\phi_2)^2 +\lambda_3(\phi_1^\dagger\phi_1) (\phi_2^\dagger\phi_2) +\lambda_4(\phi_1^\dagger\phi_2) (\phi_2^\dagger\phi_1) \nonumber \\
&+&\fr{1}{2}\left[\lambda_5(\phi_1^\dagger\phi_2)^2+ \lambda_5^*(\phi_2^\dagger\phi_1)^2\right]
+\left\{\left[\lambda_6(\phi_1^\dagger\phi_1)+\lambda_7 (\phi_2^\dagger\phi_2)\right](\phi_1^\dagger\phi_2) +{\rm
h.c.}\right\} \nonumber
\eea
This potential has $14$ free parameters [compare to only two real parameters in the potential with one doublet (\ref{potential-one-doublet})]; the real parameters $m^2_{11}, m^2_{22}, \lambda_1, \lambda_2, \lambda_3, \lambda_4$, and the complex parameters $m^2_{12}, \lambda_5, \lambda_6, \lambda_7$. 


\subsection{A simple case of 2HDM}

Here we briefly study a simple case of 2HDM, with $\lambda_1 = \lambda_2 = \lambda$, $\lambda_3 = \lambda'$ and $ \lambda_4 = \lambda_5 = \lambda_6 = \lambda_7 = m^2_{12}=0$. Therefore the potential simplifies to;
\bea
V &=& -\frac{1}{2}\left[m^2_{11}(\phi^\dagger_1\phi_1) +m^2_{22}(\phi^\dagger_2\phi_2)\right]  \\
&& +\frac{\lambda}{2} (\phi^\dagger_1\phi_1)^2 + \frac{\lambda}{2}(\phi^\dagger_2\phi_2)^2 + \frac{\lambda'}{2} (\phi^\dagger_1\phi_1)(\phi^\dagger_2\phi_2) \nonumber 
\eea
In order to have a stable vacuum the potential has to be bounded from below. It is enough to require the quartic part to be positive, therefore $\lambda >0$ and $\lambda + \lambda' >0$. We require the doublets to acquire an expectation value;
\be 
\lr{\phi_{1}} = \stolbik{0}{v_1} \quad , \quad  \lr{\phi_{2}} = \stolbik{0}{v_2}
\ee 
Requiring the potential to have a minimum;
\be 
\left(\frac{\partial V}{\partial \phi_1} \right)=0 \quad , \quad \left(\frac{\partial V}{\partial \phi_2} \right)=0
\label{vanish-linear}
\ee
imposes the following conditions;
\be 
v_1(-m^2_{11} + 2\lambda v^2_1 + \lambda' v^2_2)=0 , \quad v_2(-m^2_{22} + 2\lambda v^2_2 + \lambda' v^2_1)=0
\ee
Therefore, we have four possibilities;

\begin{itemize}
\item
The phase $A$ with $v_1=0$ and $v_2=0$\\
There is no symmetry breaking in this case, since the potential has one global minimum at the origin, with the conditions $\lambda >0$ and $\lambda + \lambda' >0$.
To get the mass eigenstates, we develop the potential around the minimum;
\be 
\phi_{1} = \stolbik{w^+_1}{\frac{h_1+i\eta_1}{\sqrt{2}}}, \quad  \phi_{2} = \stolbik{w^+_2}{\frac{h_2+i\eta_2}{\sqrt{2}}}
\ee
which results in a mass spectrum of the form:
\bea
&&m^2(w^{\pm}_a) = -\frac{m^2_{11}}{2} , \quad -\frac{m^2_{22}}{2} \nonumber \\
&&m^2(h_a, \eta_a) = -\frac{m^2_{11}}{2} , \quad -\frac{m^2_{22}}{2} 
\eea
with the extra condition $m^2_{11}, m^2_{22}<0$.

\item
The phase $B_1$ with $v_1\neq0$ and $v_2=0$\\
In this case the potential has two global minima in $\phi_1$ direction in the $(\phi_1,\phi_2)$-plane, at $\lr{\phi_1}=\pm \frac{v_1}{2}$. To study the potential, we excite the minimum;
\be 
\phi_{1} = \stolbik{w^+_1}{\frac{v_1+h_1+i\eta_1}{\sqrt{2}}}, \quad  \phi_{2} = \stolbik{w^+_2}{\frac{h_2+i\eta_2}{\sqrt{2}}}
\ee
To satisfy the condition (\ref{vanish-linear}) the linear terms must vanish, which results in;
\be 
v^2_1=\frac{m^2_{11}}{\lambda}
\ee
Therefore the mass spectrum has the form:
\bea
&&m^2(w^{\pm}_a) = 0, \quad \frac{-2\lambda m^2_{22}+\lambda' m^2_{11}}{4\lambda}   \\
&&m^2(h_a, \eta_a) = m^2_{11} , \quad \frac{-2\lambda m^2_{22}+\lambda' m^2_{11}}{4\lambda}\quad \mbox{(double degenerate)} \nonumber
\eea
Positivity of the mass eigenstates requires;
\be 
-m^2_{22}+\frac{\lambda'}{2\lambda}m^2_{11}>0
\ee
This condition is satisfied in three different cases; $\lambda>\lambda'>0$, $\lambda'>\lambda>0$ and $0>\lambda'>-\lambda$.

\item
The phase $B_2$ with $v_1=0$ and $v_2\neq0$\\
This case is very similar to $B_1$ case. Here the potential has two global minima in $\phi_2$ direction in the $(\phi_1,\phi_2)$-plane, at $\lr{\phi_2} =\pm \frac{v_2}{2}$.
Again to study the potential we excite the minimum;
\be 
\phi_{1} = \stolbik{w^+_1}{\frac{h_1+i\eta_1}{\sqrt{2}}}, \quad  \phi_{2} = \stolbik{w^+_2}{\frac{v_2+h_2+i\eta_2}{\sqrt{2}}}
\ee
From the vanishing linear terms in the potential we get;
\be 
v^2_2=\frac{m^2_{22}}{\lambda}
\ee
which results in a mass spectrum of the form:
\bea
m^2(w^{\pm}_a) &=& \frac{-2\lambda m^2_{11}+\lambda' m^2_{22}}{4\lambda}   \\
m^2(h_a, \eta_a) &=& m^2_{22} , \quad \frac{-2\lambda m^2_{11}+\lambda' m^2_{22}}{4\lambda} \quad \mbox{(double degenerate)} \nonumber
\eea
With the positivity of the mass eigenstates requiring;
\be 
-m^2_{11}+\frac{\lambda'}{\lambda}m^2_{22}>0
\ee
Similar to $B_1$, here we have three cases depending on the values of $\lambda$ and $\lambda'$.

\item
The phase $C$ with $v_1\neq0$ and $v_2\neq0$\\
The potential acquires four minima in this case, two in the $\phi_1$ direction at $\phi_1=\pm v_1/2$, and two in the $\phi_2$ direction at $\lr{\phi_2} =\pm v_2/2$. We excite the minimum;
\be 
\phi_{1} = \stolbik{w^+_1}{\frac{v_1+h_1+i\eta_1}{\sqrt{2}}}, \quad  \phi_{2} = \stolbik{w^+_2}{\frac{v_2+h_2+i\eta_2}{\sqrt{2}}}
\ee
From the vanishing linear terms in the potential we get;
\be 
v^2_1=\frac{4\lambda m^2_{11} -2\lambda' m^2_{22}}{4\lambda^2 -\lambda'^2}  \quad , \quad   v^2_2=\frac{4\lambda m^2_{22} -2\lambda' m^2_{11}}{4\lambda^2 -\lambda'^2}
\ee
which results in a 2 massless charged bosons, 2 massless neutral bosons and 2 massive neutral bosons which are the eigenvalues of the following mass matrix:
\be 
\left( \begin{array}{cc}
\frac{-m^2_{11}}{4}+\frac{3}{4}\lambda v^2_1 + \frac{\lambda'}{8}v^2_2    &     \frac{\lambda'}{4}v_1v_2  \\ 
\frac{\lambda'}{4}v_1v_2     &   \frac{-m^2_{22}}{4}+\frac{3}{4}\lambda v^2_2 + \frac{\lambda'}{8}v^2_1 
\end{array}\right)
\ee

\end{itemize}


\chapter{N-Higgs-Doublet Models}\label{chapter2}

The electroweak symmetry breaking (EWSB) in the Standard Model is based on the Brout-Englert-Higgs mechanism (BEH). Many different variants of the BEH mechanism have been proposed so far, but it is not known which one is realized in the Nature yet. In the meantime, it is important to be aware of all essential possibilities that are realized in a chosen model.

The non-minimal Higgs models have many free parameters, and it is important to analyse the model in its most general form, allowing for all possible degrees of freedom. This general analysis shows which phenomenological consequences are sensitive to the values of the parameters, which symmetries can arise in the model and how they are broken. After understanding the general structure of a model, one could restrict the model according to the existing experimental constraints.

Unfortunately, such an exhaustive analysis is not possible in many non-minimal Higgs models. Even in the case of 2HDM \cite{Lee,Hunter,CPNSh}, the scalar potential cannot be minimized explicitly in the general case. Recently, a number of tools have been developed helping to understand the properties of the general 2HDM. These methods were based on the idea of the reparametrization symmetry, or basis-invariance, of the model: a unitary transformation between the Higgs doublets changes the parameters of the model, but leads to the same physical properties of the observable particles.

This idea can be implemented via a tensorial formalism at the level of scalar fields \cite{CP,Ginzburg:2004vp,haber,haber2,oneil} or through a geometric constructions in the space of gauge-invariant bilinears \cite{sartori,nagel,heidelberg,nishi2006,ivanov0}. In the latter case the formalism was extended to include non-unitary reparametrization transformations \cite{ivanov1,ivanov2,nishi2008}, which revealed interesting geometric properties of the 2HDM in the orbit space.

It is a natural idea to extend these successful techniques to $N$ doublets. The general potential in multi-Higgs-doublet models is much richer than 2HDM, both at the level of scalar sector and Yukawa interactions \cite{weinberg1976,Branco1980,DeshpandeHe1994}.
Some properties of the general NHDM were analysed in 
\cite{heidelberg,barroso2006,nishi:nhdm,ferreira-nhdm,zarrinkamar}, with a special emphasis on $CP$-violation, \cite{nishi2006,lavoura1994}. However, a method to systematically explore all the possibilities offered with $N$ doublets was still missing.

Generalization from 2HDM to NHDM is straightforward in the tensorial formalism, however it is very difficult to translate tensorial invariants into physical observables. On the other hand, the geometric approach in the space of bilinears offers a more appealing treatment of the scalar potential, but the shape of the NHDM orbit space is rather complicated and has not been fully characterized so far.

In this Chapter we study the algebraic and geometric properties of the NHDM orbit space. First we study the three distinct, but interrelated approaches to description of gauge orbits in the space of scalar fields in Section \ref{orbit-nhdm-section}. Then we show a geometrical description of the orbit space in the general case of NHDM. In Section \ref{potential-nhdm-section} we present the general scalar potential in NHDM, and end the Chapter with a detailed study of the orbit space of three-Higgs-doublet models, which is going to be used in the next Chapters.

\section{Orbit space of NHDM}\label{orbit-nhdm-section}

The scalar potential of the NHDM is constructed from gauge-invariant bilinear combinations of the Higgs fields. The space of these combinations can be described in three algebraically different but closely related ways: via representative Higgs doublets, the so-called $K$-matrix, and a vector in the adjoint space. Here we describe and compare these three ways.

\subsection{Field space and gauge orbits}

The scalar part of the general NHDM has $N$ complex Higgs doublets with electroweak isospin $Y=1/2$:
\be
\phi_a = \doublet{\phi_a^+}{\phi_a^0}\,,\quad a=1,\dots , N 
\ee
The dimension of the space of scalar fields is $4N$, with each doublet having 4 degrees of freedom (2 neutral and 2 charged). Since the scalar Lagrangian is electroweak-symmetric, we can perform any simultaneous intradoublet $SU(2)\times U(1)$ transformation inside all doublets without changing the Lagrangian.

If we take a generic point in the Higgs space and apply all possible $SU(2)\times U(1)$ transformations, we will get a four-dimensional manifold called the (gauge) orbit. Thus, the entire $4N$-dimensional space of Higgs fields is naturally divided into non-intersecting orbits. The resulting set of orbits is a $(4N-4)$-dimensional manifold called the orbit space.

When minimizing the Higgs potential, we look for vacuum expectation values of the Higgs fields $\lr{\phi_a}$.
Then, we can characterize each orbit by a specific representative point in the Higgs space: 
\be
\label{Phi:generic}
\lr{\phi_1 } = \doublet{0}{v_1}\,,\quad 
\lr{\phi_2 } = \doublet{u_2}{v_2e^{i\xi_2}}\,,\quad 
\lr{\phi_a } = \doublet{u_a e^{i\eta_a}}{v_a e^{i\xi_a}}\,,\quad a > 2 
\ee
This point (and therefore, the entire orbit) is characterized by $4N-4$ real parameters: $N$ values of $v_a$,
$N-1$ values of $u_a$, $a>1$, $N-1$ phases $\xi_a$, $a>1$, and $N-2$ phases $\eta_a$, $a>2$. 

A point with at least one $u_a \not = 0$ corresponds to the charge-breaking vacuum,
in which the electroweak symmetry is broken completely and the photon acquires mass. For a neutral vacuum, we must set all $u_a = 0$. Thus, the representative point of a generic neutral orbit is
\be
\label{Phi:neutral}
\lr{\phi_1 } = \doublet{0}{v_1}\,,\quad 
\lr{\phi_a } = \doublet{0}{v_a e^{i\xi_a}}\,,\quad a > 1
\ee
This point is characterized by $N$ parameters $v_a$ and $N-1$ phases $\xi_a$, making the dimensionality of the neutral orbit space equal to $2N-1$, which is $2N-3$ units less than the dimensionality of the entire orbit space.

In addition to the electroweak $SU(2)\times U(1)$ transformations, which do not affect the parameters of the potential, one can consider the $SU(N)_H$ group of transformations that mix the doublets without affecting their intradoublet structure (index $H$ stands for ''horizontal" in contrast to the ''vertical" transformations that mix the intradoublet structure). These transformations transform a given Higgs potential to another Higgs potential with different coefficients. Such a transformation is called a reparametrization transformation, (or a horizontal space transformation, or a Higgs-basis change).

These transformations reparametrize the potential, but leave the physical observables invariant \cite{CP, Ginzburg:2004vp,haber}. The same applies to anti-unitary transformations. Therefore, reparametrization transformations consist of all unitary and anti-unitary transformations $U_{ab}$ where
\be 
\phi_a \to U_{ab}\phi_b \quad , \quad \phi_a \to U_{ab}\phi^*_b 
\ee
The antiunitary transformations are also known as generalized $CP$ transformations \cite{ecker-grimus,HaberSilva}.

Reparametrization transformations link different gauge orbits, so if we pick up a generic point in the gauge orbit space and apply all possible reparametrization transformations, we reach many other points in the orbit space. This divides the orbit space into non-intersecting $SU(N)$-orbits.

Using reparametrization transformations, any neutral orbit parametrized by $v_a$, $\xi_a$ can be brought to a "canonical
form" (also known as Higgs basis)
\be
\label{Phi:neutral:canonical}
\lr{\phi_1 } = \doublet{0}{v}\,,\quad 
\lr{\phi_a } = \doublet{0}{0}\ \mbox{for} \quad a > 1\,, \quad v^2 \equiv \sum_a v_a^2 
\ee
Depending on which doublet gets the non-zero value, there will be $N$ equivalent canonical forms. 

Equivalently, any point in the charge-breaking orbit space can be brought to its own canonical form;
\be
\label{Phi:charge-breaking:canonical}
\lr{\phi_1 } = \doublet{0}{v},
\lr{\phi_2 } = \doublet{u}{0},
\lr{\phi_a } = \doublet{0}{0}\ \mbox{for} \quad a > 2\,, \quad v^2 \equiv \sum_a v_a^2\, , u^2 \equiv \sum_a u_a^2
\ee
There are $N(N-1)/2$ such canonical choices. To avoid double counting of equivalent canonical forms within such choices we can restrict $v^2\ge u^2$.
The neutral orbit space corresponds to the limit $u^2\to 0$. The other extremum, $u^2=v^2$, corresponds to a rather special space of "maximally charge-breaking" (MCB) vacua.

\subsection{$K$-matrix formalism}

One could define the electroweak-scalar bilinears as components of a complex hermitian $N\times N$ matrix \cite{nagel,heidelberg,nishi2006}:
\be
\label{K:def}
K_{ab} \equiv (\phi^\dagger_b\phi_a)
\ee
According to \cite{heidelberg,nishi:nhdm} the $K$-matrix has the following properties: 
\begin{itemize}
\item
It is a hermitian positive-semidefinite matrix.
\item 
Its rank is $2$ for a charge-breaking vacuum and $1$ for a neutral vacuum. This is because we are dealing with electroweak doublets and not higher multiplets.
\end{itemize}



The fact that the $K$-matrix is a hermitian and rank $2$ matrix defines the dimensions of the neutral and charge breaking orbit space. It has at most two non-zero (and positive) eigenvalues, and the other $N-2$ eigenvalues are zero. In the case of a neutral vacuum, the $K$-matrix has only one non-zero (positive) eigenvalue and $N-1$ zeros. 

If rank$(K)=2$, then among the $N$ rows (columns) of $K$, there are at most two linearly independent rows (columns). Since this is an arbitrary choice, we require them to be the first and second rows (columns). All the other rows (columns) can be rewritten as linear combinations of lines $1$ and $2$.

Since $K$ is hermitian, these expansion coefficients are determined by the elements of the linearly independent lines \cite{NHDM2010}. Therefore, the number of algebraically independent gauge invariants $K_{ab}=\phi^\dag_b\phi_a$ can be counted from
\bea 
&& K_{1a}=\phi^\dag_a\phi_1, \quad a=1,\ldots,N    \nonumber  \\
&& K_{2b}=\phi^\dag_b\phi_2, \quad b=2,\ldots,N      
\label{minimal:K}
\eea
where $K_{11}$ and $K_{22}$ are real while the rest are complex. Therefore the dimension of the charge breaking orbit space is $4N-4$ [$2N$ from $K_{1a}$ + $2(N-1)$ from $K_{2b}$ - $2$ considering that $K_{11}$ and $K_{22}$ are real].

In the neutral orbit space rank$(K)=1$ and the first row/column is the linearly independent row/column, and therefore the dimension of the neutral orbit space is $2N-1$ [$2N$ from $K_{1a}$ - $1$ considering that $K_{11}$ is real].

\subsection{Adjoint representation}\label{section:adjoint}

The adjoint representation of the group $SU(N)_H$, is another way to look at the orbit space of the $N$-Higgs-doublet model.
Since the $K$-matrix (\ref{K:def}) is a hermitian $N\times N$ matrix, it can be written as a linear combination of other hermitian $N\times N$ matrices;
\be
K \equiv r_0 \cdot {\sqrt{\!{2 \over N(N-1)}}}\,{\bf 1}_N + r_i \lambda_i\,, \quad i = 1, \dots , N^2-1  \\
\label{K:decomposed}
\ee
where $\lambda_i$ are generators of $SU(N)$ \cite{NHDM2010}. 
The coefficients $r_0$ and $r_i$ can be calculated from the $K$-matrix: 
\be
r_0 = {\sqrt{\!{N-1\over 2N}}}\Tr\,K\,,\quad r_i = {1\over 2}\Tr[K\lambda_i]
\ee

The space of all possible real vectors $r^\mu \equiv (r_0, r_i)$ is called the adjoint space.
(The notation $r^\mu$ is similar to the Minkowski-space formalism developed for 2HDM in \cite{nishi2006,ivanov1,ivanov2,nishi2008}.)

The space of positive semidefinite $K$-matrices with rank $2$ is called the ''$K$-space". The mapping (\ref{K:decomposed}) from the $K$-space to the adjoint space $r^\mu$ defines a manifold called orbit space \cite{nishi:nhdm}.

The orbit space has a complicated shape. In 2HDM the $K$-matrix, a hermitian $2 \times 2$ matrix, automatically satisfies rank$K \le 2$, and the positive-semidefiniteness requires 
\be
\label{K:conditions:2HDM}
\Tr K \ge 0\,,\quad (\Tr K)^2 - \Tr[K^2] \ge 0 
\ee
which in the adjoint representation translates to
\be
\label{LC:def}
r_0\ge 0\,,\quad r_0^2 \ge r^2_i
\ee
which represents (the surface and the interior of) the forward light-cone in $\mathbb{R}^4$.

With more than two doublets, there are extra conditions in addition to (\ref{K:conditions:2HDM});
\be 
\frac{N-2}{2(N-1)} \leq \left(\frac{r_i}{r_0}\right)^2 \leq 1 
\label{orbit-space-definitions}
\ee
The full conditions are formulated in \cite{NHDM2010}.

The orbit space still lies inside the forward light-cone, but it occupies only a certain region inside it.

\subsection{Geometric description of the orbit space}\label{subsec:geometric}

Here we describe some geometric properties of the orbit space of NHDM in the adjoint space of vectors $r^\mu$, which can be deduced from (\ref{orbit-space-definitions}).

From $r^2_0 \geq r^2_i$, one could see that the orbit space has a conical shape. For convenience, we switch to the $(N^2-1)$-dimensional space of normalized vectors $n_i \equiv r_i/r_0$ where the neutral orbit space lies on the surface of the unit sphere $\vec n^2 = 1$, and the charge-breaking orbit space lies in a region inside it. 

The condition
\be 
\frac{N-2}{2(N-1)} \leq n^2_i \leq 1 
\ee
shows that the orbit space is actually located inside a conical shell lying between two coaxial cones, the "forward light-cone" and a certain inner cone. Note that the orbit space does not occupy the whole space between these two cones. In the case of 2HDM, $N=2$ the inner cone disappears, and the orbit space fills the entire cone. In NHDM, $N>2$, this inner cone appears as shown in Figure \ref{orbit-pic}.

\begin{figure} [ht]
\centering
\includegraphics[height=4cm]{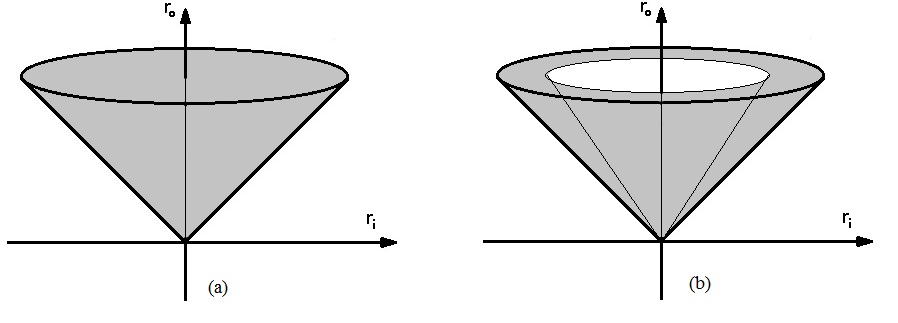}
\caption{The orbit space of 2HDM (a) versus the orbit space of NHDM (b)}
\label{orbit-pic}
\end{figure}

\subsection{The orbit space in 3HDM}\label{potential-3hdm-section}

In this Section we study the orbit space of the three-Higgs-doublet models, since certain examples of 3HDM will be presented in the next Chapters.

From (\ref{K:decomposed}) we have
\be
K \equiv r_0 \cdot {1\over\sqrt{3}}{\bf 1}_3 + r_i \lambda_i 
\ee
The standard expression for the Gell-Mann matrices is the following:
\bea
&& \lambda_1 = \left(\begin{array}{ccc} 0&1&0 \\ 1&0&0 \\ 0&0&0 \end{array}\right)\,,\quad
\lambda_2 = \left(\begin{array}{ccc} 0&-i&0 \\ i&0&0 \\ 0&0&0 \end{array}\right)\,,\quad
\lambda_3 = \left(\begin{array}{ccc} 1&0&0 \\ 0&-1&0 \\ 0&0&0 \end{array}\right) \nonumber\\[3mm]
&& \lambda_4 = \left(\begin{array}{ccc} 0&0&1 \\ 0&0&0 \\ 1&0&0 \end{array}\right)\,,\quad
\lambda_5 = \left(\begin{array}{ccc} 0&0&i \\ 0&0&0 \\ -i&0&0 \end{array}\right)\,,\quad
\lambda_6 = \left(\begin{array}{ccc} 0&0&0 \\ 0&0&1 \\ 0&1&0 \end{array}\right) \nonumber\\[3mm]
&&\lambda_7 = \left(\begin{array}{ccc} 0&0&0 \\ 0&0&-i \\ 0&i&0 \end{array}\right)\,,\quad
\lambda_8 = {1\over\sqrt{3}}\left(\begin{array}{ccc} 1&0&0 \\ 0&1&0 \\ 0&0&-2 \end{array}\right) 
\eea
The explicit expressions for the field bilinears are:
\bea
&& r_0 = {1 \over\sqrt{3}} \Tr K = {(\phi_1^\dagger\phi_1) + (\phi_2^\dagger\phi_2) + (\phi_3^\dagger\phi_3)\over\sqrt{3}} \\
&& r_i = {1\over 2}\Tr[K \lambda_i]\,, \quad r_3 = {(\phi_1^\dagger\phi_1) - (\phi_2^\dagger\phi_2) \over 2}\,,\quad
r_8 = {(\phi_1^\dagger\phi_1) + (\phi_2^\dagger\phi_2) - 2(\phi_3^\dagger\phi_3) \over 2\sqrt{3}} \quad
\nonumber\\
&&r_1 = \Re(\phi_1^\dagger\phi_2)\,,\ r_2 = \Im(\phi_1^\dagger\phi_2)\,,\
r_4 = \Re(\phi_3^\dagger\phi_1) \nonumber\\[2mm] 
&&r_5 = \Im(\phi_3^\dagger\phi_1)\,,\
r_6 = \Re(\phi_2^\dagger\phi_3)\,,\ r_7 = \Im(\phi_2^\dagger\phi_3) \nonumber
\eea
Note that $\lambda_5$ and therefore $r_5$ have opposite sign to the usual convention. This is done to show the symmetry upon cyclic permutation.

Recall that the orbit space of 3HDM sits between two coaxial shells
\be 
\frac{1}{4} \leq n^2_i \leq 1 \quad , \quad n_i = \frac{r_i}{r_0} \quad , \quad (i= 0, \cdots , 8)
\ee
As mentioned before, the orbit space has a complicated shape in the space of $r$'s. However, it is possible to project the orbit space onto a 3-dimensional space. We will see in Chapter \ref{chapter4} that the octahedral potential is linear in terms of the parameters $x$, $y$ and $z$ introduced as: 
\be
x = n^2_1 + n^2_4 + n^2_6 \quad , \quad y = n^2_2 + n^2_5 + n^2_7 \quad , \quad z = n^2_3 + n^2_8
\ee
It turns out that projecting the orbit space onto the 3-dimensional  $(x,y,z)$-space, gives it a much simpler shape.

\subsubsection{The orbit space in the $(x,y,z)$-space}

The definitions of $x$, $y$ and $z$ satisfy $1/4 \leq x+ y+ z \leq 1$, meaning that the orbit space sits inside a 3-dimensional pyramid in the $(x,y,z)$-space. However it does not fill the entire pyramid. We do not have the exact shape of the orbit space yet. Here we present the constraints that we have found on the borders of the orbit space inside the pyramid, which is represented in Figure (\ref{fig-xyz2}).

\begin{figure} [ht]
\centering
\includegraphics[height=7cm]{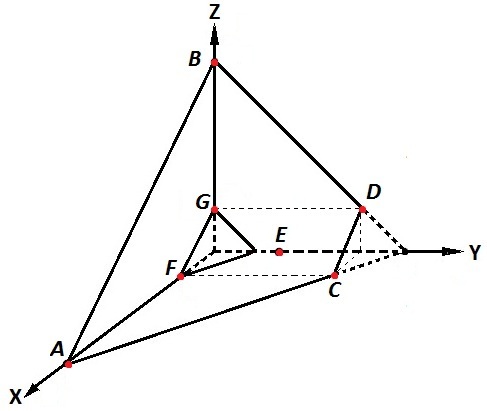}
\caption{The pyramid, with the apex at $(0,0,0)$, in the $(x,y,z)$-space where the orbit space projection fits between the shown points. Note that the base of the pyramid corresponds to the neutral orbit space.}
\label{fig-xyz2}
\end{figure}

\begin{itemize}
\item
{\boldmath $x_{min}$ and $x_{max}$ on $X$-axis (points $F$ and $A$)}\\
On the $XZ$-plane $y=0$ and therefore $1/4 \leq x+z \leq 1$, which represents a trapezoid. With the extra condition $z=0$ on the $X$-axis the minimum and maximum values of $x$ are respectively $x_F =1/4$ and $x_A =1$, which are realized through
\be
F : \doublet{0}{v}\doublet{v\sin(\pi/3)}{v\cos(\pi/3)}\doublet{v \sin(-\pi/3)}{v \cos(-\pi/3)} \quad , \quad A : \doublet{0}{v}\doublet{0}{v}\doublet{0}{v} \nonumber
\ee

\item
{\boldmath $z_{min}$ and $z_{max}$ on $Z$-axis (points $G$ and $B$)}\\
On the $XZ$-plane $1/4 \leq x+z \leq 1$, with the extra condition $x=0$ the minimum and maximum values of $z$ on the $Z$-axis are respectively $x_G =1/4$ and $x_B =1$, which are realized through
\be
G : \doublet{0}{v}\doublet{v}{0}\doublet{0}{0} \quad , \quad B : \doublet{0}{v}\doublet{0}{0}\doublet{0}{0} \nonumber
\ee

\item
{\boldmath $y_{unique}$ on $Y$-axis (point $E$)}\\
On the $Y$-axis $z=0$, which in terms of the doublets translates to $|\phi_a^{\dagger}\phi_a| = v^2$. The most general case satisfying this condition is realized through v.e.v.'s of the form $$ \doublet{0}{v} \doublet{v\sin\alpha}{v\cos\alpha e^{i\eta_2}}\doublet{v\sin\beta e^{i\eta_3}}{v\cos\beta e^{i\eta_4}}$$ With this v.e.v we calculate 
\be 
x=\frac{1}{3}\left[\cos^2\alpha\cos^2\eta_2 +\left(\cos\alpha\cos\beta\cos(\eta_4-\eta_2) + \sin\alpha\sin\beta\cos\eta_3  \right)^2 + \cos^2\beta\cos^2\eta_4 \right] \nonumber
\ee
Since we are on the $Y$-axis, we need $x=0$. The following four sets of conditions vanish $x$:
\bea
1: && \cos\eta_2=0, \quad \cos\eta_4=0, \quad \sin\alpha\sin\beta\cos\eta_3=-\cos\alpha\cos\beta  \nonumber \\
2: && \cos\alpha=0, \quad \cos\beta=0, \quad \cos\eta_3=0  \nonumber \\
3: && \cos\alpha=0, \quad \cos\eta_4=0, \quad \sin\beta\cos\eta_3=0  \nonumber \\
4: && \cos\eta_2=0, \quad \cos\beta=0, \quad \sin\alpha\cos\eta_3=0  \nonumber
\eea
Substituting any of these sets in $y$
\be 
y = \frac{1}{3}\left[\cos^2\alpha\sin^2\eta_2 + \left(\cos\alpha\cos\beta\sin(\eta_4-\eta_2) + \sin\alpha\sin\beta\sin\eta_3  \right)^2 + \cos^2\beta\sin^2\eta_4 \right] \nonumber
\ee
results in $y=1/3$, which proves that this point is the only point the orbit space touches the $Y$-axis. This point is realized at point $E$ with
\be 
E : \doublet{0}{i v}\doublet{v\sin\alpha}{v\cos\alpha}\doublet{v\sin\beta e^{i\eta_3}}{v\cos\beta}, \quad \cos\alpha \cos\beta+ \sin\alpha\sin\beta\cos\eta_3 = 0 \nonumber
\ee

\item
{\boldmath $y_{max}$ in $YZ$-plane (point $D$)}\\
In the $YZ$-plane $x=0$, which in terms of the doublets translates to $\Re(\phi_a^{\dagger}\phi_b) = 0$ with $( a,b = 1,2,3) $. The most general case satisfying this condition is realized through v.e.v.'s of the form $$ \doublet{0}{v_1}\doublet{0}{iv_2}\doublet{0}{0}$$ With this v.e.v we calculate 
\be 
y= \frac{\Im^2(\phi_a^{\dagger}\phi_b)}{r_0^2} = \frac{3(v_1v_2)^2}{({v_1}^2+{v_2}^2)^2}  \nonumber
\ee
which has maximum value $y_{max}=3/4$ at $v_1 =v_2$.
Therefore point $C$ in the $YZ$-plane realizes a vertex of the orbit space with
\be 
D : \doublet{0}{v}\doublet{0}{iv}\doublet{0}{0} \nonumber
\ee

\item
{\boldmath $y_{max}$ in $XY$-plane (point $C$)}\\
In the $XY$-plane $z=0$, which in terms of the doublets translates to $|\phi_a^{\dagger}\phi_a| = v^2$. The most general neutral v.e.v. satisfying this condition is realized through $$ \doublet{0}{v}\doublet{0}{ve^{i\eta_2}}\doublet{0}{ve^{i\eta_3}}$$ With this v.e.v we calculate 
\be 
y= \frac{1}{3}\left(\sin^2\eta_2 + \sin^2(\eta_3 - \eta_2)+ \sin^2 \eta_3  \right) \nonumber
\ee
Trying to maximize $y$, one gets the condition $\sin(2\eta_2)=-\sin(2\eta_3)=\sin2(\eta_3-\eta_2)$, which requires $\eta_2=-\eta_3$ or $\eta_2=n\pi/2 \pm \eta_3$, for $\eta_2=n\pi/2$ or $\eta_2=n\pi/2 + \pi/6$. Substituting each of  these four conditions in $y$, one gets $y_{max}=3/4$, which is realized in
\be 
C : \doublet{0}{v}\doublet{0}{ve^{2i\pi/3}}\doublet{0}{ve^{-2i\pi/3}} \nonumber
\ee

\item
{\boldmath The border line at $y=3/4$ (edge $CD$)}\\
The border of the orbit space at $y=3/4$ represents neutral vacua, therefore the upper components of the most general v.e.v. satisfying this condition are all zero. With $y=3/4$, we also have the condition $x+z=1/4$ which is satisfied in the most general case with
\be 
\doublet{0}{v_1}\doublet{0}{v_2e^{i\eta}}\doublet{0}{v_2e^{-i\eta}} \nonumber
\ee
With this v.e.v we calculate 
\bea 
y &=& \frac{\Im^2(\phi_a^{\dagger}\phi_b)}{r_0^2} = \frac{3}{(1+2r)^2}\left[2r\sin^2\eta + r^2\sin^2(2\eta) \right] \nonumber \\
&& \mbox{with}\quad  \left(\frac{v_2}{v_1}\right)^2=r \nonumber
\eea
To maximize $y$, one requires $\cos(2\eta)=-\frac{1}{2r}$. Substituting this value in $y$, results in $y=3/4$ for any value of $r$, which means that $x+z=1/4$ for any value of $r$. This condition represents the edge $CD$.

\item
{\boldmath The border line in $XZ$-plane (edge $AB$)}\\
In the $XZ$-plane $y=0$, which in terms of the doublets translates to $\Im(\phi_a^{\dagger}\phi_b) = 0$. The most general case satisfying this condition is realized through v.e.v.'s of the form $$ \doublet{0}{v_1}\doublet{0}{v_2}\doublet{0}{v_3}$$ With this v.e.v we calculate 
\be 
x= \frac{3(v_1v_2)^2+3(v_2v_3)^2+(v_3v_1)^2}{(v_1^2+v_2^2+v_3^2)^2} \quad , \quad z= \frac{v_1^4+v_2^4+v_3^4-v_1^2v_2^2-v_2^2v_3^2-v_3^2v_1^2}{(v_1^2+v_2^2+v_3^2)^2} \nonumber
\ee
For every value of $v_1, v_2$ and $v_3$, $x+z=1$, which represents the $AB$ edge in Figure (\ref{fig-xyz}).

\item
{\boldmath The border line in $YZ$-plane (edge $BD$)}\\
In the $YZ$-plane $x=0$, the most general case satisfying this condition, $\Re(\phi_a^{\dagger}\phi_b) = 0$, is realized through $$ \doublet{0}{v_1}\doublet{0}{iv_2}\doublet{0}{0}$$ which is a $CP$ violating minima. With this v.e.v we calculate 
\be 
y= \frac{3(v_1v_2)^2}{({v_1}^2+{v_2}^2)^2}\quad , \quad z= \frac{v_1^4 +v_2^4- v_1^2v_2^2}{{v_1}^2+{v_2}^2)^2} \nonumber
\ee
Introducing $\left(\frac{v_2}{v_1}\right)^2=r$, and rewriting $y$ and $z$ in terms of $r$
\be 
y= \frac{3r}{(1+r)^2} \quad , \quad z= \frac{1+r^2-r}{(1+r)^2} \nonumber
\ee
one could see that for any value of $r$, $y(r) + z(r) =1$, which represents the border line in the $YZ$-plane, the $BD$ edge.

\item
{\boldmath The border line in $XY$-plane (edge $AC$)}\\
In the $XY$-plane $z=0$, the most general v.e.v. satisfying this condition, $|\phi_a^{\dagger}\phi_a| = v^2$, is realized through $$ \doublet{0}{v}\doublet{0}{ve^{i\eta_2}}\doublet{0}{ve^{i\eta_3}}$$ With this v.e.v we calculate 
\be 
y= \frac{1}{3}\left(\sin^2\eta_2 + \sin^2(\eta_3 - \eta_2)+ \sin^2 \eta_3  \right), \quad x= \frac{1}{3}\left(\cos^2\eta_2 + \cos^2(\eta_3 - \eta_2)+ \cos^2 \eta_3  \right) \nonumber
\ee
For every value of $\eta_2$ and $\eta_3$, $x+y=1$, which represents the $AD$ edge in Figure (\ref{fig-xyz}).
\end{itemize}

So it seems that the orbit space lies inside the shape shown in Figure (\ref{fig-xyz}).

\begin{figure} [ht]
\centering
\includegraphics[height=7cm]{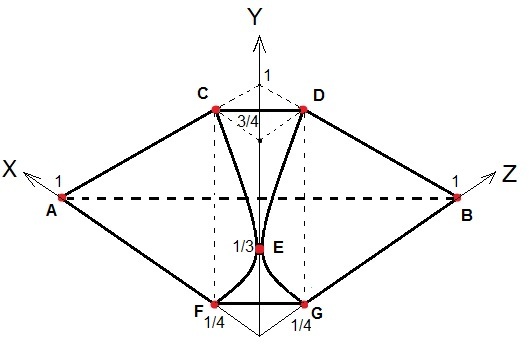}
\caption{The borders of the orbit space of 3HDM in the $(x,y,z)$-space.}
\label{fig-xyz}
\end{figure}

\section{The scalar potential in NHDM}\label{potential-nhdm-section}

In this Section we show how the scalar potential of multi-Higgs-doublet models is presented in the orbit space\footnote{The idea to switch to the orbit space in order to simplify the task of a group-invariant potential minimization is rather old, see \cite{sartori-old} and references therein.}. The generic Higgs potential in NHDM in terms of the doublets is written in a tensorial form \cite{CP,haber}:
\be
\label{V:tensorial}
V = Y_{\bar{a}b}(\phi^\dagger_a \phi_b) + Z_{\bar{a}b\bar{c}d}(\phi^\dagger_a \phi_b)(\phi^\dagger_c \phi_d) 
\ee
where all indices run from $1$ to $N$. The potential is constructed from $N^2$ bilinears $(\phi^\dagger_a \phi_b)$. Coefficients in the quadratic and quartic parts of the potential are grouped into components of Hermitian tensors $Y_{\bar{a}b}$
and $Z_{\bar{a}b\bar{c}d}$, respectively; there are $N^2$ independent components in $Y$ and $N^2(N^2+1)/2$ independent components in $Z$.

In the adjoint space the Higgs potential has the following form: 
\be
\label{V:general}
V = -M_\mu r^\mu + {1\over 2}\Lambda_{\mu\nu}r^\mu r^\nu \equiv 
- (M_0 r_0 + M_i r_i) + {1\over 2}\left(\Lambda_{00}r_0^2 + 2 \Lambda_{0i}r_0 r_i + \Lambda_{ij}r_ir_j\right) 
\ee
where the tensor $Y_{\bar{a}b}$ is written as the scalar $M_0$ and the vector $M_i$ in the adjoint space.
The number of free parameters in $M_0$ and the vector $M_i$ is $[1+ (N^2-1) =] N^2$.
The tensor $Z_{\bar{a}b\bar{c}d}$ is represented as the scalar $\Lambda_{00}$, vector $\Lambda_{0i}$ and symmetric tensor $\Lambda_{ij}$ in the adjoint space. The number of free parameters here is $[1+(N^2-1)+N^2(N^2-1)/2 =] N^2(N^2+1)/2$. 


Using the adjoint space $r^\mu$ makes the analysis of the Higgs potential in NHDM much simpler. It is easy to formulate and prove conditions for existence of a symmetry. We will present examples of this in Chapter \ref{chapter3} and \ref{chapter4}.

\chapter{Frustrated symmetries in NHDM}\label{chapter3}

In this Chapter we introduce a peculiar feature that arises in models with more than two Higgs doublets. We show that there exists certain explicit symmetries of the potential, which are necessarily broken whenever the electroweak symmetry breaking takes place. We term them ``frustrated symmetries'' because of their similarity to the phenomenon of
geometric frustration in condensed-matter physics. We think that such symmetries can be an interesting option when building models with desired properties beyond the Standard model.

The structure of the Chapter is the following; In Section \ref{frustrated} we explain how frustrated symmetries arise in NHDM with $N >2$. In Section \ref{3HDM-examples} we give several examples of such symmetries in 3HDM. The results of this Chapter are published in \cite{frustrated-paper}.

\section{General observations}\label{frustrated}

Suppose we have constructed a Higgs potential with some explicit symmetry.
That is, we have found coefficients $M_i$, $\Lambda_{0i}$ and $\Lambda_{ij}$ in the potential:
\be
V = - M_0 r_0 - M_i r_i + {1 \over 2}\Lambda_{00} r_0^2 + \Lambda_{0i} r_0 r_i + {1 \over 2}\Lambda_{ij} r_i r_j 
\ee
such that the potential is left invariant under group $G$ of 
transformation between doublets $\phi_a$. We will restrict ourselves to reparametrization transformations which are either  unitary or antiunitary transformations, and leave the Higgs kinetic term invariant. This means that $r_0$ is invariant under $G$, while the components $r_i$ are transformed by orthogonal transformations. Note that:
\be
r_0 = \sqrt{\frac{N-1}{2N}} \sum_a \phi_a^\dagger \phi_a \quad , \quad r_i= \frac{1}{2} \sum_{a,b} \phi_a^\dagger \lambda_{a,b}^i \phi_b
\ee
where $\lambda_i$'s are generators of $SU(N)$.

In particular, it is possible to devise such a symmetry group $G$ of transformations of doublets, 
which can be implemented via a non-trivial tensor  $\Lambda^{(G)}_{ij}$,
but not in a vector $M_i$ or $\Lambda_{0i}$. Let's study a potential with such symmetry; 
\be
\label{potential2}
V^{(G)} = - M_0 r_0 + {1 \over 2}\Lambda_{00} r_0^2 + {1 \over 2}\Lambda^{(G)}_{ij} r_i r_j
\ee


After EWSB, the doublets acquire some vacuum expectation values;
\be 
\lr{\phi_a} =\frac{1}{\sqrt{2}} \doublet{u_a}{v_a}
\ee

Therefore $\lr{r_0}$ proportional to $\lr{\phi_a^\dagger \phi_a}$ has a non-zero value and $\lr{r_i}$ is proportional to $\lr{\phi_a^\dagger \phi_b}$.

If all $\lr{r_i} = 0$, then for any pair of doublets $\lr{\phi_a^\dagger \phi_b}= 0$. Therefore the potential, depending only on $r_0 \neq 0$ which is invariant under $G$, conserves the symmetry $G$ after EWSB, and the vacuum remains invariant under $G$.

If at least one component $\lr{r_i} \not = 0$, then the vacuum spontaneously breaks the symmetry group $G$, either completely or down to a proper subgroup, simply because by construction no non-trivial vector $\lr{r_i}$ can be $G$-symmetric.

Thus, in order to see if a given symmetry has a chance to be conserved in a non-trivial vacuum,
we need to check whether or not the point $\vec n = \lr{\vec r}/\lr{r_0} = 0$ belongs to the orbit space of the model.

In Chapter \ref{chapter2} we showed that the point $\vec n = 0$ belongs to the orbit space 
only in the two-Higgs-doublet model. Therefore, in 2HDM any symmetry imposed on the potential can, in principle,
be conserved in the vacuum state provided we have chosen appropriate coefficients of the potential.
Indeed, if one acts on the state;
\be
\lr{\phi_1} =\frac{1}{\sqrt{2}} \doublet{0}{v} \quad , \quad \lr{\phi_2} =\frac{1}{\sqrt{2}} \doublet{v}{0} \label{2HDMvev}
\ee
with any (anti)unitary transformation which mixes the doublets, one arrives at the same state up to an electroweak transformation. That is, the corresponding gauge orbit is $G$-symmetric.

In NHDM with $N>2$, the point $\vec n = 0$ cannot be realized through doublets, since
\be
\dfrac{N-2}{2(N-1)} \leq \vec{n}^{2} \leq 1 
\ee
Therefore, the minimum of the $G$-symmetric potential (\ref{potential2}) unavoidably breaks the explicit symmetry after EWSB.

We call such a symmetry a {\bf frustrated symmetry} for the following reason: setting all $n_i = 0$ is equivalent to setting all products $(\fd_a \f_b)$  to zero while keeping all the norms $|\phi^\dagger_a \phi_a|$ non-zero and equal.
This is impossible to achieve with more than two doublets, simply because we have too little freedom of where
to place vacuum expectation values inside all three doublets. 

Even if the first two doublets are chosen as in (\ref{2HDMvev}), the third doublet of the same norm will have a non-zero product with the first or the second doublet.
\be
\lr{\phi_1} =\frac{1}{\sqrt{2}} \doublet{0}{v} \quad , \quad \lr{\phi_2} =\frac{1}{\sqrt{2}} \doublet{v}{0} \quad , \quad \lr{\phi_3} =\frac{1}{\sqrt{2}} \doublet{?}{?}
\ee

In other words, although we can "optimize" the v.e.v.'s in any pair of doublets (equal norms, zero product),
these optimal configurations are mutually incompatible in three or more doublets and cannot be satisfied simultaneously.
This is precisely what frustration in condensed-matter physics is about.

\section{Examples of frustrated symmetries in 3HDM}\label{3HDM-examples}

In this section we present examples of frustrated symmetries we have found in the three-Higgs-doublet model. 

\subsection{First example}
The simplest case is given by the potential which is symmetric
under any $SU(3)$ rotation and mixes the three doublets. It corresponds to the exceptional case $\Lambda_{ij} =0$:
\bea
V &=& -M_0r_0 + \Lambda_0 r^2_0 \nonumber \\
 &=& -\frac{M_0}{\sqrt{3}} (\fd_1 \f_1 + \fd_2 \f_2 + \fd_3 \f_3) + \frac{\Lambda_0}{3} (\fd_1 \f_1 + \fd_2 \f_2 + \fd_3 \f_3)^2 
\eea 
Clearly, no non-trivial vacuum can be symmetric under the entire $SU(3)$ group, even up to EW transformations. Here we study the case of:
\be
\lr{\phi_a} =\frac{1}{\sqrt{2}} \doublet{0}{v}
\ee
Let's study the excitation around this point:
\be
\phi_a = \doublet{w_{a}^{+}}{\frac{1}{\sqrt{2}}(v+h_{a}+i\eta_{a})}
\ee
And rewrite the potential in terms of the new parameters:
\bea
V &&= \left[ -\frac{M_0}{2\sqrt{3}}(2v)+\frac{\Lambda_0}{12}(12v^{3}) \right] (h_{1}+h_{2}+h_{3})\nonumber \\
&&+ \left[ -\frac{M_0}{2\sqrt{3}}+\frac{\Lambda_0}{12}(10v^{2}) \right] (h_{1}^{2}+h_{2}^{2}+h_{3}^{2})\nonumber \\
&& + \left[-\frac{M_0}{2\sqrt{3}}+\frac{\Lambda_0}{12}(6v^{2}) \right] (\eta_{1}^{2}+\eta_{2}^{2}+\eta_{3}^{2})\nonumber \\
&&+ \left[ \frac{\Lambda_0}{12}(8v^{2}) \right] (h_{1}h_{2}+h_{2}h_{3}+h_{3}h_{1})\nonumber \\
&& + \left[ -\frac{M_0}{\sqrt{3}}+\frac{\Lambda_0}{12}(12v^{2}) \right] (w_{1}^{+}w_{1}^{-}+w_{2}^{+}w_{2}^{-}+w_{3}^{+}w_{3}^{-})\nonumber \\
&& + \quad \mbox{higher order terms}
\eea
From vanishing linear terms, one gets $v^{2}=\frac{M_0}{\sqrt{3}\Lambda_0}$, which results in the following mass spectrum;
\bea
m^2(\eta_a)&:&   0  \nonumber \\
m^2( w_a^{\pm})&:&  0 \nonumber \\
m^2(h_a)&:& 0, \quad \frac{M_0}{3\sqrt{3}} \quad \mbox{(double degenerate)}
\eea
The massless bosons are of course expected, since the original symmetry of the potential is continuous (Goldstone Theorem).

\subsection{Second example}

The second example is given by the "tetrahedral" symmetry in 3HDM, defined by the following scalar potential:
\bea
V & = &  - M_0 r_0 + \Lambda_0 r_0^2 + \Lambda_1(r_1^2+r_4^2+r_6^2) + \Lambda_2(r_2^2+r_5^2+r_7^2)  \nonumber \\
&& + \Lambda_3(r_3^2+r_8^2) + \Lambda_4(r_1r_2 - r_4r_5 + r_6r_7) 
\label{tetrahedral3HDM}
\eea
In terms of doublets, this potential has the form:
\bea
V&=& - {M_0 \over \sqrt{3}} (\fd_1 \f_1 + \fd_2 \f_2 + \fd_3 \f_3) \nonumber \\
&&+ {\Lambda_0 \over 3} \left[(\fd_1 \f_1) + (\fd_2 \f_2) + (\fd_3 \f_3)\right]^2 \nonumber \\
&&+ {\Lambda_1 \over 3} \left[(\fd_1 \f_1)^2 + (\fd_2 \f_2)^2 + (\fd_3 \f_3)^2-(\fd_1 \f_1)(\fd_2 \f_2)-(\fd_2 \f_2)(\fd_3 \f_3)-(\fd_3 \f_3)(\fd_1 \f_1)\right]\nonumber \\
&&+ \Lambda_2 \left[(\Re \fd_1\f_2)^2 + (\Re \fd_2\f_3)^2 + (\Re \fd_3\f_1)^2\right]  \nonumber \\
&&+ \Lambda_3 \left[(\Im \fd_1\f_2)^2 + (\Im \fd_2\f_3)^2 + (\Im \fd_3\f_1)^2\right] \nonumber \\
&&+ \Lambda_4 \left[(\Re \fd_1\f_2) (\Im \fd_1\f_2) + (\Re \fd_2\f_3) (\Im \fd_2\f_3) + (\Re \fd_3\f_1) (\Im \fd_3\f_1)\right]
\label{tetrahedral3HDMfields}
\eea
This potential is symmetric under the chiral tetrahedral group (Appendix B). 

Note that the absence of any term linear in $r_i$ in (\ref{tetrahedral3HDM}) is due to the fact that no such term respects the tetrahedral symmetry. 

For $\Lambda_2 <0$ we have a vacuum of the form:
\be
\lr{\phi_a} =\frac{1}{\sqrt{2}} \doublet{0}{v}
\ee
And study the excitations around this point:
\be
\phi_a = \doublet{w_{a}^{+}}{\frac{1}{\sqrt{2}}(v+h_{a}+i\eta_{a})}
\ee
The potential has the following form in terms of the new parameters:
\bea
V &&=\left[ \dfrac{-M_{0}}{2\sqrt{3}}(2v)+\frac{\Lambda_{0}}{12}(12v^{3})+\frac{\Lambda_{2}}{4}(4v^{3}) \right](h_{1}+h_{2}+h_{3})\nonumber \\
&& +\left[ \overbrace{ \frac{-M_0}{2\sqrt{3}}+\frac{\Lambda_0}{12}(10v^2)+\frac{\Lambda_1}{12}(4v^2)+\frac{\Lambda_2}{2}(2v^2) }^A  \right] (h_{1}^{2}+h_{2}^{2}+h_{3}^{2})\nonumber \\
&& +\left[\overbrace{ \dfrac{-M_{0}}{2\sqrt{3}}+\frac{\Lambda_{0}}{12}(6v^{2})+\frac{\Lambda_{3}}{4}(2v^{2})}^B \right](\eta_{1}^{2}+\eta_{2}^{2}+\eta_{3}^{2})\nonumber \\
&& +\left[ \overbrace{ \frac{-M_0}{\sqrt{3}}+\frac{\Lambda_0}{12}(12v^2)}^C \right]({w_1}^{+}{w_1}^{-}+{w_2}^{+}{w_2}^{-}+{w_3}^{+}{w_3}^{-})\nonumber \\
&& +\left[ \overbrace{\frac{\Lambda_2}{4}(2v^2)+\frac{\Lambda_4}{4}(-iv^2) }^D \right] ({w_1}^{-}{w_2}^{+}+{w_2}^{-}{w_3}^{+}+{w_3}^{-}{w_1}^{+})\nonumber \\
&& +\left[\overbrace{\frac{\Lambda_{2}}{4}(2v^{2})+\frac{\Lambda_{4}}{4}(iv^{2})}^H \right] (w_{1}^{+}w_{2}^{-}+w_{2}^{+}w_{3}^{-}+w_{3}^{+}w_{1}^{-})\nonumber \\
&& +\left[\overbrace{\frac{\Lambda_{0}}{12}(8v^{2})-\frac{\Lambda_{1}}{12}(4v^{2})+\frac{\Lambda_{2}}{4}(4v^{2})}^{2E} \right] (h_{1}h_{2}+h_{2}h_{3}+h_{3}h_{1})\nonumber \\
&& +\left[\overbrace{\frac{\Lambda_{2}}{4}(2v^{2})-\frac{\Lambda_{3}}{4}(2v^{2})}^{2F}\right](\eta_{1}\eta_{2}+\eta_{2}\eta_{3}+\eta_{3}\eta_{1})\nonumber \\
&& +\left[\overbrace{\frac{\Lambda_{4}}{4}(2v^{2})}^{2G}\right](h_{1}\eta_{2}-h_{2}\eta_{1}+h_{2}\eta_{3}-h_{3}\eta_{2}+h_{3}\eta_{1}-h_{1}\eta_{3})
\label{V2ndExample}
\eea
From the vanishing linear terms $v^{2}=\frac{M_{0}}{\sqrt{3}(\Lambda_{0}+\Lambda_{2})}$, which simplifies the coefficients in (\ref{V2ndExample}):
\begin{eqnarray*}
A=(\frac{\Lambda_{0}+\Lambda_{1}}{3})v^{2} \quad &,& \quad B= (\frac{-\Lambda_{2}+\Lambda_{3}}{2})v^{2}\\
C= (-\Lambda_{2})v^{2} \quad &,& \quad D= (\frac{2\Lambda_{2}-i\Lambda_{4}}{4})v^{2}\\
H= (\frac{2\Lambda_{2}+i\Lambda_{4}}{4})v^{2} \quad &,& \quad 2E= (\frac{2\Lambda_{0}-\Lambda_{1}+3\Lambda_{2}}{3})v^{2}\\
2F= (\frac{\Lambda_{2}-\Lambda_{3}}{2})v^{2} \quad &,& \quad 2G= (\frac{\Lambda_{4}}{2})v^{2}
\end{eqnarray*}
These coefficients result in a mass spectrum of the form;
\bea
m^2( w_a^{\pm})&:&  \frac{6|\Lambda_2| \pm \sqrt{3}\Lambda_4}{4}v^2 \\
m^2(h_a)&:& \frac{v^2}{4} \left[ 2\Lambda_1+|\Lambda_2|+3\Lambda_3 \pm \sqrt{\Delta} \right] \nonumber \\ 
m^2(\eta_a)&:&  \frac{v^2}{4} \left[ 2\Lambda_1+|\Lambda_2|+3\Lambda_3 \pm \sqrt{ \Delta } \right]\nonumber \\
\mbox{with} &&   \Delta =  4\Lambda_1^2 +2|\Lambda_2|^2 +9\Lambda_3^2 -20\Lambda_1|\Lambda_2| -12\Lambda_1\Lambda_3 +30|\Lambda_2|\Lambda_3 +3\Lambda_4^2 \nonumber
\eea
The positivity of mass eigenstates imposes constraints on $\Lambda$'s;
\bea
&-6|\Lambda_{2}|< \sqrt{3} \Lambda_{4} < 6|\Lambda_{2}| &  \\
&|\Lambda_{2}|^2 +3\Lambda^2_{4}< 24(\Lambda_{1}\Lambda_{3} - |\Lambda_{2}|\Lambda_{3} + \Lambda_{1}|\Lambda_{2}|) & \nonumber
\eea

\subsection{Third example}

The third example of frustrated symmetry is the ''octahedral" symmetry in 3HDM, whose potential
is given by;
\be
V=- M_0 r_0 + \Lambda_0 r_0^2+ \Lambda_1(r_1^2+r_4^2+r_6^2)+ \Lambda_2(r_2^2+r_5^2+r_7^2)+ \Lambda_3(r_3^2+r_8^2)
\ee
This potential is symmetric under the octahedral group (Appendix B).

If $\Lambda_2 < 0$, the following v.e.v.'s realize the global minima of this potential;
\be
\lr{\phi_a} =\frac{1}{\sqrt{2}} \doublet{0}{v}
\ee
As usual we check the excitations of the fields around this point:
\be
\phi_a = \doublet{w_{a}^{+}}{\frac{1}{\sqrt{2}}(v+h_{a}+i\eta_{a})}
\ee
then rewrite the potential in terms of the new fields:
\bea
V&&= \left[ \dfrac{-M_{0}}{2\sqrt{3}}(2v)+\frac{\Lambda_{0}}{12}(12v^{3})+\frac{\Lambda_{2}}{4}(4v^{3}) \right](h_{1}+h_{2}+h_{3}) \nonumber \\
&& +\left[\overbrace{ \dfrac{-M_{0}}{2\sqrt{3}}+\frac{\Lambda_{0}}{12}(10v^{2})+\frac{\Lambda_{2}}{4}(2v^{2})+\frac{\Lambda_{1}}{12}(4v^{2})}^A \right] (h_{1}^{2}+h_{2}^{2}+h_{3}^{2})  \nonumber \\
&& +\left[\overbrace{ \dfrac{-M_{0}}{2\sqrt{3}}+\frac{\Lambda_{0}}{12}(6v^{2})+\frac{\Lambda_{3}}{4}(2v^{2})}^B \right](\eta_{1}^{2}+\eta_{2}^{2}+\eta_{3}^{2})  \nonumber \\
&& +\left[ \overbrace{ \dfrac{-M_{0}}{\sqrt{3}}+\frac{\Lambda_{0}}{12}(12v^{2})}^C \right](w_{1}^{+}w_{1}^{-}+w_{2}^{+}w_{2}^{-}+w_{3}^{+}w_{3}^{-})  \nonumber \\
&& +\left[ \overbrace{\frac{\Lambda_{2}}{4}(2v^{2})+}^D\right] (w_{1}^{-}w_{2}^{+}+w_{2}^{-}w_{3}^{+}+w_{3}^{-}w_{1}^{+}+w_{1}^{+}w_{2}^{-}+w_{2}^{+}w_{3}^{-}+w_{3}^{+}w_{1}^{-})  \nonumber \\
&& +\left[\overbrace{\frac{\Lambda_{0}}{12}(8v^{2})-\frac{\Lambda_{2}}{4}(4v^{2})-\frac{\Lambda_{1}}{12}(4v^{2})}^{2E} \right] (h_{1}h_{2}+h_{2}h_{3}+h_{3}h_{1}) \nonumber \\
&& +\left[\overbrace{\frac{\Lambda_{2}}{4}(2v^{2})-\frac{\Lambda_{3}}{4}(2v^{2})}^{2F}\right](\eta_{1}\eta_{2}+\eta_{2}\eta_{3}+\eta_{3}\eta_{1})   
\eea
With $v^{2}=\frac{M_{0}}{\sqrt{3}(\Lambda_{0}+\Lambda_{2})}$, we simplify the coefficients
\bea
A=(\frac{\Lambda_{0}+\Lambda_{2}}{3})v^{2} \quad &,& \quad B= (\frac{-\Lambda_{2}+\Lambda_{3}}{2})v^{2} \nonumber \\ \nonumber
C= (-\Lambda_{2})v^{2} \quad &,& \quad D= (\frac{\Lambda_{2}}{2})v^{2} \nonumber \\
2E= (\frac{2\Lambda_{0}+3\Lambda_{2}-\Lambda_{1}}{3})v^{2} \quad &,& \quad 2F= (\frac{\Lambda_{2}-\Lambda_{3}}{2})v^{2}  \nonumber 
\eea 
The eigenvalues of the mass matrix are:
\bea
m^2( w_a^{\pm})&:& 0, \quad \frac{3|\Lambda_2|}{2}v^2 \quad \mbox{(double degenerate)}   \nonumber \\
m^2(h_a)&:& 0, \quad \frac{6\Lambda_0-8|\Lambda_2|-2\Lambda_1}{3}v^2 \quad , \quad \frac{|\Lambda_2|+\Lambda_1}{3}v^2   \nonumber \\
m^2(\eta_a)&:& 0, \quad (3|\Lambda_2|+3\Lambda_3)v^2 \quad \mbox{(double degenerate)}   
\eea
Positivity of the eigenvalues of the mass matrix puts the following constraints on $\Lambda$'s:
\bea
&\Lambda_{2}<0 \quad , \quad \Lambda_{1}+|\Lambda_{2}|>0 \quad , \quad \Lambda_{0}-|\Lambda_{2}|>0&  \nonumber \\
&\Lambda_{3}>0 \quad , \quad \Lambda_{3}+|\Lambda_{2}|>0 &   
\eea

We note a remarkable phenomenological feature of this model: it is 2HDM-like.
Due to remaining symmetry, it exhibits certain degeneracy in the mass spectrum of the physical Higgs bosons, 
yielding just one mass for both charged Higgs bosons and three different masses for the neutral ones,
which precisely mimics the typical Higgs spectrum of 2HDM. This model is studied in detail in Chapter \ref{chapter4}.


\subsection{Forth example}

The forth example is of an Abelian frustrated symmetry in 3HDM; a $Z_3 \times Z_3$ symmetric potential, which is given by;
\bea
V & = & -\frac{M_0}{\sqrt{3}} \left[(\phi_1^\dagger \phi_1)+ (\phi_2^\dagger \phi_2)+(\phi_3^\dagger \phi_3)\right]
+ \frac{\Lambda_0}{3} \left[(\phi_1^\dagger \phi_1)+ (\phi_2^\dagger \phi_2)+(\phi_3^\dagger \phi_3)\right]^2 \nonumber\\
&&+ \Lambda_1 \left[(\phi_1^\dagger \phi_1)^2+ (\phi_2^\dagger \phi_2)^2+(\phi_3^\dagger \phi_3)^2\right]
+ \Lambda_2 \left[|\phi_1^\dagger \phi_2|^2 + |\phi_2^\dagger \phi_3|^2 + |\phi_3^\dagger \phi_1|^2\right] \nonumber\\
&&+ \Lambda_3 \left[(\phi_1^\dagger \phi_2)(\phi_1^\dagger \phi_3) + (\phi_2^\dagger \phi_3)(\phi_2^\dagger \phi_1) + (\phi_3^\dagger \phi_1)(\phi_3^\dagger \phi_2)\right]+ h.c.
\eea
We study the vacuum point:
\be
\lr{\phi_a} =\frac{1}{\sqrt{2}} \doublet{0}{v}
\ee
And check the excitations of the fields around this point:
\be
\lr{\phi_a} = \doublet{w_{a}^{+}}{\frac{1}{\sqrt{2}}(v+h_{a}+i\eta_{a})}
\ee

Then rewrite the potential in terms of the new fields:
\bea
V&=& \left[ \dfrac{-M_{0}}{\sqrt{3}}(v)+ (\Lambda_{0}+ \Lambda_1+ \Lambda_{2} + \Lambda_3 + {\Lambda_3}^*)(v^{3}) \right](h_{1}+h_{2}+h_{3})  \nonumber \\ 
&& +\left[\overbrace{ \dfrac{-M_{0}}{2\sqrt{3}}+\frac{\Lambda_{0}}{12}(10v^{2}) +\frac{\Lambda_{1}}{4}(6v^{2}) +\frac{\Lambda_{2}}{4}(2v^{2})+\frac{\Lambda_3+ {\Lambda_3}^*}{4}(v^{2})}^A \right] (h_{1}^{2}+h_{2}^{2}+h_{3}^{2})\nonumber \\
&& +\left[\overbrace{ \dfrac{-M_{0}}{2\sqrt{3}}+\frac{\Lambda_{0}}{12}(6v^{2})+\frac{\Lambda_1}{4}(2v^{2}) + \frac{\Lambda_{2}}{4}(2v^{2})-\frac{\Lambda_3+ {\Lambda_3}^*}{4}(v^{2}) }^B \right](\eta_{1}^{2}+\eta_{2}^{2}+\eta_{3}^{2})\nonumber \\
&& +\left[ \overbrace{ \dfrac{-M_{0}}{\sqrt{3}}+\Lambda_{0}v^{2} + \Lambda_1 v^2 }^C \right](w_{1}^{+}w_{1}^{-}+w_{2}^{+}w_{2}^{-}+w_{3}^{+}w_{3}^{-}) \nonumber \\
&& +\left[ \overbrace{\Lambda_{2}v^{2}+  \frac{\Lambda_3 + {\Lambda_3}^*}{2}(v^2)}^D\right] (w_{1}^{-}w_{2}^{+}+ w_{2}^{-}w_{3}^{+}+ w_{3}^{-}w_{1}^{+}+ w_{1}^{+}w_{2}^{-}+ w_{2}^{+}w_{3}^{-}+ w_{3}^{+}w_{1}^{-})\nonumber \\
&& +\left[\overbrace{\frac{\Lambda_{0}}{12}(8v^{2})+\Lambda_{2}v^{2}+\frac{\Lambda_3+ {\Lambda_3}^*}{4}(5v^{2})}^{2E} \right] (h_{1}h_{2}+h_{2}h_{3}+h_{3}h_{1})\nonumber \\
&& +\left[\overbrace{\frac{\Lambda_3+ {\Lambda_3}^*}{4}(v^{2})}^{2F}\right](\eta_{1}\eta_{2}+\eta_{2}\eta_{3}+\eta_{3}\eta_{1})\nonumber \\
&& +\left[\overbrace{i\frac{\Lambda_3- {\Lambda_3}^*}{4}(v^{2})}^{G}\right](h_1\eta_3 +h_2\eta_1 +h_3\eta_2 -h_1\eta_1 -h_2\eta_2 -h_3\eta_3) 
\eea
With $v^{2}=\frac{M_{0}}{\sqrt{3}(\Lambda_0 +\Lambda_1 +\Lambda_2 +\Lambda_3 +{\Lambda_3}^* )}$ we simplify the coefficients of the potential:
\bea
A=(\frac{2\Lambda_{0}+2\Lambda_1 -\Lambda_3 -{\Lambda_3}^*}{2})v^{2} \quad &,& \quad B= (\frac{-3\Lambda_3 -3{\Lambda_3}^*}{4})v^{2} \nonumber \\
C= (-\Lambda_{2} -\Lambda_3 -{\Lambda_3}^*)v^{2} \quad &,& \quad D= (\frac{2\Lambda_{2} +\Lambda_3 +{\Lambda_3}^*}{2})v^{2} \nonumber \\
2E= (\frac{8\Lambda_{0}+12\Lambda_{2}+ 15\Lambda_3 +15{\Lambda_3}^*}{12})v^{2} \quad &,& \quad 2F= (\frac{\Lambda_3 +{\Lambda_3}^*}{4})v^{2} \nonumber \\
G= i(\frac{\Lambda_3 -{\Lambda_3}^*}{4})v^{2}&& \nonumber
\eea 
The eigenvalues of the mass matrix are:
\bea
m^2( w_a^{\pm})&:&  \Lambda_2 v^2 , \quad \frac{-4\Lambda_2 -3\Lambda_3 -3{\Lambda_3}}{2}v^2  \\
m^2(h_a)&:& \frac{v^2}{12}\left( 16\Lambda_0 +24\Lambda_1 -24\Lambda_2 -27\Lambda_3 -27{\Lambda_3}^* \right) \nonumber \\ 
&& \frac{v^2}{18}\left( 20\Lambda_0 +12\Lambda_1 +12\Lambda_2 +9\Lambda_3 +9{\Lambda_3}^* \right) \nonumber \\ 
&& \frac{v^2}{12}\left[ 8\Lambda_0 +12\Lambda_1 -6\Lambda_2 -21\Lambda_3 -21{\Lambda_3}^* + \sqrt{\Delta} \right] \nonumber \\ 
m^2(\eta_a)&:&  -\frac{7v^2}{4} \left(\Lambda_3 +{\Lambda_3}^* \right) , \quad -v^2 \left(\Lambda_3 +{\Lambda_3}^* \right) \nonumber \\ 
&&  \frac{v^2}{12}\left[ 8\Lambda_0 +12\Lambda_1 -6\Lambda_2 -21\Lambda_3 -21{\Lambda_3}^* - \sqrt{\Delta} \right] \nonumber \\ 
\mbox{with} && \Delta = \left(8\Lambda_0 +12\Lambda_1 -12\Lambda_2 -3\Lambda_3 -3{\Lambda_3}^*  \right)^2 + \frac{9(3\sqrt{3}-15)}{4} \left( \Lambda_3 - {\Lambda_3}^* \right)^2 \nonumber
\eea


\chapter{3HDM with octahedral symmetry}\label{chapter4}

In this Chapter we focus on a specific discrete symmetry group that can be imposed on the scalar potential of the three-Higgs-doublet model: the full (achiral) octahedral symmetry group $O_h$. We see two main goals in this study. 

First; $O_h$ is among the largest realizable finite symmetry groups which can be imposed on the scalar potential in 3HDM. The word "realizable" stresses that when imposing such a symmetry group, we obtain a potential that is symmetric exactly under this and not any larger group. It is therefore interesting to check what are the phenomenological consequences of such a high symmetry of the potential.

Second; the octahedral 3HDM serves as a good illustration of the power of geometric and group-theoretic methods which were developed in \cite{NHDM2010} for the space of Higgs bilinears and which are further developed in this Chapter. We believe that these methods can be used in the analysis of other similar models.

The structure of the Chapter is as follows:  In Section \ref{potential-octahedral} we introduce the octahedral 3HDM potential and describe its symmetry group and symmetry breaking possibilities. In Section \ref{minimization} we show the main geometric idea in order to minimize the potential, with a simplified example. Then we apply the method to the full 3-dimensional orbit space and derive positivity conditions. In Section \ref{different-minima} we perform a phenomenological analysis on each possible type of symmetry breaking, and in Appendix \ref{detailed} we study a specific minimum in more detail.

\section{The scalar potential}\label{potential-octahedral}

The Octahedral 3HDM potential is defined by (see Appendix B):
\be
V=-M_{0}r_{0}+\Lambda_{0}r_{0}^2+\Lambda_{1}(r_{1}^2+r_{4}^2+r_{6}^2)+\Lambda_{2}(r_{2}^2+r_{5}^2+r_{7}^2)+\Lambda_{3}(r_{3}^2+r_{8}^2)
\label{Octahedralgeneral}
\ee
To write the potential in terms of the doublets, we remind the reader of the explicit expressions for the field bilinears in 3HDM:
\bea
&& r_0 = {(\phi_1^\dagger\phi_1) + (\phi_2^\dagger\phi_2) + (\phi_3^\dagger\phi_3)\over\sqrt{3}}\\
&& r_3 = {(\phi_1^\dagger\phi_1) - (\phi_2^\dagger\phi_2) \over 2}\,,\quad
r_8 = {(\phi_1^\dagger\phi_1) + (\phi_2^\dagger\phi_2) - 2(\phi_3^\dagger\phi_3) \over 2\sqrt{3}} \nonumber\\
&&r_1 = \Re(\phi_1^\dagger\phi_2)\quad , \quad  r_2 = \Im(\phi_1^\dagger\phi_2) \quad , \quad  r_4 = \Re(\phi_3^\dagger\phi_1)\nonumber\\
&&r_5 = \Im(\phi_3^\dagger\phi_1)\quad, \quad r_6 = \Re(\phi_2^\dagger\phi_3)\quad , \quad r_7 = \Im(\phi_2^\dagger\phi_3)\nonumber
\eea
Therefore the potential has the following form in terms of fields: 
\bea
V&=&-\frac{M_0}{\sqrt{3}}\left(\phi_1^{\dagger}\phi_1+\phi_2^{\dagger}\phi_2+\phi_3^{\dagger}\phi_3\right)+\frac{\Lambda_0}{3}\left(\phi_1^{\dagger}\phi_1+\phi_2^{\dagger}\phi_2+\phi_3^{\dagger}\phi_3\right)^2  \\
&&+\Lambda_1\left[(\Re\phi_1^{\dagger}\phi_2)^2+(\Re\phi_2^{\dagger}\phi_3)^2+(\Re\phi_3^{\dagger}\phi_1)^2\right] \nonumber \\
&&+\Lambda_2\left[(\Im\phi_1^{\dagger}\phi_2)^2+(\Im\phi_2^{\dagger}\phi_3)^2+(\Im\phi_3^{\dagger}\phi_1)^2\right]\nonumber \\
&&+\frac{\Lambda_3}{3}\left[(\phi_1^{\dagger}\phi_1)^2+(\phi_2^{\dagger}\phi_2)^2+(\phi_3^{\dagger}\phi_3)^2-(\phi_1^{\dagger}\phi_1)(\phi_2^{\dagger}\phi_2)-(\phi_2^{\dagger}\phi_2)(\phi_3^{\dagger}\phi_3)-(\phi_3^{\dagger}\phi_3)(\phi_1^{\dagger}\phi_1)\right]\nonumber 
\eea
This potential is symmetric under full achiral group $O_h$ (Appendix B).


\subsection{Symmetry breaking patterns}\label{breaking-patterns}

As it is shown below, in a 3HDM a sufficiently symmetric group $G$, such as the octahedral symmetry, can not break down completely after EWSB, since it will create so many minima \cite{A4-potential}. So, it breaks down to certain subgroups of $G$, which narrows down the list of possible patterns of v.e.v.s. Several such patterns can be simply guessed. Whether they are actually the minima of the potential for some values of the parameters, will be studied later; with this group-theoretic argument we can only list patterns which have a chance to be minima.

Below we list several such patterns. We will focus on neutral minima only, and therefore we present only the lower components of the doublets.

\begin{itemize}
\item
{\boldmath $(v,v,v)$}. This minimum is invariant under permutations of the doublets and under $CP$; the symmetry group is $D_3\times Z_2$ of order 12. Applying transformations which do not leave it invariant, we obtain four minima of this type: $(v,v,v)$, $(v,v,-v)$, $(v,-v,v)$, and $(v,-v,-v)$.

\item
{\boldmath $(v,0,0)$}. This point is invariant under sign flips, under $CP$ and under transposition of the last two doublets, which form the symmetry group $D_4 \times Z_2$ of order 16. The three minima of this type are $(v,0,0)$, $(0,v,0)$, $(0,0,v)$. 

\item
{\boldmath $(v,v,0)$}.  This point is symmetric under the sign flip of the third doublet, permutation of the first and second doublet, $CP$ transformation and their combinations. The symmetry group is $(Z_2)^3$ of order 8. The six minima of this type are $(v,\pm v,0)$, $(v,0,\pm v)$, $(0,v,\pm v)$.

\item
{\boldmath $(v, iv,0)$}. This point is invariant under the sign flip of the third doublet, the $CP$ transformation compensated by the sign flip of the second doublet, exchange of the first two doublets compensated by sign flip of the second doublet, and their products. The symmetry group is $Dh_4$ of order 8. The six minima are $(v,\pm iv,0)$, $(v,0,\pm iv)$, $(0,v,\pm iv)$.

\item
{\boldmath $(v, v e^{i\pi/3},v e^{-i\pi/3})$} and {\boldmath $(v, v e^{2i\pi/3},v e^{-2i\pi/3})$}.
The symmetry group is $S_3$ and is made of permutations of the three doublets compensated by appropriate sign flips and/or $CP$. There are six distinct minima of the first type and two distinct minima of the second type.
\end{itemize}

This list exhausts all sufficiently symmetric points. In the next Section we will find which ones can actually be minima of the potential.

\section{Minimization of the scalar potential}\label{minimization}

The octahedral 3HDM potential (\ref{Octahedralgeneral}) looks deceptively simple. If one plunges into straightforward minimization of the potential, it becomes a challenging task to reconstruct the full phase diagram of the model: that is, to find the positions (and even types) of the minima for arbitrary values of $\Lambda_i$. Here we develop an efficient geometric method to solve this problem. Before we proceed with the solution, let us illustrate this method with a simple toy model.

\subsection{Toy model}\label{toy-model}

Suppose only real-valued fluctuations in all three doublets are allowed ($r_2 = r_5 = r_7 =0$). Then the potential simplifies to
\bea
V&=&-M_{0}r_{0}+r_{0}^2\left(\Lambda_{0} + \Lambda_{1} x + \Lambda_{3}z \right) \nonumber \\
\mbox{with} && x = n_{1}^2+n_{4}^2+n_{6}^2\quad , \quad z = n_3^2 + n_8^2
\label{toy1}
\eea
Recall from Section \ref{potential-3hdm-section} that in 3HDM $1/4 \le n^2_i \le 1$, therefore $1/4 \le x+z \le 1$, which defines a trapezoid shown in Figure (\ref{fig-xz}).

\begin{figure}[htb]
\begin{center}
\includegraphics[height=7cm]{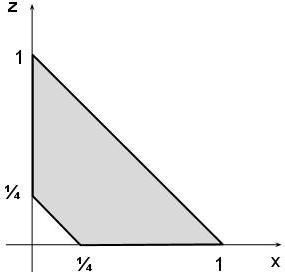}
\caption{The orbit space of the toy model projected on the $(x,z)$-plane.}
\label{fig-xz}
\end{center}
\end{figure}

The four vertices of this trapezoid correspond to the following patterns of the three doublets:
\bea
x = 0\,,\ z = 1: && \doublet{0}{v}\,,\  \doublet{0}{0}\,,\ \doublet{0}{0} \label{vertices-toy}\\
x = 1\,,\ z = 0: && \doublet{0}{v}\,,\  \doublet{0}{v}\,,\ \doublet{0}{v} \nonumber\\
x = 0\,,\ z = 1/4: && \doublet{0}{v}\,,\  \doublet{v}{0}\,,\ \doublet{0}{0} \nonumber\\
x = 1/4\,,\ z = 0: && \doublet{0}{v}\,,\  \doublet{v\sin\alpha}{v\cos\alpha}\,,\ \doublet{v\sin\beta}{v\cos\beta}\,,\ \alpha = - \beta = {\pi\over 3}  \nonumber
\eea
The first two points correspond to the neutral vacuum, and they were already listed in the previous section, while the last two are charge-breaking minima. One can also show that all points inside this trapezoid are realizable by fields (in other words, when the full orbit space is projected on the $(x,z)$-plane, it covers the entire trapezoid).

\subsubsection{Positivity conditions}

The positivity conditions are restrictions on the parameters which guarantee that the potential is bounded from below for all values of the Higgs fields. It is known that it is necessary and sufficient to require that the quartic part of the potential is positive inside the orbit space.

Now comes the crucial point. Since the potential (\ref{toy1}) is a linear function in $x$ and $z$ defined inside this trapezoid and therefore convex, it is sufficient to write the positivity conditions at the four vertices (\ref{vertices-toy}); then they will be automatically satisfied at all points inside the trapezoid. Thus, we arrive at the following positivity conditions for our toy model:
\be
\Lambda_0 + \Lambda_1 > 0\,,\quad 
\Lambda_0 + \Lambda_3 > 0\,,\quad 
4\Lambda_0 + \Lambda_1 > 0\,,\quad 
4\Lambda_0 + \Lambda_3 > 0 \label{positivity-toy}
\ee

In general, none of these inequalities can be removed because $\Lambda_0$ can be positive or negative. 

Now let us show how the global minima can be found for different values of the parameters. Since the potential is a linear function of $x$ and $z$, we can introduce on the $(x,z)$-plane the {\em direction of steepest descent} (DSD), which is given by vector $(-\Lambda_1,-\Lambda_3)$. Then, we need to find what points of the trapezoid lie farthest along this direction; such points give the $(x,z)$ coordinates of the global minima, that is, patterns of v.e.v.s of the doublets. The last step is to find the overall v.e.v. normalization $v$.	

By checking all possible directions of steepest descent, one finds situations where vertices or edges of the trapezoid correspond to the global minimum (see Figure(\ref{fig-xz-descent})).

\begin{figure}[htb]
\begin{center}
\includegraphics[height=7cm]{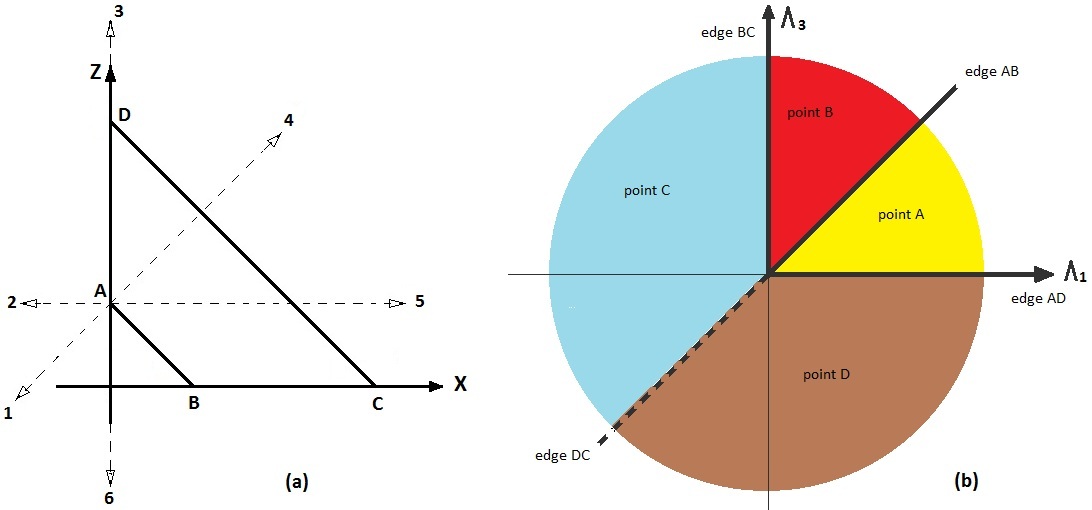}
\caption{(a) Various possible directions of steepest descent and the corresponding minima. (b) Phase diagram of the corresponding potential.}
\label{fig-xz-descent}
\end{center}
\end{figure}

If DSD is parallel to direction 1, which happens at $\Lambda_1=\Lambda_3 > 0$, then the entire edge $AB$ corresponds to the global minimum. If DSD is between directions 1 and 2 ($\Lambda_1 > \Lambda_3 > 0$), point A is the minimum. When DSD is along direction 2 ($\Lambda_1 > 0, \Lambda_3 = 0$), all the points on the $AD$ edge are minima. When DSD rotates further from direction 2 to direction 4, point D becomes the minimum. When DSD reaches direction 4 ($\Lambda_1 = \Lambda_3 < 0$), then all the points on the $DC$ edge become minima of the potential. And so on. In this way we can construct the full phase diagram of the model in the $\Lambda_i$ parameter space.



\section{Different minima}\label{different-minima}

In this section we analyse the potential for each of the minima found in Section \ref{breaking-patterns}. We focus only on the four possible points corresponding to the neutral vacuum.

\subsection{The minimum $(v,v,v)$}
In order for this point to be a minimum, it is necessary to have $\Lambda_1 <0$. The other parameters, $\Lambda_2, \Lambda_3$, can be of any sign, but if negative, they cannot be lower than a certain limit. Explicit expansion of the potential at this point gives the following scalar mass spectrum:
\bea
m^2(h_a^{\pm})&:& 0, \quad \frac{3}{2}|\Lambda_1|v^2 \quad \mbox{(double degenerate)} \nonumber \\
m^2(h_a^0)&:& (|\Lambda_1|+\Lambda_3)v^2 \quad \mbox{(double degenerate)} \nonumber \\
&& \frac{3}{2}(|\Lambda_1|+\Lambda_2)v^2 \quad \mbox{(double degenerate)} \nonumber\\
&& 0, \quad 2(\Lambda_0-|\Lambda_1|)v^2
\eea
The positivity of the mass eigenstates gives extra conditions on $\Lambda$'s:
\be 
\Lambda_1<0, \quad \Lambda_3+|\Lambda_1|>0, \quad \Lambda_0-|\Lambda_1|>0 , \quad \Lambda_2+|\Lambda_1|>0
\ee

The remarkable feature of this spectrum is that it is 2HDM-like, namely, it has only one massive charged scalar
and three massive neutral scalars. We study this interesting case in more detail in Appendix C.

\subsection{The minimum $(v,0,0)$}

In order for this point to be the minimum, $\Lambda_3$ must be negative. Checking the excitations around the minima, will result is $v^2=\frac{\sqrt{3}M_0}{\Lambda_0 - |\Lambda_3|}$, and the mass spectrum is:
\bea
m^2(h_a^{\pm})&:& 0, \quad \frac{|\Lambda_3|}{2}v^2 \quad \mbox{(double degenerate)} \nonumber \\
m^2(h_a^0)&:& 0, \quad \frac{2(\Lambda_0-|\Lambda_3|)}{3}v^2  \nonumber \\
&&\frac{\Lambda_1+|\Lambda_3|}{2}v^2 \quad \mbox{(double degenerate)} \nonumber \\
&&\frac{\Lambda_2+|\Lambda_3|}{2}v^2 \quad \mbox{(double degenerate)} 
\eea
With this breaking pattern, the spectrum is again 2HDM-like. Positivity of the mass eigenstates puts extra condition on $\Lambda$'s:
\be
\Lambda_0 > 0, \quad \Lambda_0 - |\Lambda_3| > 0, \quad \Lambda_1 +|\Lambda_3|> 0 , \quad \Lambda_2 +|\Lambda_3|> 0
\ee

\subsection{The minimum $(v,iv,0)$}
This point appears as a minimum when $\Lambda_1 >0, \Lambda_2,\Lambda_3<0$. For $v^2=\frac{6M_0}{\sqrt{3}(4\Lambda_0 - 3|\Lambda_2| -|\Lambda_3|)}$ the mass spectrum has the following form:
\bea
m^2(h_a^{\pm})&:&  0, \quad |\Lambda_2| v^2, \quad \frac{|\Lambda_2| +|\Lambda_3|}{2}v^2 \nonumber \\
m^2(h_a^0)&:& 0, \quad \frac{\Lambda_1+|\Lambda_3|}{2}v^2  \quad \mbox{(double degenerate)} \nonumber \\
&& \frac{3\Lambda_1+3|\Lambda_2|}{2}v^2, \quad (|\Lambda_2|-|\Lambda_3|)v^2, \quad \frac{4\Lambda_0-3|\Lambda_2|-|\Lambda_3|}{3}v^2 
\eea

There appear extra constraints on the $\Lambda$'s from the positivity of the mass eigenstates here:
\bea 
&|\Lambda_2| + |\Lambda_3 | >0,  \quad \Lambda_1 +|\Lambda_3| >0, \quad \Lambda_1 +|\Lambda_2| >0  \nonumber \\
&  |\Lambda_2| -|\Lambda_3| >0, \quad 4\Lambda_0 - 3|\Lambda_2| -|\Lambda_3| >0  
\eea

\subsection{The minimum $(\pm v,ve^{i\pi/3},ve^{-i\pi/3})$}

For this point to be minimum, one requires $\Lambda_2<0$, $\Lambda_1<|\Lambda_2|$ and $\Lambda_1<\Lambda_3$. The mass spectrum has the following form
\bea
m^2(h_a^{\pm})&:& 0, \quad \frac{3|\Lambda_2|}{2}v^2,\quad \frac{3(|\Lambda_2|-\Lambda_1)}{4}v^2  \nonumber \\
m^2(h_a^0)&:& 0, \quad (4\Lambda_0-3|\Lambda_2| +\Lambda_1) v^2  \nonumber \\
&& \left[\frac{3(|\Lambda_2| + \Lambda_1)}{2} + \Lambda_3 - \Lambda_1 \pm \sqrt{\Delta}\right]v^2  
\quad \mbox{(double degenerate)} \nonumber \\
&& \mbox{with} \quad \Delta=\left(\frac{3(|\Lambda_2| + \Lambda_1)}{2}\right)^2 + \left(\Lambda_3 - \Lambda_1\right)^2 
\eea
This spectrum shows a bizarre pattern: the charged Higgs bosons are non-degenerate, while the neutrals are 2HDM-like.


\chapter{Abelian symmetries in NHDM}\label{chapter5}

\section{Symmetries in NHDM}

One of the interesting branches of study in non-minimal Higgs sectors is the possible symmetries one can impose on the potential. In particular, three topics are of importance:
\begin{enumerate}
\item
Which symmetries can a general NHDM potential have? 
\item
What are the possible ways to spontaneously break these symmetries? 
\item
How can these symmetries be extended into the fermion sector?
\end{enumerate}

In this Chapter we explore the first question. When studying the symmetries of the potential, one focuses on the reparametrization transformations, 
which mix different doublets but leave the intradoublet structure untouched. These transformations keep the kinetic term invariant, and are either unitary or antiunitary.

In this Chapter we focus on the unitary transformations. Such transformations form the groups $U(N)$. We are looking for the physically distinct transformations, therefore we leave out the common phase rotations since they act trivially on the
potential. The overall phase factor multiplication which is taken into account by the gauge group $U(1)_Y$ should be left out. This leaves us with the special unitary group $SU(N)$ which has to be factor-grouped by its center $Z_N$. We call the resulting group $PSU(N)$. We construct this $PSU(N)$ for every NHDM, and look for the maximal Abelian subgroups of it, and then explore all its realizable subgroups. 

There are several maximal Abelian subgroups in $PSU(N)$, but they are all conjugate to each other. Therefore we could pick one maximal Abelian subgroup and exhaust its subgroups. Figure (\ref{TreeNHDM}) shows the tree of subgroups in NHDM.
\begin{figure}[htb]
\begin{center}
\includegraphics[height=7.5cm]{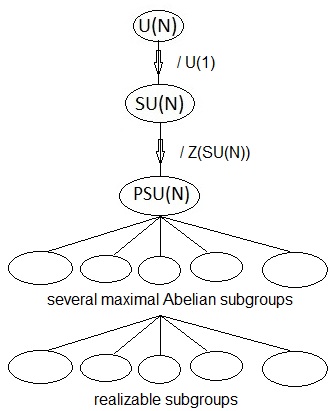}
\caption{The tree of subgroups in NHDM}
\label{TreeNHDM}
\end{center}
\end{figure}

When searching for the symmetry group which can be implemented in the scalar sector of a non-minimal Higgs model,
one must distinguish between {\em realizable} and {\em non-realizable} groups, \cite{Ivanov2}.
If it is possible to write a potential which is symmetric under a group $G$
but not symmetric under a larger symmetry group containing $G$, we call $G$ a realizable group.
A non-realizable group, on the contrary, automatically leads to a larger symmetry group of the potential.
The true symmetry properties of the potentials are reflected in realizable groups.

For example, a potential which depends on the first doublet only via $(\phi_1^\dagger \phi_1)$ 
is obviously symmetric under the cyclic group $\Z_n$ of discrete phase rotations of this doublet for any $n$. However, these $\Z_n$'s have no interest on their own because they trivially arise as subgroups of the larger symmetry group of this potential $U(1)$ describing arbitrary phase rotations. It is this $U(1)$, not its individual subgroups $\Z_n$, which has a chance to be the realizable symmetry group of this potential.


Therefore, if one focuses on the scalar sector of the model and aims at classifying the possible symmetries 
which can be implemented in the scalar potential, it is natural to restrict one's choice to the realizable symmetry groups.

In this Chapter we introduce the group of reparametrization transformations for NHDM. Next we develop the strategy to find all unitary Abelian realizable symmetry groups in NHDM. Using this strategy we present the full list of unitary Abelian realizable symmetries in 3HDM and 4HDM. In the attempt to try to derive the full list of unitary Abelian realizable symmetries in NHDM we prove several statements for a general NHDM potential. The results of this Chapter is published in \cite{Abelian2012}.

\section{Scalar sector of the $N$-Higgs doublet model}\label{section-symmetries}

\subsection{The group of reparametrization transformations}

When discussing symmetries of the potential, we focus on the reparametrization transformations, 
which are non-degenerate linear transformations mixing different doublets $\phi_a$ but keeping invariant the kinetic term (which includes interaction of the Higgs fields with the gauge sector of the model). Alternatively, they can be defined as norm-preserving transformations of doublets that do not change the intradoublet structure. 

Keeping the kinetic term invariant (and preserving charge conservation) requires conservation of the norm $(\phi_a^\dagger \phi_a)$, which results in the reparametrization transformation to be unitary (a Higgs-family transformation) or antiunitary (a generalized $CP$-transformation): 
\be
U: \quad \phi_a \mapsto U_{ab}\phi_b\qquad \mbox{or} \qquad U_{CP} = U \cdot J:\quad \phi_a \mapsto U_{ab}\phi^\dagger_b\,,
\ee
with a unitary matrix $U_{ab}$. The antiunitary transformations $U_{CP} = U\cdot J$, will be studied in Chapter \ref{chapter6}.

In this Chapter we focus on the unitary transformations $U$. A priori, such transformations form the group $U(N)$. However,
the overall phase factor multiplication is already taken into account by the $U(1)_Y$ from the gauge group.
This leaves us with the $SU(N)$ group of reparametrization transformations. 
\be 
\overbrace{ \left(\begin{array}{ccc}
e^{i\alpha_1} &   &    \\
  & \cdots &   \\
  &  &  e^{i\alpha_n}    \\
\end{array}
\right)}^{U(N)} 
\simeq 
 \overbrace{ \left(  \underbrace{ \begin{array}{ccc}
 e^{i\beta_1}  &   &   \\
  & \cdots &    \\
  &   & e^{i\beta_n}   \\ 
\end{array} }_{det=1}
 \right)}^{SU(N)} 
\cdot
 \overbrace{ e^{i\beta_0} }^{U(1)}
\ee
where $\beta_0=\frac{\Sigma \alpha_i}{n}$, $\beta_i=\alpha_i - \beta_0$ and $\Sigma \beta_n=0$.

Then group $SU(N)$ has a non-trivial center $Z(SU(N))= \Z_N$ generated by the diagonal matrix $\exp(2\pi i/N)\cdot 1_N$, where $1_N$ is the identity matrix. Therefore, the group of {\em physically distinct} unitary reparametrization transformations $G$ is
\be
G= PSU(N) \simeq SU(N)/\Z_N\,.\label{Gu} 
\ee

\section{Finding Abelian groups in NHDM}\label{section-strategy}

In this Section we introduce the strategy to find all realizable Abelian subgroups of the NHDM potential for any $N$. This strategy can be outlined as follows: we first describe maximal Abelian subgroups of $SU(N)$, then we construct the maximal Abelian subgroups of $PSU(N)$. Next we explore the realizable subgroups of this maximal Abelian subgroup.

Before presenting the strategy of constructing the maximal Abelian groups in $SU(N)$, we start with an example which would facilitate understanding the main idea.

\subsection{Heuristic example}

Consider a given potential $V$ of the $N$ Higgs doublet model. We want to find phase rotations which
leave this potential invariant. These phase rotations form the group $[U(1)]^N $, the group of all diagonal unitary $N\times N$ matrices acting in the space of doublets. This group is a subgroup of the group of all possible phase rotations of doublets $U(N)$, $[U(1)]^N \subset U(N)$, and can be parametrized by $N$ parameters $\alpha_j \in [0,2\pi)$:
\be
\mbox{diag}[e^{i\alpha_1},\, e^{i\alpha_2},\, \dots ,\, e^{i\alpha_N}]
\label{maximaltorusUN}
\ee
The potential $V$ is a collection of $k$ monomial terms each of the form $(\phi^\dagger_a \phi_b)$
or $(\phi^\dagger_a \phi_b)(\phi^\dagger_c \phi_d)$.
Upon a generic phase rotation (\ref{maximaltorusUN}), each monomial term gains its own phase rotation.
For example, $(\phi^\dagger_a \phi_b)$ with $a\not = b$ gains the phase $\alpha_b - \alpha_a$,
$(\phi^\dagger_a \phi_b)(\phi^\dagger_a \phi_c)$ with $a, b, c$ all distinct gains the phase $\alpha_b +\alpha_c - 2\alpha_a$, etc.

So, each monomial gets a phase rotation which is a linear function of $\alpha$'s with integer coefficients:
$\sum_{j=1}^N m_j \alpha_j$. The vector of coefficients $m_j$ can be brought by permutation and overall sign change to one of the following forms:
\bea
(1,\,-1,\,0,\,\dots) \quad &:& \quad (\phi_a^\dagger \phi_b)  \nonumber \\
(2,\,-2,\,0,\,\dots) \quad &:& \quad (\phi_a^\dagger \phi_b)^2 \nonumber \\
(2,\,-1,\,-1,\,\dots) \quad &:& \quad (\phi_a^\dagger \phi_b)(\phi_a^\dagger \phi_c) \nonumber \\
(1,\,1,\,-1,\,-1,\,\dots) \quad &:& \quad (\phi_a^\dagger \phi_b)(\phi_c^\dagger \phi_d) 
\label{fourtypes}
\eea
Note that in all cases $\sum_j m_j = 0$, which simply means that each bilinear is $U(1)$-symmetric.

The phase transformation properties of a given monomial are fully described by its vector $m_j$.
The phase transformation properties of the potential $V$, which is a collection of $k$ monomials, is characterized
by $k$ vectors $m_{1,j},\, m_{2,j},\, \dots,\, m_{k,j}$, each $m_{i,j}$ being of one of the types in (\ref{fourtypes}).

If we want a monomial to be invariant under a given transformation defined by phases $\{\alpha_j\}$,
we require that $\sum_{j=1}^N m_j \alpha_j = 2 \pi n$ with some integer $n$. If we want the entire potential to be invariant under a given 
phase transformation, we require this for each individual monomial. In other words, we require that 
there exist $k$ integers $\{n_i\}$ such that the phases $\{\alpha_j\}$ satisfy the following system of linear equations:
\be
\sum_{j=1}^N m_{i,j}\alpha_j = 2 \pi n_i\,,\quad \mbox{for all $1 \le i \le k$} \label{systemUN}
\ee
Solving this system for $\{\alpha_j\}$ yields the phase rotations that leave the given potential $V$ invariant.

One class of solutions can be easily identified: if all $\alpha_j$ are equal, $\alpha_j = \alpha$, then (\ref{systemUN}) 
with $n_i = 0$ is satisfied for any $\alpha$. These solutions form the $U(1)$ subgroup inside $[U(1)]^N$ 
and simply reflect the fact that the potential is constructed from bilinears $(\phi^\dagger_a \phi_b)$.
These solutions become trivial when we pass from the $U(N)$ to the $SU(N)$ group of transformations.
However, there can exist additional solutions of (\ref{systemUN}). They form a group which remains non-trivial 
once we pass from $U(N)$ to $SU(N)$ and further to $PSU(N)$, which is the group of physically distinct 
Higgs-family reparametrization transformations.
It is these solutions that we are interested in.

In order to find these solutions, we note that a matrix with integer entries can be ''diagonalized"
by a sequence of elementary operations on its rows or columns: permutation, sign change, and addition
of one row (column) to another row (column).
''Diagonalization" for a non-square matrix means that the only entries $m_{i,j}$ that can 
remain non-zero are at $i = j$. After that, the system splits into $k$ equations on $N$ phases of the generic form
\be
m_{i,i}\alpha'_i = 2\pi  n'_i\,, \quad  \alpha'_i \in [0,2\pi)\,,\quad n'_i \in \Z 
\ee
with non-negative integers $m_{i,i}$.

\begin{itemize}
\item
If $m_{i,i}=0$, this equation has a solution for any $\alpha_i$; the $i$-th equation gives a factor $U(1)$ to the 
symmetry group of the potential.
\item
If $m_{i,i} =1$, then this equation has no non-trivial solution, and the $i$-th equation does not contribute to the symmetry group.
\item
If $m_{i,i} =d_i > 1$, then this equation has $d_i$ solutions which are multiples of $\alpha_i = 2\pi/d_i$, 
and the $i$-th equation contributes the factor $\Z_{d_i}$ to the symmetry group.
\end{itemize}
The full symmetry group of phase rotations is then constructed as direct product of these factors.

Thus, the task reduces to studying which diagonal values of the matrix $m_{i,j}$ can arise in a model with $N$ doublets. 
For small values of $N$ this task can be solved explicitly,
while for general $N$ one must rely upon subtle properties of $m_{i,j}$ which stem from (\ref{fourtypes}).

\subsection{Maximal Abelian subgroups of $PSU(N)$}

In the previous Section we outlined the main idea of the strategy. However, we worked there in the group $U(N)$, while the reparametrization group is $PSU(N)$.

We start by reminding the reader that a maximal Abelian subgroup of $G$ (\ref{Gu}) is an Abelian group that is not contained in a larger Abelian subgroup of $G$. In principal, there can be several maximal Abelian subgroups in $G$, any subgroup of $G$ must be either a maximal one, or lie inside a maximal one. Therefore, we first need to identify all maximal Abelian subgroups of $G$, and then study their realizable subgroups.


In $SU(N)$, all maximal Abelian subgroups are the so called {\em maximal tori} \cite{Abelian2012},
\be
[U(1)]^{N-1} = U(1)\times U(1) \times \cdots \times U(1)
\label{maximaltorus1}
\ee

All such maximal tori are conjugate to each other, which means with two given maximal tori $T_1$ and $T_2$ there exist $g \in SU(N)$ such that $g^{-1}T_1g = T_2$. Therefore, without loss of generality one could pick up a specific maximal torus, for example, the one that is represented by phase rotations of individual doublets
\be
\mbox{diag}[e^{i\alpha_1},\, e^{i\alpha_2},\, \dots ,\, e^{i\alpha_{N-1}},\, e^{-i\sum\alpha_i}]
\label{maximaltorus2}
\ee
and study its subgroups.

The analysis then proceeds as we explained in the previous Section, with an additional condition that all $\alpha$'s sum to zero, $\left(\sum\alpha_N=0 \right)$.

However, the group of distinct reparametrization transformations is $PSU(N)$, which has a richer structure. It is proved that there are two sorts of maximal Abelian groups in $PSU(N)$ \cite{Abelian2012}:
\begin{itemize}
\item 
Maximal tori, which will be constructed in Section \ref{maximal-in-psun}.
\item
Certain finite Abelian groups which are not subgroups of maximal tori and should be treated separately.
\end{itemize}



\subsection{Maximal tori in $PSU(N)$}\label{maximal-in-psun}

Here we explicitly construct the maximal torus in $PSU(N)$. We first introduce some convenient notation. A diagonal unitary matrix acting in the space of Higgs doublets 
and performing phase rotations of individual doublets, such as (\ref{maximaltorus2}),
will be written as a vector of phases:
\be
\left(\alpha_1,\,\alpha_2,\,\dots,\,\alpha_{N-1}, -\sum\alpha_i\right)
\ee
In addition, if $M_1,\ldots,M_k$ are subsets of a group $G$, then $\langle M_1,\ldots,M_k\rangle$ denotes the subgroup generated by $M_1\cup\ldots\cup M_k$ of $G$.
Then we construct a maximal torus in $SU(N)$
\be
T_0 = \langle U(1)_1,U(1)_2 , \cdots , U(1)_{N-1}\rangle \nonumber 
\ee
where
\bea
U(1)_1 & = & \alpha_1(-1,\, 1,\, 0,\, 0,\, \dots,\, 0) \nonumber\\
U(1)_2 & = & \alpha_2(-2,\, 1,\, 1,\, 0,\, \dots,\, 0) \nonumber\\
U(1)_3 & = & \alpha_3(-3,\, 1,\, 1,\, 1,\, \dots,\, 0) \nonumber\\
\vdots &  & \vdots \nonumber\\
U(1)_{N-1} & = & \alpha_{N-1}(-N+1, \, 1,\, 1,\, 1,\, \dots,\, 1) 
\label{groupsUi}
\eea
with all $\alpha_i \in [0,2\pi)$.
 
For every $j$ we have $\langle U(1)_1,\ldots,U(1)_{j-1}\rangle\cap
U(1)_j=\{e\}$. Therefore:
\be 
T_0 = U(1)_1 \times U(1)_2 \times \cdots \times U(1)_{N-1} \label{T0}
\ee
In particular, any element $u \in SU(N)$ can be uniquely written as
\be
u = u_1(\alpha_1) u_2(\alpha_2) \cdots u_{N-1}(\alpha_{N-1})\,,\quad u_i \in U(1)_i
\ee
Moreover, the center $Z(SU(N))$ is contained in the last group and is generated by $\alpha_{N-1}=2\pi/N$.
One can therefore introduce
\be 
\overline{U(1)}_{N-1} = U(1)_{N-1}/Z(SU(N)) 
\ee
which can be parametrized as
\be
\overline{U(1)}_{N-1} = \alpha_{N-1}\left(-{N-1 \over N}, \, {1\over N},\, \dots,\, {1 \over N}\right)
\label{def-of-U}
\ee
where $\alpha_{N-1} \in [0,2\pi)$.
Therefore, the maximal torus in $PSU(N)$ is written as
\be
T = U(1)_1\times U(1)_2 \times \cdots \times \overline{U(1)}_{N-1} \label{maximal-torus-PSUN}
\ee


\subsection{Identifying the symmetries of the potential}\label{Identify}

Next we study which subgroups of the maximal torus $T$ in (\ref{maximal-torus-PSUN}) can be realizable in the scalar sector of NHDM.

We start from the most general $T$-symmetric potential:
\be
V(T) = - \sum_a m_a^2(\phi_a^\dagger \phi_a) + \sum_{a,b} \lambda_{ab} (\phi_a^\dagger \phi_a)(\phi_b^\dagger \phi_b)
+ \sum_{a \not = b} \lambda'_{ab} (\phi_a^\dagger \phi_b)(\phi_b^\dagger \phi_a) \label{Tsymmetric}
\ee
 
Each term in this potential transforms trivially under the entire $T$. The important fact is
that a sufficiently general potential of this form has no other unitary symmetry. In fact, when we start from the $T$-symmetric potential (\ref{Tsymmetric})
and add more terms, we will never generate any new unitary symmetry that was not already present in $T$ \cite{Abelian2012}. This is the crucial step in making sure that the groups described below are realizable. Our task now is to find which subgroups of $T$ can be obtained in this way.


Consider a bilinear $\phi_a^\dagger \phi_b$ where $a\not = b$. It gains a phase change under $T$ (\ref{maximal-torus-PSUN}) 
which linearly depends on the angles $\alpha_i$:
\be
\phi_a^\dagger \phi_b \to \exp[i(p_{ab}\alpha_1 + q_{ab}\alpha_2 + \dots + t_{ab}\alpha_{N-1})]\cdot \phi_a^\dagger \phi_b \label{pq-generic-NHDM}
\ee
with some integer coefficients $p_{ab},\, q_{ab},\, \dots,\, t_{ab}$.
Note that all coefficients are antisymmetric in their indices: $p_{ba}= - p_{ab}$.
These coefficients can be represented by real antisymmetric matrices with integer values, or graphically,
as labels of the edges of $N-1$ oriented graphs, one for each $U(1)$ group.
Each such graph has $N$ vertices, corresponding to doublets $\phi_a$; all vertices are joined with arbitrarily oriented edged, orientation indicated by an arrow. An edge oriented from $\phi_b$ towards $\phi_a$ (edge $b \to a$) is associated with the bilinear $\phi_a^\dagger\phi_b$ and is labelled by $p_{ba}$ in the first graph, $q_{ba}$ in the second graph, etc.
Examples of these graphs are shown for 3HDM and 4HDM in Section \ref{section-3HDM-4HDM}.

The Higgs potential is a sum of monomial terms which are linear or quadratic in $\phi_a^\dagger \phi_b$.
Consider one such term and calculate its coefficients $p,\dots, t$.
Let us first focus on how this term depends on any single $U(1)_i$ subgroup of $T$.
There are two possibilities depending on the value of the $i$-th coefficient:
\begin{itemize}
\item
If the coefficient $k$ in front of $\alpha_i$ is zero, this terms is $U(1)_i$-symmetric.
\item
If the coefficient $k \not = 0$, then this term is symmetric under the $\Z_{|k|}$ subgroup of $U(1)_i$-group
generated by phase rotations by $2\pi/|k|$.
\end{itemize}
However, even if a given monomial happens to have finite symmetry groups with respect to each single $U(1)_i$,
its symmetry group under the entire $T$ is still continuous: for any set of coefficients one can adjust
angles $\alpha_i$ in such a way that $p_{ab}\alpha_1 + \dots + t_{ab}\alpha_{N-1} = 0$. Therefore, to make sure the potential is symmetric under a finite symmetry group, we must sum at least $N-1$ such monomials. 

When studying symmetries of a given term or a sum of terms, we cannot limit ourselves 
to individual $U(1)_i$ groups but must consider the full maximal torus. 

The strategy presented below guarantees that we find all possible realizable subgroups of the maximal torus $T$, both finite and infinite. 

Consider a Higgs potential $V$ which, in addition to the $T$-symmetric part (\ref{Tsymmetric}) contains $k\ge 1$ additional terms, with coefficients $p_1,\,q_1,\,\dots t_1$ to $p_k,\,q_k,\,\dots t_k$. 
This potential defines the following $(N-1)\times k$ matrix of coefficients:
\be
X(V) = \left(\begin{array}{cccc}
p_1 & q_1 & \cdots & t_1\\
p_2 & q_2 & \cdots & t_2\\
\vdots & \vdots && \vdots \\
p_k & q_k & \cdots & t_k
\end{array}
\right) = 
\left(\begin{array}{ccc}
m_{1,1} & \cdots & m_{1,N-1}\\[2.5mm]
\vdots & \vdots & \vdots \\[2.5mm]
m_{k,1} &  \cdots & m_{k,N-1}
\end{array}
\right)
\ee
Here the second form of the matrix agrees with the notation of (\ref{systemUN}).
The symmetry group of this potential can be derived from the set of non-trivial solutions for $\alpha_i$ of the following equations:
\be
X(V) \left(\begin{array}{c} \alpha_1 \\ \vdots \\ \alpha_{N-1} \end{array}\right) 
= \left(\begin{array}{c} 2\pi n_1 \\ \vdots \\2 \pi n_{k} \end{array}\right) \label{XVeq}
\ee
There are two major possibilities depending on the rank of this matrix.
\begin{itemize}
\item
{\bf Finite symmetry group}\\
If rank$X(V)=N-1$, then there is no non-trivial solution of the equation (\ref{XVeq}) with the trivial right-hand side (i.e. all $n_i = 0$).
Instead, there exists a unique solution for any non-trivial set of $n_i$, and all such solutions form the finite group of phase rotations
of the given potential.

To find which symmetry groups can be obtained in this way, 
we take exactly $N-1$ monomials, so that the matrix $X(V)$ becomes a square matrix with a non-zero determinant.
After diagonalizing the matrix $X(V)$ with integer entries [by swapping rows (columns), adding and subtracting rows (columns)], the matrix $X(V)$ becomes diag$(d_1,\dots, d_{N-1})$, where $d_i$ are positive non-zero integers. 
This matrix still defines the equation (\ref{XVeq}) for $\alpha'_i$, which are linear combinations of $\alpha_i$'s and
with $n_i^\prime \in \Z$. Therefore, the finite symmetry group of this matrix is $\Z_{d_1} \times \cdots \times \Z_{d_{N-1}}$
(where $\Z_1$ means no non-trivial symmetry).

Note also that each of the allowed manipulations conserves the absolute value of the determinant of $X(V)$.
Therefore, even before diagonalization one can calculate the order of the finite symmetry group as $|\det X(V)|$.

This derivation leads us to the strategy that identifies all finite subgroups of torus realizable as symmetry groups of the Higgs potential
in NHDM: write down all possible monomials with $N$-doublets, consider all possible subsets with exactly $N-1$ distinct monomials,
construct the matrix $X$ for this subset and find its symmetry group following the above scheme. 
Although this strategy is far from being optimal, its algorithmic nature allows it to be easily implemented in a machine code.

\item
{\bf Continuous symmetry group}\\
If $\mathrm{rank} X(V) < N-1$, so that $D=(N-1)-\mathrm{rank}X(V) > 0$, then there exists a $D$-dimensional subspace 
in the space of angles $\alpha_i$, which solves the equation (\ref{XVeq}) for the trivial right-hand side.
One can then focus on the orthogonal complement of this subspace, where no non-trivial solution 
of the equation is possible, and repeat the above strategy
to find the finite symmetry group $G_D$ in this subspace. The symmetry group of the potential 
is then $[U(1)]^D \times G_D$.
\end{itemize}

\be
X(V) = \left( \begin{array}{ccccccc}
d_1 &   &  & | &  & & \\
  & \cdots & & |  & & & \\
  &   & d_i & |  & & & \\
-& -& -& -& - &- & -\\ 
 &   &  & | & 0 & & \\
 &   &  & | &  & \cdots & \\
 &   &  & | &  & & 0  \\ 
\end{array}
\right)
\ee 


\section{Abelian symmetries in general NHDM}\label{section-NHDM}

The algorithm described in Section (\ref{section-strategy}) can be used to find all Abelian groups realizable as the symmetry groups of the Higgs potential for any $N$. We do not yet have the full list of finite Abelian groups for a generic $N$ presented in a compact form, although we put forth a conjecture concerning this issue, see Conjecture \ref{conjecture-list} below. However, several strong results can be proved about the order and possible structure of finite realizable subgroups of the maximal torus.

Throughout this Section we will often use $n:=N-1$. Also, whenever we mention in this Section a finite Abelian group
we actually imply a finite realizable subgroup of the maximal torus.

\subsection{Upper bound on the order of finite Abelian groups} 

It can be expected from the general construction that for any given $N$ there exists an upper bound on the order of
finite realizable subgroups of the maximal torus in NHDM. In this Section we prove the following theorem:
\begin{theorem}\label{theorem-order}
The exact upper bound on the order of the realizable finite subgroup of maximal torus in NHDM is
\be
|G| \le 2^{N-1}
\ee
\end{theorem}

Before presenting the proof, let us first develop some convenient tools.
First, with the choice of the maximal torus (\ref{maximal-torus-PSUN}), we construct $n=N-1$ bilinears
$(\phi_1^\dagger \phi_{i+1})$, $i=1,\dots,n$.
The vectors of coefficients $a_i = (p_i,\, q_i,\, \dots,\, t_i)$ defined in (\ref{pq-generic-NHDM}) can be easily written:
\bea
a_1 &=& (2,3,4,\dots,n,1) \nonumber\\
a_2 &=& (1,3,4,\dots,n,1) \nonumber\\
a_3 &=& (1,2,4,\dots,n,1) \nonumber\\
\vdots && \vdots \nonumber\\
a_n &=& (1,2,3,\dots,n-1,1) \label{a-matrix-def}
\eea
One can use these vectors to construct the $n\times n$ matrix $A$:
\be
A = \left(\begin{array}{c}
a_1\\
\vdots\\
a_n
\end{array}
\right)\,, \quad \det A = 1 \label{detaa}
\ee
From the unit determinant we can also conclude that after diagonalization the matrix $A$ becomes the unit matrix.

Now consider a bilinear $(\phi_i^\dagger \phi_j)$ with $i,j \not = 1$; its vector of coefficients can be represented as $a_{j-1} - a_{i-1}$, [phase $(\phi_i^\dagger \phi_j) =$ phase $(\phi_i^\dagger \phi_1)(\phi_1^\dagger \phi_j)$]. More generally, for any monomial $(\phi_i^\dagger \phi_j)(\phi_k^\dagger \phi_m)$ with any $i,j,k,m$, the vector of coefficients
has the form $a_{j-1} - a_{i-1} + a_{m-1} - a_{k-1}$, where $a_0$ is understood as zero.
This means that the vector of coefficients of any monomial can be represented as a linear combination of $a$'s
with coefficients $0$, $\pm 1$ and $\pm 2$:
\bea
(\phi_i^\dagger \phi_j) , (\phi_i^\dagger \phi_1)(\phi_1^\dagger \phi_j) \quad : \quad & a_{j-1} - a_{i-1} & \quad \rightarrow  \quad +1, -1  \nonumber\\ 
(\phi_1^\dagger \phi_j) \quad : \quad & a_{j-1} & \quad \rightarrow  \quad  +1 , 0  \nonumber\\ 
(\phi_1^\dagger \phi_j)^2 \quad : \quad & 2 a_{j-1} & \quad \rightarrow  \quad  +2 , 0  \nonumber\\ 
(\phi_1^\dagger \phi_i)(\phi_1^\dagger \phi_j) \quad : \quad & a_{j-1}+a_{i-1} & \quad \rightarrow  \quad  +1   \nonumber\\ 
(\phi_i^\dagger \phi_j)^2 \quad : \quad & 2a_{j-1} - 2a_{i-1} &\quad \rightarrow  \quad +2, -2 
\eea

Since the $X$-matrix is constructed from $n$ such vectors,  we can represent it as
\be
X_{ik} = c_{ij}A_{jk} \label{XcA}
\ee
The square $n \times n$ matrix $c_{ij}$ can contain rows only of the following nine types (up to permutation and overall sign change):
\bea
\mbox{type }1: &(\phi_i^\dagger \phi_1)(\phi_1^\dagger \phi_1)& :(1,\,0,\,\cdots,\,0) \nonumber\\
\mbox{type }2: &(\phi_1^\dagger \phi_i)^2 & :(2,\,0,\,\cdots,\,0) \nonumber\\
\mbox{type }3: &(\phi_1^\dagger \phi_i)(\phi_1^\dagger \phi_j) & :(1,\,1,\,0,\,\cdots,\,0) \nonumber\\
\mbox{type }4: &(\phi_1^\dagger \phi_i)(\phi_j^\dagger \phi_1)& :(1,\,-1,\,0,\,\cdots,\,0) \nonumber\\
\mbox{type }5: &(\phi_1^\dagger \phi_i)(\phi_j^\dagger \phi_i)& :(2,\,-1,\,0,\,\cdots,\,0) \label{nine-types}\\
\mbox{type }6: &(\phi_i^\dagger \phi_j)(\phi_k^\dagger \phi_l)& :(1,\, 1,\,-1,\,0,\,\cdots,\,0) \nonumber\\
\mbox{type }7: &(\phi_i^\dagger \phi_j)^2& :(2,\, -2,\,0,\,\cdots,\,0) \nonumber\\
\mbox{type }8: &(\phi_i^\dagger \phi_j)(\phi_i^\dagger \phi_k)& :(2,\, -1,\,-1,\,0,\,\cdots,\,0) \nonumber\\
\mbox{type }9: &(\phi_i^\dagger \phi_j)(\phi_k^\dagger \phi_l)& :(1,\, 1,\,-1,\,-1,\, 0,\,\cdots,\,0) \nonumber
\eea
It follows from (\ref{XcA}) and (\ref{detaa}) that $\det X = \det c\cdot \det A = \det c$.
Therefore, the order of any finite group is given by the module of determinant of $c$: $|G| = |\det c|$.

Let us also note two properties of the strings of type 1--9. Take any such string of length $n$, $x_{(n)} = (x_1,\,x_2,\,\dots,\,x_n)$,
which is obtained from (\ref{nine-types}) by an arbitrary permutation and possibly an overall sign flip.
Then any substring $x_{(n-1)} = (x_1,\,\dots,\,x_{k-1},\, x_{k+1},\,\dots,\,x_n)$ obtained by removing
an arbitrary element $x_k$ is also of type 1--9.
Moreover, the element removed can be added at any place, and still the string
$x'_{(n-1)} = (x_1,\,\dots,\,x_m+x_k,\,\dots,\,x_{k-1},\, x_{k+1},\,\dots,\,x_n)$
remains of type 1--9.
Both properties can be proved by direct inspection of all the strings.

Now we are ready to prove Theorem~\ref{theorem-order}.

\begin{proof}
We prove by induction. Suppose that for any $(n-1)\times (n-1)$ square matrix $D_{n-1}$
whose rows are strings of type 1--9, its determinant $d_{n-1} = \det D_{n-1}$ is limited by
$|d_{n-1}| \le 2^{n-1}$. Take now any $n\times n$ matrix $D_{n}$ constructed from the same
family of strings and compute its determinant $d_{n}$ by minor expansion over the first row,
with $d^{(1)}_{n-1}$, $d^{(2)}_{n-1}$, $\dots$, being the relevant minors.
The procedure then depends on what type the first row is.
\begin{enumerate}
\item
If the first row is of type 1 or 2, then $|d_{n}| \le 2 |d^{(1)}_{n-1}| \le 2^n$.
\item
If the first row is of type 3 or 4, then $|d_{n}| = | d^{(1)}_{n-1} \pm d^{(2)}_{n-1}| \le
| d^{(1)}_{n-1}| + |d^{(2)}_{n-1}|  \le 2^n$.
\item
If the first row is of type 5, then we permute the columns so that it becomes exactly as in (\ref{nine-types})
and then add the second column to the first one. Then, the first row becomes $(1,\,-1,\,0,\,\cdots,\,0)$.
The first minor does not change and contains rows such as $(x_2,\,x_3,\,\dots,\,x_n)$,
while the second minor contains rows $(x_1+x_2,\,x_3,\,\dots,\,x_n)$.
Due to the properties discussed above, these strings are also of type 1--9,
therefore the induction assumption applies to both minors.
We therefore conclude that
$|d_{n}| = | d^{(1)}_{n-1} + d^{(2)}_{n-1}| \le 2^n$.
\item
If the first row is of type 6, then repeat the same procedure with only one change, that is we add the second column
to the third. The first row becomes $(1,\,1,\,0,\,\dots,\,0)$, while
the other rows have generic form $(x_1,\,x_2,\,x_2+x_3,\, x_4,\, \dots)$.
The first minor contains rows of the form $(x_2,\,x_2+x_3,\, x_4,\, \dots)$, which is
equivalent to $(x_2,\,x_3,\, x_4,\, \dots)$, while the second minor contains
$(x_1,\,x_2+x_3,\, x_4,\, \dots)$. Both rows are of type 1--9, therefore
$|d_{n}| = | d^{(1)}_{n-1} - d^{(2)}_{n-1}| \le 2^n$.
\item
If the first row is of type 7, we follow the procedure described for type 5 and get
$|d_{n}| = 2| d^{(1)}_{n-1}| \le 2^n$.
\item
If the first row is of type 8, we add the second and the third columns to the first, so that
the first row becomes $(0,\,-1,\,-1,\,\dots,\,0)$, while the other rows are of the form
$(x_1+x_2+x_3,\,x_2,\,x_3,\, x_4,\,\dots)$.
The second minor then contains rows $(x_1+x_2+x_3,\,x_3,\,x_4,\, \dots)$ which are equivalent
to $(x_1+x_2,\,x_3,\,x_4,\, \dots)$, being of an allowed type.
The third minor contains $(x_1+x_2+x_3,\,x_2,\,x_4,\, \dots)$ equivalent
to $(x_1+x_3,\,x_2,\,x_4,\, \dots)$, again of an allowed type.
Therefore the induction assumption applies to both minors, and we conclude that
$|d_{n}| = | d^{(2)}_{n-1} - d^{(3)}_{n-1}| \le 2^n$.
\item
If the first row is of type 9, we add the third column to the first and the fourth column to the second.
The first row turns into $(0,\,0,\,-1,\,-1,\,\dots,\,0)$, while the other rows become of the form
 $(x_1+x_3,\,x_2+x_4,\,x_3,\,x_4,\, \dots)$.
The third minor is built of rows $(x_1+x_3,\,x_2+x_4,\,x_4,\, \dots)$, which are equivalent to
$(x_1+x_3,\,x_2,\,x_4,\, \dots)$, while the fourth minor is built of rows $(x_1+x_3,\,x_2+x_4,\,x_3,\, \dots)$ equivalent to
$(x_1,\,x_2+x_4,\,x_3,\, \dots)$. Both are of the allowed type, so the induction assumption
applies to both minors, and
$|d_{n}| = | d^{(4)}_{n-1} - d^{(3)}_{n-1}| \le 2^n$.
\end{enumerate}
Therefore, $|d_{n}| \le 2^n$ follows for any type of the first row, which completes the proof.
\end{proof}

\subsection{Cyclic groups and their products}

In this Section we prove two propositions which show that a rather broad class of finite Abelian groups 
are indeed realizable as the symmetry groups of the Higgs
potential in the NHDM. First we show which cyclic groups are realizable and then consider direct products
of cyclic groups. 

Before we proceed to the main results, we need to prove that whenever we have the decomposition (\ref{XcA}), 
it is indeed sufficient to diagonalize the matrix $c_{ij}$ instead of full matrix $X$
in order to reconstruct the finite symmetry group.
Let us recall that the matrix $X$ defines a system of linear equations on angles $\alpha_i$:
\be
X_{ij}\alpha_j = c_{ik}A_{kj} \alpha_j = 2\pi r_i 
\ee
This system remains the same upon
\begin{itemize}
\item
simultaneous sign change of the $k$-th column in $c$ and $k$-th row in $A$ 
\item
simultaneous exchange of two columns in $c$ and two rows in $A$ 
\item
simultaneous summation of two columns in $c$ and subtraction of two rows in $A$  
\item
sign changes of the $k$-th column in $A$ and $\alpha_k$ 
\item
exchange of columns $i$ and $j$ in $A$ and exchange $\alpha_i \leftrightarrow \alpha_j$ 
\item 
summation of two columns in $A$ and subtraction of two $\alpha$'s 
\item
sign changes of the $k$-th row in $c$ and of the integer parameters $r_k$ 
\item
exchange of two rows in $c$ and of two $r$'s 
\item
summation of two rows in $c$ and of two parameters $r$ 
\end{itemize}
It means that all allowed manipulations with integer matrices described earlier
can be used for $A$ and $c$.

The next step is to diagonalize $A$; the matrix $c$ becomes modified by a series of allowed transformations.
But diagonal $A$ is equal to the unit matrix. Therefore, we are left only with $c$ in the system of equations.
We then diagonalize $c$ and construct the symmetry group from its diagonal values. 

This means that in order to prove a given group is realizable, we simply need to give an example of matrix $c_{ij}$ constructed from rows of type 1-9 (\ref{nine-types}) , which yields the desired diagonal values after diagonalization.

Now we move to two propositions.

\begin{proposition}\label{prop-cyclic-NHDM}
The cyclic group $\Z_p$ is realizable for any positive integer $p \le 2^n$.
\end{proposition}

\begin{proof}We start with the following $n \times n$ matrix $c_{ij}$:
\be
c_{2^n} = 
\left(\begin{array}{ccccccc}
2 & -1 & 0 & 0 &\cdots & 0 & 0 \\
0 & 2 & -1 & 0 &\cdots & 0 & 0 \\
0 & 0 & 2 & -1 &\cdots & 0 & 0 \\
\vdots & \vdots & \vdots & \vdots & & \vdots & \vdots \\
0 & 0 & 0 & 0 &\cdots & 2 & -1 \\
0 & 0 & 0 & 0 & \cdots & 0 & 2
\end{array}\right) \label{c2n}
\ee
We could manipulate this matrix, by multiplying the last column by 2 and adding it to the $n-1$'th column, and repeating this procedure for other columns, we arrive at 
\be
\left(\begin{array}{ccccccc}
0 & -1 & 0 & 0 &\cdots & 0 & 0 \\
0 & 0 & -1 & 0 &\cdots & 0 & 0 \\
0 & 0 & 0 & -1 &\cdots & 0 & 0 \\
\vdots & \vdots & \vdots & \vdots & & \vdots & \vdots \\
0 & 0 & 0 & 0 &\cdots & 0 & -1 \\
2^n & 2^{n-1} & 2^{n-2} & 2^{n-3} & \cdots & 2^2 & 2
\end{array}\right) \nonumber
\ee
Now we could vanish the elements $2^{n-1}, \cdots, 2^2, 2$ in the last row, by adding other rows to it as many times as needed. Then we swap the rows $(-1)^n$ times to bring the matrix into a diagonal form, which also kills the minus sign in front of all the $1$'s;
\be 
c_{2^n} = \left(\begin{array}{ccccccc}
2^n & 0 & 0 &\cdots & 0 & 0 \\
0 &  1 & 0 &\cdots & 0 & 0 \\
0 & 0 & 1 &\cdots & 0 & 0 \\
\vdots & \vdots & \vdots & & \vdots & \vdots \\
0 & 0 & 0 & \cdots & 0 & 1
\end{array}\right)  \nonumber
\ee
which produces the finite group $\Z_{2^n}$.

Now, let us modify $c_{2^n}$ in (\ref{c2n}) by replacing the zero at the left lower corner by $-1$.
Then the same transformation leads us to the group $\Z_{2^n-1}$.

If, instead, we replace any zero in the first column by $-1$, $c_{k1} = -1$,
then the diagonalization procedure leads us to the group $\Z_p$ with $p = 2^n - 2^{n-k}$.

Now, generally, take any integer $0 \le q < 2^{n}$ and write it in the binary form.
This form uses at most $n$ digits. Write this binary form as a vector with $n$ components
and subtract it from the first column of the matrix $c_{2^n}$. Then, after diagonalization,
we obtain the symmetry group $\Z_p$ with $p = 2^n - q$. This completes the proof.
\end{proof}

Just to illustrate this construction, for $n=5$ and $q = 23$ we write
$$
q = 23_{10} = 10111_2 = \left(\begin{array}{c}
1 \\
0 \\
1 \\
1 \\
1
\end{array}
\right); \quad c_{2^5} - q \equiv 
\left(\begin{array}{ccccccc}
1 & -1 & 0 & 0 & 0 \\
0 & 2 & -1 & 0 & 0 \\
-1 & 0 & 2 & -1 & 0 \\
-1 & 0 & 0 & 2 & -1 \\
-1 & 0 & 0 & 0 & 2
\end{array}\right)  
$$
Repeating the above procedure we obtain at the lowest left corner $2^5 - 2^4 - 2^2 - 2^1 - 2^0 = 32 - 23 = 9$,
which gives the $\Z_9$ group.

Now we show that not only cyclic groups but many of their products can be realized as symmetry groups of some potential.

\begin{proposition}\label{prop-many-cyclic-NHDM}
Let $n = \sum_{i=1}^k n_i$ be a partitioning of $n$ into a sum of non-negative integers $n_i$.
Then, the finite group 
\be
G = \Z_{p_1}\times \Z_{p_2}\times \cdots \times \Z_{p_k}\label{manyZn}
\ee
is realizable for any $0 < p_i \leq 2^{n_i}$.
\end{proposition}

\begin{proof}
Let us start again from the matrix (\ref{c2n}). For any partitioning of $n$ one can turn this matrix into a block-diagonal
matrix by replacing some of the $-1$'s by zeros. The matrix is then represented by the smaller square blocks of size
$n_1,\, n_2,\, \dots ,\, n_k$. The $i$-th block has exactly the form of (\ref{c2n}), with $n$ replaced by $n_i$.
Therefore, within this block one can encode any cyclic group $\Z_{p_i}$, with $0 < p_i \le 2^{n_i}$.
Since each block can be treated independently, we can encode any group of the form (\ref{manyZn}).


\be 
\left(\begin{array}{cccccccccc}
2 & -1 &| & 0 & 0 &\cdots & 0 & |& 0 &0 \\
0 & 2 &| & 0 & 0 &\cdots & 0 & |& 0 &0 \\
- & - &- & - & - &\cdots & - & -& - &- \\
0 & 0 &| & 2 & -1 &\cdots & 0 & |& 0 &0 \\
\vdots & \vdots & | & \vdots & \vdots &  & \vdots & | & \vdots & \vdots \\
0 & 0 &| & 0 & 0 &\cdots & 2 & |& 0  & 0 \\
- & - &- & - & - &\cdots & - & - & - &- \\
0 & 0 &| & 0 & 0 &\cdots &0  & |& 2 & -1 \\
0 & 0 &| & 0 & 0 & \cdots & 0 & |& 0 & 2
\end{array}\right)
\ee

\end{proof}

Let us note that this proposition does not yet exhaust all possible finite Abelian groups with order $\le 2^n$.
For example, for $n=5$, we can think of the group $\Z_5 \times \Z_5$ whose order is smaller than $2^5 = 32$.
However, there exists no partitioning of $5$ that would lead to this group by applying the proposition just proved.
At this moment, we cannot prove whether such groups are also realizable.
However, we conjecture that they are. Together with the earlier conjecture on the upper bound on the order of the
finite group, we can now propose the following conjecture:
\begin{conjecture}\label{conjecture-list}
Any finite Abelian group with order $\le 2^{N-1}$ is realizable in NHDM.
\end{conjecture}
If proven true, this conjecture will give the complete classification of realizable subgroups of the maximal torus in NHDM.

\section{Examples of 3HDM and 4HDM}\label{section-3HDM-4HDM}

In this Section we illustrate the general strategy presented above with the particular cases of 3HDM and 4HDM.
We give the full list of Abelian groups realizable as symmetry groups of the Higgs potential and show explicit examples of such potentials. 

\subsection{Unitary Abelian symmetries of the 3HDM potential}

In 3HDM the representative maximal torus $T \subset PSU(3)$ is parametrized as
\be
T = U(1)_1\times U(1)_2\,,\quad U(1)_1 = \alpha(-1,\,1,\,0)\,,\quad U(1)_2 = \beta\left(-{2 \over 3},\, {1 \over 3},\, {1 \over 3}\right)
\label{3HDM-maximaltorus}
\ee
where $\alpha,\beta \in [0,2\pi)$ (see \ref{def-of-U}). 
The most general Higgs potential symmetric under $T$ is
\bea
\label{3HDMpotential-Tsymmetric}
V &=& - m_1^2 (\phi_1^\dagger \phi_1) - m_2^2 (\phi_2^\dagger \phi_2) - m_3^2 (\phi_3^\dagger \phi_3)\\
&&+ \lambda_{11} (\phi_1^\dagger \phi_1)^2 + \lambda_{22} (\phi_2^\dagger \phi_2)^2 + \lambda_{33} (\phi_3^\dagger \phi_3)^2\nonumber\\
&&+ \lambda_{12} (\phi_1^\dagger \phi_1) (\phi_2^\dagger \phi_2) 
+ \lambda_{23} (\phi_2^\dagger \phi_2) (\phi_3^\dagger \phi_3)
+ \lambda_{13} (\phi_3^\dagger \phi_3) (\phi_1^\dagger \phi_1)
\nonumber\\
&&+ \lambda'_{12} |\phi_1^\dagger \phi_2|^2 + \lambda'_{23} |\phi_2^\dagger \phi_3|^2 + \lambda'_{13} |\phi_3^\dagger \phi_1|^2 \nonumber
\eea
There are six bilinear combinations of doublets transforming non-trivially under $T$ with coefficients
\be
\begin{tabular}{c | cc}
& $p$ & $q$ \\
\hline
$(\phi_2^\dagger \phi_1)$ & $-2$ & $-1$ \\
$(\phi_3^\dagger \phi_2)$ & $1$ & $0$ \\
$(\phi_1^\dagger \phi_3)$ & $1$ & $1$
\end{tabular}
\label{pq-explicit-3HDM}
\ee
and their conjugates whose coefficients $p$ and $q$ carry opposite signs with respect to (\ref{pq-explicit-3HDM}). 
These coefficients are shown graphically 
as labels of the edges of two oriented graphs shown in Figure (\ref{graph3HDM}).

\begin{figure} [ht]
\centering
\includegraphics[height=3cm]{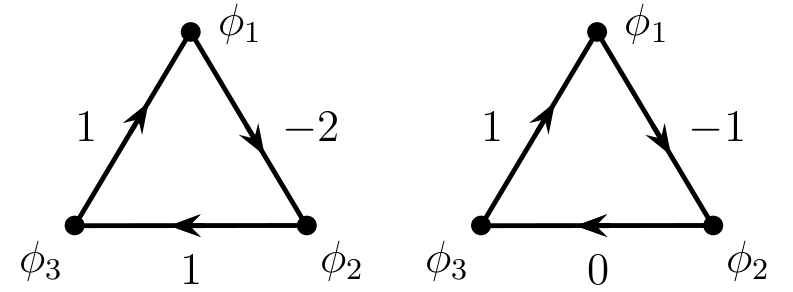}
\caption{Coefficients $p$ and $q$ as labels of triangles representing bilinears in 3HDM}
\label{graph3HDM}
\end{figure}

Before applying the general strategy described in the previous Section, let us 
check which finite symmetry groups can arise as subgroups of either $U(1)_1$ or $U(1)_2$ individually.

Consider first $U(1)_1$. The coefficient $p$ of any monomial can be obtained either by picking the labels from the first graph in Figure (\ref{graph3HDM}) directly or by summing two such labels, multiplied by $-1$ when needed.
In this way we can obtain any $|p|$ from $0$ to $4$. For a monomial with $p=0$, the symmetry group is the entire $U(1)_1$. 
For $|p|=1$, there is no non-trivial symmetry. For $|p| = 2,3,4$, we obtain the cyclic group $\Z_{|p|}$.
In each case it is straightforward to construct the monomials with a given symmetry;
for example, $(\phi_1^\dagger \phi_3)(\phi_2^\dagger \phi_3)$ and its conjugate are $U(1)_1$-symmetric,
while $(\phi_1^\dagger \phi_2)(\phi_1^\dagger \phi_3)$,  $(\phi_1^\dagger \phi_2)(\phi_3^\dagger \phi_2)$
and their conjugates are symmetric under the $\Z_3$-group with a generator 
\be
a = (\omega,\,1,\,\omega^{-1})\,,\quad \omega \equiv \exp\left(2\pi i/3\right) \label{Z3generator}
\ee
In the case of the group $U(1)_2$, the labels can sum up to $|q| = 0,1,2$. Thus, the largest finite group here is $\Z_2$.

As we mentioned in the previous Section, one cannot limit oneself to subgroups of individual $U(1)_i$ factors
or to direct products of such subgroups. 
In order to find all realizable groups, one has to write the full list of possible monomials and then calculate the symmetry
group of all distinct pairs (for $N=3$) of monomials. For example, if the two monomials are
$v_1 = \lambda_1(\phi_1^\dagger \phi_2)(\phi_1^\dagger \phi_3)$ and $v_2 = \lambda_2(\phi_2^\dagger \phi_1)(\phi_2^\dagger \phi_3)$,
then the matrix $X(v_1+v_2)$ has form
\be
X(v_1+v_2) = \left(\begin{array}{cc} 3 & 2 \\ -3 & -1 \end{array}\right) \to \left(\begin{array}{cc} 3 & 2 \\ 0 & 1 \end{array}\right) 
\to \left(\begin{array}{cc} 3 & 0 \\ 0 & 1 \end{array}\right)
\ee
and it produces the symmetry group $\Z_3$. The solution of the equation 
\be
X(v_1+v_2) \left(\begin{array}{c} \alpha \\ \beta \end{array}\right) = \left(\begin{array}{c} 2\pi n_1 \\ 2 \pi n_2 \end{array}\right)\label{Xv1v2eq}
\ee
yields $\alpha = 2\pi/3\cdot k$, $\beta=0$,
which implies the transformation matrix of the doublets (\ref{Z3generator}).

In 3HDM there are, up to complex conjugation, three bilinears and nine products of two bilinears transforming non-trivially under $T$. 

The full list of monomials in this case is as follows:
\begin{itemize}
\item
Linear: $(\phi_2^\dagger \phi_1), (\phi_3^\dagger \phi_2),(\phi_1^\dagger \phi_3) $
\item
Quadratic: $(\phi_2^\dagger \phi_1)^2, (\phi_3^\dagger \phi_2)^2,(\phi_1^\dagger \phi_3)^2, (\phi_2^\dagger \phi_1)(\phi_3^\dagger \phi_2),(\phi_3^\dagger \phi_2)(\phi_1^\dagger \phi_3),  (\phi_1^\dagger \phi_3)(\phi_2^\dagger \phi_1), \\ (\phi_2^\dagger \phi_1)(\phi_2^\dagger \phi_3),(\phi_3^\dagger \phi_2)(\phi_3^\dagger \phi_1), (\phi_1^\dagger \phi_3)(\phi_1^\dagger \phi_2)$ and also $(\phi_2^\dagger \phi_1)(\phi_i^\dagger \phi_i), (\phi_3^\dagger \phi_2)(\phi_i^\dagger \phi_i), \\ (\phi_1^\dagger \phi_3)(\phi_i^\dagger \phi_i)$ and their complex conjugates
\end{itemize}
It is a straightforward exercise to check all possible pairs of monomials (67 pairs); in fact, by using permutations of the doublets
the number of truly distinct cases is rather small. This procedure reveals just one additional finite group $\Z_2 \times \Z_2$,
which arises when at least two terms among $(\phi_1^\dagger\phi_2)^2$, $(\phi_2^\dagger\phi_3)^2$, $(\phi_3^\dagger\phi_1)^2$
are present. This group is simply the group of sign flips of individual doublets.

Thus, we arrive at the full list of subgroups of the maximal torus realizable as the symmetry groups of the Higgs potential in 3HDM:
\be
\Z_2,\quad \Z_3,\quad  \Z_4,\quad \Z_2\times \Z_2,\quad U(1),\quad U(1)\times \Z_2,\quad U(1)\times U(1)  \label{list3HDM}
\ee
Most of these groups were identified in \cite{Ferreira} in the search of "simple" symmetries of 3HDM scalar potential.
In that work a symmetry was characterized not by its group, as in our work, but by providing
a single unitary transformation $S$ and then reconstructing the pattern in the Higgs potential which arises
after requiring that it is $S$-symmetric.
In certain cases the authors of  \cite{Ferreira} found that the potential is symmetric under a larger group than $\langle S\rangle$,
in accordance with the notion of realizable symmetry discussed above.

The explicit correspondence between the seven symmetries $S_1,\dots, S_7$ of \cite{Ferreira} and the list (\ref{list3HDM})
is the following:
\bea
&&S_1 \to \Z_2\,,\quad S_2 \to U(1) \ \mbox{realized as}\ U(1)_2\,,\quad S_3 \to \Z_3 \label{correspondence}\\
&&S_4 \to \Z_4\,,\quad S_5 \to U(1) \ \mbox{realized as}\ U(1)_1\,,\quad S_6 \to U(1)\times \Z_2, 
\quad S_7 \to U(1)\times U(1) \nonumber
\eea
In addition to these symmetries, our list (\ref{list3HDM}) contains one more group $\Z_2\times \Z_2$,
which was not found in \cite{Ferreira} because it does not correspond to a "simple' symmetry.

Let us explicitly write the potentials which are symmetric under each group in (\ref{list3HDM}).
\begin{itemize}
\item
${\bf U(1)\times U(1) = T}$\\ 
The most general $T$-symmetric potential of 3HDM is given by (\ref{3HDMpotential-Tsymmetric}).
Every term in this potential is of the form $(\phi_i^{\dagger}\phi_i)$ or $(\phi_i^{\dagger}\phi_j)(\phi_j^{\dagger}\phi_i)$ from which the coefficients of $p$ and $q$ both become zero, therefore this potential is symmetric under $U(1)_1 \times U(1)_2$.
\item
${\bf U(1)}$\\ This group can be realized in two non-equivalent ways, which are conjugate either to $U(1)_1$ or $U(1)_2$ in (\ref{3HDM-maximaltorus}).

The general $U(1)_1$-invariant potential contains the following terms, in addition to (\ref{3HDMpotential-Tsymmetric}):
\be
\lambda_{1323}(\phi_1^\dagger\phi_3)(\phi_2^\dagger\phi_3) + h.c.\label{U11potential}
\ee

while the general $U(1)_2$-invariant potential contains
\bea
&&-m_{23}^2(\phi_2^\dagger\phi_3) + \left[\lambda_{1123} (\phi_1^\dagger\phi_1) + \lambda_{2223} (\phi_2^\dagger\phi_2) 
+ \lambda_{3323} (\phi_3^\dagger\phi_3)\right](\phi_2^\dagger\phi_3)\nonumber\\
&&+ \lambda_{2323}(\phi_2^\dagger\phi_3)^2 + \lambda_{2113}(\phi_2^\dagger\phi_1)(\phi_1^\dagger\phi_3) + h.c. \label{U12potential}
\eea

It must be stressed that these potentials are written for the specific convention of groups $U(1)_1$ and $U(1)_2$ used in 
(\ref{3HDM-maximaltorus}). In other bases, the explicit terms symmetric under a $U(1)_1$ or $U(1)_2$-type groups will look differently.
For example, the term $(\phi_2^\dagger\phi_1)(\phi_3^\dagger\phi_1)$ is symmetric under a $U(1)_1$-type
transformation with phases $(0,\alpha,-\alpha)$.

\item
${\bf U(1) \times \Z_2}$\\
Looking back at (\ref{U11potential}), one might be tempted to think that the true symmetry group of this term is not $U(1)$ but $U(1)\times\Z_2$,
because this term is also invariant under sign flip of $\phi_3$. However, inside $SU(3)$, a transformation with phases $(0,0,\pi)$
is equivalent to $(-\pi,\pi,0)$ which is already included in $U(1)_1$. 

A potential whose true symmetry group is $U(1) \times \Z_2$ is given by, in addition to (\ref{3HDMpotential-Tsymmetric}),
\be 
\lambda_{2323} (\phi_2^\dagger\phi_3)^2 + h.c.
\ee
This term is symmetric not only under the full $U(1)_2$,
but also under $(-\pi,\pi,0)$, which generates a $\Z_2$ subgroups inside $U(1)_1$.
\item
${\bf \Z_4}$\\
The potential symmetric under $\Z_4$ contains, in addition to (\ref{3HDMpotential-Tsymmetric}), the following terms:
\be
\lambda_{1323}(\phi_1^\dagger\phi_3)(\phi_2^\dagger\phi_3) + \lambda_{1212}(\phi_1^\dagger\phi_2)^2 + h.c. \label{3hdm-z4-pot}
\ee

As discussed before, the potential must contain at least $N-1=2$ distinct terms to be symmetric under a finite symmetry group. Again, this set of terms is specific for the particular choice of the $(U(1)_1,U(1)_2)$-basis on the maximal torus. This group is generated simply by $a_4=(i,-i,1)$.
\item
${\bf \Z_3}$\\
The potential symmetric under $\Z_3$ contains
\be
\lambda_{1232}(\phi_1^\dagger\phi_2)(\phi_3^\dagger\phi_2) + 
\lambda_{2313}(\phi_2^\dagger\phi_3)(\phi_1^\dagger\phi_3) + 
\lambda_{3121}(\phi_3^\dagger\phi_1)(\phi_2^\dagger\phi_1) + h.c. \label{z3-3hdm-pot}
\ee
In fact, any pair of the three terms is already sufficient to define a $\Z_3$-symmetric potential.
Note also that different $(U(1)_1,U(1)_2)$-bases on the maximal torus lead to the same $\Z_3$-group, because
inside $PSU(3)$ the following groups are isomorphic:
\be
\langle\left(1,\omega,\omega^2\right)\rangle \simeq
\langle\left(\omega,\omega^2,1\right)\rangle \simeq
\langle\left(\omega^2,1,\omega\right)\rangle\,,\quad \omega \equiv \exp(2\pi i/3)\,.
\ee
\item
${\bf \Z_2 \times \Z_2}$\\
This group can be realized simply as a group of independent sign flips of the three doublets (up to the overall sign flip), generated by $a_2=(-1,-1,1)$ and $a'_2=(1,-1,-1)$ .
Every term in the potential must contain each doublet in pairs. In addition to (\ref{3HDMpotential-Tsymmetric}),
the $\Z_2 \times \Z_2$-symmetric potential can contain
\be
\lambda_{1212}(\phi_1^\dagger\phi_2)^2 + 
\lambda_{2323}(\phi_2^\dagger\phi_3)^2 + 
\lambda_{3131}(\phi_3^\dagger\phi_1)^2 + h.c. \label{z2Xz2-3hdm-pot}
\ee
\item
${\bf \Z_2}$\\
This group can be realized, for example, as a group of sign flips of the third doublet, generated by $a_2=(-1,-1,1)$.
The potential can contain any term where $\phi_3$ appears in pairs.

\end{itemize}

In Figure (\ref{3HDMtree}) we show a pictorial description of subgroups of unitary realizable Abelian symmetries in 3HDM.

\begin{figure} [ht]
\centering
\includegraphics[height=7cm]{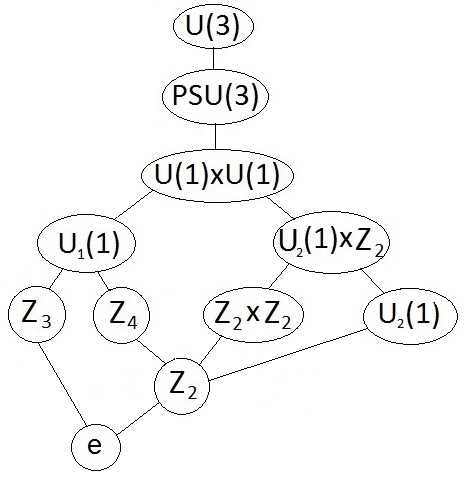}
\caption{Tree of subgroup of unitary Abelian symmetries in 3HDM}
\label{3HDMtree}
\end{figure}

\subsection{The $\Z_3 \times \Z_3$-group}

The only finite Abelian group that is not contained in any maximal torus in $PSU(3)$ is $\Z_3 \times \Z_3$.
Although there are many such groups inside $PSU(3)$, all of them are conjugate to each other. Thus, only one
representative case can be considered, and we choose the following two generators
\be
a = \left(\begin{array}{ccc} 1 & 0 & 0 \\ 0 & \omega & 0 \\ 0 & 0 & \omega^2 \end{array}\right),\quad
b = \left(\begin{array}{ccc} 0 & 1 & 0 \\ 0 & 0 & 1 \\ 1 & 0 & 0 \end{array}\right),\quad \omega = \exp\left({2\pi i \over 3}\right)
\ee
A generic potential that stays invariant under this group of transformations is
\bea
V & = &  - m^2 \left[(\phi_1^\dagger \phi_1)+ (\phi_2^\dagger \phi_2)+(\phi_3^\dagger \phi_3)\right]
+ \lambda_0 \left[(\phi_1^\dagger \phi_1)+ (\phi_2^\dagger \phi_2)+(\phi_3^\dagger \phi_3)\right]^2 \nonumber\\
&&+ \lambda_1 \left[(\phi_1^\dagger \phi_1)^2+ (\phi_2^\dagger \phi_2)^2+(\phi_3^\dagger \phi_3)^2\right]
+ \lambda_2 \left[|\phi_1^\dagger \phi_2|^2 + |\phi_2^\dagger \phi_3|^2 + |\phi_3^\dagger \phi_1|^2\right] \nonumber\\
&&+ \lambda_3 \left[(\phi_1^\dagger \phi_2)(\phi_1^\dagger \phi_3) + (\phi_2^\dagger \phi_3)(\phi_2^\dagger \phi_1) + (\phi_3^\dagger \phi_1)(\phi_3^\dagger \phi_2)\right]
+ h.c.\label{VZ3Z3}
\eea
with real $M_0$, $\lambda_0$, $\lambda_1$, $\lambda_2$ and complex $\lambda_3$.

The potential (\ref{VZ3Z3}) is not symmetric under any continuous Higgs-family transformation,
which can be proved, for example, using the adjoint representation of the bilinears described in Chapter \ref{chapter2}. 
However it is symmetric under the exchange of any two doublets, 
e.g. $\phi_1 \leftrightarrow \phi_2$, generating group $Z_2$. But this $Z_2$ group doesn't commute with $\Z_3 \times \Z_3$ group.
Therefore, the potential (\ref{VZ3Z3}) is symmetric under a larger group which is non-Abelian. We can conclude that the symmetry group $\Z_3 \times \Z_3$ is not realizable for 3HDM.

\subsection{Unitary Abelian symmetries of the 4HDM potential}\label{4HDM-examples}

In the case of four Higgs doublets the representative maximal torus in $PSU(4)$ is $T = U(1)_1\times U(1)_2 \times U(1)_3$,
where
\be
U(1)_1  =  \alpha(-1,\, 1,\, 0,\, 0)\,,\quad
U(1)_2  =  \beta(-2,\, 1,\, 1,\, 0)\,,\quad
U(1)_3  =  \gamma\left(-{3 \over 4}, \, {1 \over 4},\, {1 \over 4},\, {1 \over 4}\right)   \label{Phases-4HDM}
\ee
The phase rotations of a generic bilinear combination of doublets under $T$ is characterized by three integers $p,\,q,\,r$;
\be
(\phi_a^\dagger \phi_b) \to \exp[i(p_{ab}\alpha + q_{ab}\beta + r_{ab}\gamma)](\phi_a^\dagger \phi_b) 
\ee
There are twelve bilinear combinations of doublets transforming non-trivially under $T$ with coefficients:
\be
\begin{tabular}{c | ccc}
& $p$ & $q$ & $r$ \\
\hline
$(\phi_2^\dagger \phi_1)$ & $-2$ & $-3$  &$-1$ \\
$(\phi_3^\dagger \phi_2)$ & $1$ & $0$  &$0$ \\
$(\phi_4^\dagger \phi_3)$ & $0$ & $1$  &$0$ \\
$(\phi_4^\dagger \phi_1)$ & $-1$ & $-2$  &$-1$ \\
$(\phi_4^\dagger \phi_2)$ & $1$ & $1$  &$0$ \\
$(\phi_1^\dagger \phi_3)$ & $1$ & $3$  &$1$ \\
\end{tabular}
\label{pqr-explicit-4HDM}
\ee
and their conjugates whose coefficients $p$, $q$ and $r$ carry opposite signs.

These coefficients can again be represented graphically as labels of edges of three simplices shown in Figure (\ref{graph4HDM}).

\begin{figure} [ht]
\centering
\includegraphics[height=4cm]{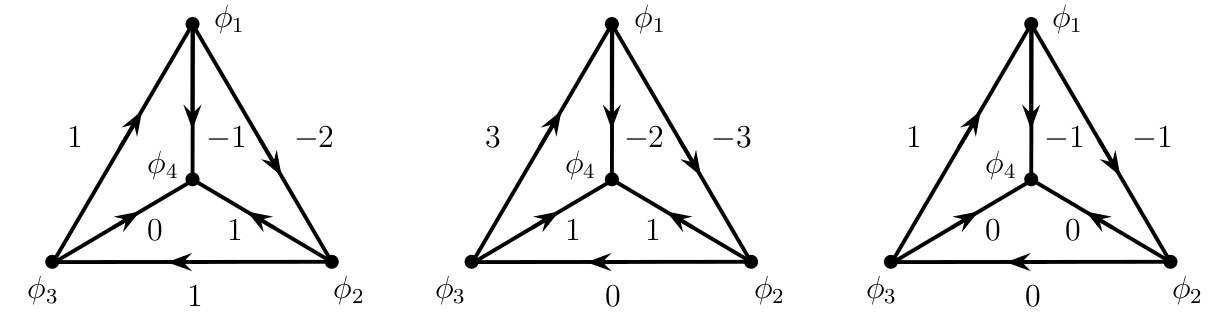}
\caption{Coefficients $p$, $q$ and $r$ as labels of simplices representing bilinears
in 4HDM}
\label{graph4HDM}
\end{figure}

With four doublets, one can construct (up to conjugation) 42 monomials transforming non-trivially under $T$.
The full list of monomials in this case is as follows:
\begin{itemize}
\item
Linear: $(\phi_1^\dagger \phi_2), (\phi_2^\dagger \phi_3), (\phi_3^\dagger \phi_4), (\phi_4^\dagger \phi_1), (\phi_2^\dagger \phi_4), (\phi_3^\dagger \phi_1)$
\item
Quadratic: $(\phi_1^\dagger \phi_2)^2, (\phi_2^\dagger \phi_3)^2, (\phi_3^\dagger \phi_4)^2, (\phi_4^\dagger \phi_1)^2, (\phi_2^\dagger \phi_4)^2, (\phi_3^\dagger \phi_1)^2, \\ 
(\phi_1^\dagger \phi_2)(\phi_2^\dagger \phi_3),  (\phi_1^\dagger \phi_2)(\phi_3^\dagger \phi_4),(\phi_1^\dagger \phi_2)(\phi_4^\dagger \phi_1),(\phi_1^\dagger \phi_2)(\phi_2^\dagger \phi_4),(\phi_1^\dagger \phi_2)(\phi_3^\dagger \phi_1), \\
 (\phi_2^\dagger \phi_3)(\phi_3^\dagger \phi_4), (\phi_2^\dagger \phi_3)(\phi_4^\dagger \phi_1) , (\phi_2^\dagger \phi_3)(\phi_2^\dagger \phi_4), (\phi_2^\dagger \phi_3)(\phi_3^\dagger \phi_1) , (\phi_3^\dagger \phi_4)(\phi_4^\dagger \phi_1), \\ (\phi_3^\dagger \phi_4)(\phi_2^\dagger \phi_4), (\phi_3^\dagger \phi_4)(\phi_3^\dagger \phi_1),  (\phi_4^\dagger \phi_1)(\phi_2^\dagger \phi_3),(\phi_4^\dagger \phi_1)(\phi_3^\dagger \phi_1), (\phi_2^\dagger \phi_1)(\phi_3^\dagger \phi_1),\\ $ along with their complex conjugates, and also  $(\phi_1^\dagger \phi_2)(\phi_i^\dagger \phi_i), (\phi_2^\dagger \phi_3)(\phi_i^\dagger \phi_i), \\ (\phi_3^\dagger \phi_4)(\phi_i^\dagger \phi_i), (\phi_4^\dagger \phi_1)(\phi_i^\dagger \phi_i), (\phi_2^\dagger \phi_4)(\phi_i^\dagger \phi_i), (\phi_3^\dagger \phi_1)(\phi_i^\dagger \phi_i)$ and their complex conjugates
\end{itemize}

When we search for realizable finite groups, we pick up all possible combinations of three distinct monomials $v_1$, $v_2$, $v_3$,
construct the $3 \times 3$ matrix $X(v_1+v_2+v_3)$, check that it is non-degenerate and then diagonalize it to obtain diag$(d_1,d_2,d_3)$.
The symmetry group is then $\Z_{d_1}\times \Z_{d_2}\times \Z_{d_3}$.
Although this brute force algorithm can give the full list of all realizable finite subgroups of the maximal torus in 4HDM, 
we can in fact apply the results of the Section \ref{section-NHDM} to find this list:
\be
\Z_k\ \mbox{with}\ k = 2,\, \dots ,\, 8, \qquad
\Z_2\times \Z_k\ \mbox{with}\ k = 2,\,3,\,4, \qquad
\Z_2 \times \Z_2 \times \Z_2 \label{list4HDMfinite}
\ee
In short, these are all finite Abelian groups of order $\leq 8$.
Perhaps, the most surprising of these groups is $\Z_7$, because it does not appear as a realizable subgroup of 
any single $U(1)_i$ factor [no two labels on Figure (\ref{graph4HDM}) sum up to 7].
Realizable continuous Abelian symmetry groups can also be found following the general strategy.

Let us explicitly write the potentials which are symmetric under each group in (\ref{list4HDMfinite}).
\begin{itemize}

\item
${\bf U(1)\times U(1) \times U(1) = T}$\\ 
The most general $T$-symmetric potential of 4HDM is:
\bea
V &=& - m_1^2 (\phi_1^\dagger \phi_1) - m_2^2 (\phi_2^\dagger \phi_2) - m_3^2 (\phi_3^\dagger \phi_3)- m_4^2 (\phi_4^\dagger \phi_4)  \nonumber \\
&+& \lambda_{11} (\phi_1^\dagger \phi_1)^2 + \lambda_{22} (\phi_2^\dagger \phi_2)^2 + \lambda_{33} (\phi_3^\dagger \phi_3)^2 + \lambda_{44} (\phi_4^\dagger \phi_4)^2 \nonumber \\
&+& \lambda_{12} (\phi_1^\dagger \phi_1) (\phi_2^\dagger \phi_2) 
+ \lambda_{13}  (\phi_1^\dagger \phi_1)(\phi_3^\dagger \phi_3)
+ \lambda_{14} (\phi_1^\dagger \phi_1)(\phi_4^\dagger \phi_4)  \nonumber \\
&+& \lambda_{23} (\phi_2^\dagger \phi_2) (\phi_3^\dagger \phi_3)
+ \lambda_{24} (\phi_2^\dagger \phi_2) (\phi_4^\dagger \phi_4)
+ \lambda_{34} (\phi_3^\dagger \phi_3) (\phi_4^\dagger \phi_4) \nonumber \\
&+& \lambda'_{12} |\phi_1^\dagger \phi_2|^2 + \lambda'_{13} |\phi_1^\dagger \phi_3|^2 + \lambda'_{14} |\phi_1^\dagger \phi_4|^2  \nonumber \\
&+& \lambda'_{23} |\phi_2^\dagger \phi_3|^2 +\lambda'_{24} |\phi_2^\dagger \phi_4|^2 +\lambda'_{34} |\phi_3^\dagger \phi_4|^2  
\label{4HDM-Tsymmetric-pot}
\eea
which are the terms of the form $(\phi_i^\dagger \phi_i)$ or $(\phi_i^\dagger \phi_j)(\phi_j^\dagger \phi_i)$.

\item
${\bf U(1)}$\\
The general $U(1)_1$-invariant potential contains, in addition to (\ref{4HDM-Tsymmetric-pot}), the following terms:
\bea 
&&-m_{43}^2(\phi_4^\dagger \phi_3)+\left[\lambda_{1143} (\phi_1^\dagger\phi_1) + \lambda_{2243} (\phi_2^\dagger\phi_2) 
+ \lambda_{3343} (\phi_3^\dagger\phi_3)\right](\phi_4^\dagger\phi_3) \nonumber \\
&&+ \lambda_{4343}(\phi_4^\dagger\phi_3)^2 +\lambda_{4142}(\phi_4^\dagger \phi_1)(\phi_4^\dagger \phi_2)+h.c.+\lambda_{4113}(\phi_4^\dagger \phi_1)(\phi_1^\dagger \phi_3)+h.c. \nonumber \\
&&+\lambda_{4132}(\phi_4^\dagger \phi_1)(\phi_3^\dagger \phi_2)+h.c.
\eea
The general $U(1)_2$-invariant potential contains, in addition to (\ref{4HDM-Tsymmetric-pot}), the following terms:
\bea 
&&-m_{32}^2(\phi_3^\dagger \phi_2)+\left[\lambda_{1132} (\phi_1^\dagger\phi_1) + \lambda_{2232} (\phi_2^\dagger\phi_2) 
+ \lambda_{3332} (\phi_3^\dagger\phi_3)\right](\phi_3^\dagger\phi_2)  \\ \nonumber
&&+ \lambda_{3232}(\phi_3^\dagger\phi_2)^2 +\lambda_{4324}(\phi_4^\dagger \phi_3)(\phi_2^\dagger \phi_4)+h.c.+\lambda_{2113}(\phi_2^\dagger \phi_1)(\phi_1^\dagger \phi_3)+h.c.  \nonumber
\eea
The general $U(1)_3$-invariant potential contains, in addition to (\ref{4HDM-Tsymmetric-pot}), the following terms:
\bea 
&&-m_{32}^2(\phi_3^\dagger \phi_2)+\left[\lambda_{1132} (\phi_1^\dagger\phi_1) + \lambda_{2232} (\phi_2^\dagger\phi_2) 
+ \lambda_{3332} (\phi_3^\dagger\phi_3)\right](\phi_3^\dagger\phi_2) \nonumber \\
&&+ \lambda_{3232}(\phi_3^\dagger\phi_2)^2 -m_{43}^2(\phi_4^\dagger \phi_3)+\left[\lambda_{1143} (\phi_1^\dagger\phi_1) + \lambda_{2243} (\phi_2^\dagger\phi_2) + \lambda_{3343} (\phi_3^\dagger\phi_3)\right](\phi_4^\dagger\phi_3) \nonumber \\
&&+ \lambda_{4343}(\phi_4^\dagger\phi_3)^2 -m_{42}^2(\phi_4^\dagger \phi_2)+ \left[\lambda_{1142} (\phi_1^\dagger\phi_1) + \lambda_{2242} (\phi_2^\dagger\phi_2) + \lambda_{3342} (\phi_3^\dagger\phi_3)\right](\phi_4^\dagger\phi_2) \nonumber \\
&&+ \lambda_{4242}(\phi_4^\dagger\phi_2)^2 +\lambda_{4113}(\phi_4^\dagger \phi_1)(\phi_1^\dagger \phi_3)+h.c.+\lambda_{1241}(\phi_1^\dagger \phi_2)(\phi_4^\dagger \phi_1)+h.c. \nonumber \\
&&+\lambda_{2113}(\phi_2^\dagger \phi_1)(\phi_1^\dagger \phi_3)+h.c. 
\eea

\item
${\bf Z_2}$\\
One could write terms that are $Z_2$-symmetric according to the labels of the edges in Figure (\ref{graph4HDM}). However, to make sure that a potential is only symmetric under a $Z_2$ group, we use the strategy in Section (\ref{section-NHDM}). We start with the following $3 \times 3$ matrix:
\be 
c_{2^3}=\left(\begin{array}{ccc}
2 & -1 & 0  \\
0 & 2 & -1  \\
0 & 0 & 2   
\end{array}\right)  \label{4HDM-matrix}
\ee
For any group $Z_p$ we need to write $q=2^{N-1}-p$ in a binary form, which uses at most $n=N-1$ digits. Now we write this binary form as a vector with $n$ components and subtract it form the first column of the matrix (\ref{4HDM-matrix}).
\be
q = 6_{10} = 110_2 = \left(\begin{array}{c}
1 \\
1 \\
0
\end{array}
\right),  \qquad c_{2^3} - q \equiv 
\left(\begin{array}{ccc}
1 & -1 & 0  \\
-1 & 2 & -1  \\
0 & 0 & 2 
\end{array}\right)   \label{c-matrix-z2-4hdm}
\ee
The $3 \times 3$ $c$-matrix (\ref{c-matrix-z2-4hdm}) represents the coefficients of the $A$-matrix (\ref{detaa}). Therefore the matrix (\ref{c-matrix-z2-4hdm}) could be written as:
\be 
\left(\begin{array}{c}
a_1 -a_2  \\
-a_1 +2a_2 -a_3 \\
2a_3
\end{array}\right)  \label{a-Z2}
\ee
Recalling that $a_1 = (\phi_1^\dagger \phi_2)$, $a_2 =(\phi_1^\dagger \phi_2)$, and $a_3 = (\phi_1^\dagger \phi_2)$, the terms corresponding to the elements of the matrix (\ref{a-Z2}) are as follows:
\bea
a_1 -a_2  &:& (\phi_1^\dagger \phi_2)(\phi_3^\dagger \phi_1)    \nonumber \\
-a_1 +2a_2 -a_3 &:&  (\phi_2^\dagger \phi_3)(\phi_4^\dagger \phi_3)  \nonumber \\
2a_3 &:&  (\phi_1^\dagger \phi_4)^2  
\eea
Therefore a $Z_2$-invariant potential contains, in addition to (\ref{4HDM-Tsymmetric-pot}), the following terms:
\be 
\lambda_{1231}(\phi_1^\dagger \phi_2)(\phi_3^\dagger \phi_1)+\lambda_{2343}(\phi_2^\dagger \phi_3)(\phi_4^\dagger \phi_3)+ \lambda_{1414}(\phi_1^\dagger \phi_4)(\phi_1^\dagger \phi_4) +h.c.  \label{Z2-invariant}
\ee

To find the phases of the transformation with which this $Z_2$ group is generated, we have to reproduce the $X$-matrix form the $c$-matrix (\ref{c-matrix-z2-4hdm}). 
\be 
X_{ik} = c_{ij} A_{jk} \quad \rightarrow \quad X= \left(\begin{array}{ccc}
1 & -1 & 0  \\
-1 & 2 & -1  \\
0 & 0 & 2 
\end{array}\right) \left(\begin{array}{c}
a_1  \\
a_2  \\
a_3 
\end{array}\right)
\ee
The $A$-matrix for 4HDM ($n=3$) is known from (\ref{a-matrix-def}). Therefore:
\be 
X= \left(\begin{array}{ccc}
1 & -1 & 0  \\
-1 & 2 & -1  \\
0 & 0 & 2 
\end{array}\right) \left(\begin{array}{ccc}
2 & 3 & 1  \\
1 & 3 & 1  \\
1 & 2 & 1 
\end{array}\right) = \left(\begin{array}{ccc}
1 & 0 & 0  \\
-1 & 1 & 0  \\
2 & 4 & 2 
\end{array}\right)
\ee
Equation (\ref{XVeq}) in this case is reduced to the form:
\be 
\left(\begin{array}{ccc}
1 & 0 & 0  \\
-1 & 1 & 0  \\
2 & 4 & 2 
\end{array}\right) \left(\begin{array}{c}
\alpha  \\
\beta  \\
\gamma 
\end{array}\right) = \left(\begin{array}{c}
2\pi n_1  \\
2\pi n_2  \\
2\pi n_3 
\end{array}\right)
\ee
Solving this set of equations for $\alpha$, $\beta$ and $\gamma$ results in $\alpha = \beta =0$ and $\gamma=\pi \cdot k$. Adding the phases in (\ref{Phases-4HDM}), with the known values of $\alpha$, $\beta$ and $\gamma$, one finds the phases of the $Z_2$ group generator:
\be 
a_2=\frac{\pi}{4}(-3,1,1,1)
\ee
One can check that $a^2$ lies in the center $Z(SU(4))$, and that the terms (\ref{Z2-invariant}) are invariant under $a$.

\item
${\bf Z_3}$\\
Using the same strategy used for the $Z_2$ group, one can write a $Z_3$-invariant potential with the following terms, in addition to (\ref{4HDM-Tsymmetric-pot}):
\be 
\lambda_{1231}(\phi_1^\dagger \phi_2)(\phi_3^\dagger \phi_1)+\lambda_{1343}(\phi_1^\dagger \phi_3)(\phi_4^\dagger \phi_3)+\lambda_{2414}(\phi_2^\dagger \phi_4)(\phi_1^\dagger \phi_4)+h.c.
\ee
This group is generated by the transformation with phases
\be 
a_3=\frac{2\pi}{12}(-5,3,3,-1)
\ee

\item
${\bf Z_4}$\\
We use the same strategy to find $Z_4$-invariant potential. In addition to (\ref{4HDM-Tsymmetric-pot}), this potential contains the following terms:
\be 
\lambda_{1231}(\phi_1^\dagger \phi_2)(\phi_3^\dagger \phi_1)+\lambda_{1343}(\phi_1^\dagger \phi_3)(\phi_4^\dagger \phi_3)+\lambda_{1414}(\phi_1^\dagger \phi_4)^2+h.c.
\ee
The phases of the $Z_4$ group generator are:
\be 
a_4=\frac{\pi}{2}(-2,1,1,0)
\ee

\item
${\bf Z_5}$\\
Following the same method, one finds a $Z_5$-invariant potential contains, in addition to (\ref{4HDM-Tsymmetric-pot}), the following terms:
\be 
\lambda_{1232}(\phi_1^\dagger \phi_2)(\phi_3^\dagger \phi_2)+\lambda_{2343}(\phi_2^\dagger \phi_3)(\phi_4^\dagger \phi_3)+\lambda_{2414}(\phi_2^\dagger \phi_4)(\phi_1^\dagger \phi_4)+h.c.
\ee
This group is generated by the transformation with phases
\be 
a_5=\frac{2\pi}{20}(-9,7,3,-1)
\ee

\item
${\bf Z_6}$\\
An example of a $Z_6$-invariant potential would have, in addition to (\ref{4HDM-Tsymmetric-pot}), the following terms:
\be 
\lambda_{1232}(\phi_1^\dagger \phi_2)(\phi_3^\dagger \phi_2)+\lambda_{2343}(\phi_2^\dagger \phi_3)(\phi_4^\dagger \phi_3)+\lambda_{1414}(\phi_1^\dagger \phi_4)^2+h.c.
\ee
The phases of the $Z_6$ group generator are:
\be 
a_6=\frac{\pi}{3}(-3,2,1,0)
\ee

\item
${\bf Z_7}$\\
An example of terms that possess this symmetry group is
\be 
\lambda_{1232}(\phi_1^\dagger \phi_2)(\phi_3^\dagger \phi_2)+\lambda_{1343}(\phi_1^\dagger \phi_3)(\phi_4^\dagger \phi_3)+\lambda_{2414}(\phi_2^\dagger \phi_4)(\phi_1^\dagger \phi_4)+h.c.
\ee
including (\ref{4HDM-Tsymmetric-pot}).
This group is generated by the transformation with phases
\be 
a_7=\frac{\pi}{14}(-21,19,3,-1)
\ee

\item
${\bf Z_8}$\\
The matrix (\ref{4HDM-matrix}) is already the appropriate $C$-matrix for a $Z_8$ symmetry. From this matrix one could write the terms for a potential with this symmetry:
 \be 
\lambda_{1232}(\phi_1^\dagger \phi_2)(\phi_3^\dagger \phi_2)+\lambda_{1343}(\phi_1^\dagger \phi_3)(\phi_4^\dagger \phi_3)+\lambda_{1414}(\phi_1^\dagger \phi_4)^2+h.c.
\ee
in addition to (\ref{4HDM-Tsymmetric-pot}).
The phases of the $Z_8$ group generator are:
\be 
a_8=\frac{\pi}{16}(-15, 21, -7, 1)
\ee

\item
${\bf Z_2 \times Z_2}$\\
To find a $Z_2 \times Z_2$-symmetric potential, we use the strategy described in preposition (\ref{prop-many-cyclic-NHDM}). Starting from the matrix (\ref{4HDM-matrix}), we break it into two block diagonal matrices by replacing the element $c_{12}$ by zero. 
\be
 \left(\begin{array}{cccc}
2 & | 0 & 0  \\
- & - & - & -\\
0 &|& 2 & -1  \\
0 &|& 0 & 2   
\end{array}\right) \nonumber
\ee
The first matrix represents a $Z_2$ group. To get a $Z_2$ group out of the second matrix one has to follow the same strategy as before; writing $2(=2^2-2)$ in a binary code, and subtract the vector of this binary code from the first column of the second matrix. Therefore having:
$$
q = 2_{10} = 10_2 = \left(\begin{array}{c}
1 \\
0
\end{array}
\right); \quad c_{2^2} - q \equiv 
\left(\begin{array}{cc}
1 & -1  \\
0 & 2 
\end{array}\right)
$$
Therefore the full $c$-matrix that represents a $Z_2 \times Z_2$ symmetry has the following form:
\be 
\left(\begin{array}{ccc}
2 & 0 & 0  \\
0 & 1 & -1  \\
0 & 0 & 2 
\end{array}\right)
\ee
From this matrix, one could easily derive the terms of a $Z_2 \times Z_2$-symmetric potential:
 \be 
\lambda_{1212}(\phi_1^\dagger \phi_2)^2+\lambda_{1341}(\phi_1^\dagger \phi_3)(\phi_4^\dagger \phi_1)+\lambda_{1414}(\phi_1^\dagger \phi_4)^2+h.c.
\ee
in addition to (\ref{4HDM-Tsymmetric-pot}).
Equation (\ref{XVeq}) in this case has two sets of solutions for $\alpha$, $\beta$ and $\gamma$. Therefore the $Z_2 \times Z_2$ group is generated by the two transformations with phases
\be 
a_{22}=\frac{\pi}{4}(-3,1,1,1) \quad , \quad a'_{22}=\frac{\pi}{4}(-5,-1,7,-1)
\ee

\item
${\bf Z_2 \times Z_3}$\\
Following the same strategy, we find the terms in a $ Z_2 \times Z_3$-symmetric potential are:
\be 
\lambda_{1212}(\phi_1^\dagger \phi_2)^2+\lambda_{1343}(\phi_1^\dagger \phi_3)(\phi_4^\dagger \phi_3)+\lambda_{3414}(\phi_3^\dagger \phi_4)(\phi_1^\dagger \phi_4)+h.c.
\ee
in addition to (\ref{4HDM-Tsymmetric-pot}).
To find the generators of this group we follow the same method as before, resulting in two generators;
\be 
a_{23}=\frac{2\pi}{3}(-3,2,1,0) \quad , \quad a'_{23}=\frac{\pi}{4}(-5,3,3,-1)
\ee
However $2$ and $3$ are mutually prime numbers, therefore $Z_2 \times Z_3 \equiv Z_6$, so this group only has one generator. 
Adding up $a_{23}$ and $a'_{23}$ one gets;
\be 
a''_{23}=\frac{\pi}{12}(9,-7,1,-3)
\ee
which is equivalent to the generator of $Z_6$ upon phase rotations.

\item
${\bf Z_2 \times Z_4}$\\
Using the same method one could write the terms of a $Z_2 \times Z_4$-symmetric potential:
\be 
\lambda_{1212}(\phi_1^\dagger \phi_2)^2+\lambda_{1343}(\phi_1^\dagger \phi_3)(\phi_4^\dagger \phi_3)+\lambda_{1414}(\phi_1^\dagger \phi_4)^2+h.c.
\ee
in addition to (\ref{4HDM-Tsymmetric-pot}).
This group is generated by the transformation with phases
\be 
a_{24}=\frac{\pi}{8}(-7,1,5,1) \quad , \quad a'_{24}=\frac{\pi}{4}(-5,-1,3,3)
\ee

\item
${\bf Z_2 \times Z_2 \times Z_2}$\\
To obtain this symmetry, one has to break the matrix (\ref{4HDM-matrix}) into three block diagonal matrices, by replacing all $-1$'s with zero. Therefore the simple matrix:
\be 
\left(\begin{array}{ccc}
2 & 0 & 0  \\
0 & 2 & 0  \\
0 & 0 & 2 
\end{array}\right)
\ee
results in a $ Z_2 \times Z_2 \times Z_2$-symmetric potential, with terms:
\be 
\lambda_{1212}(\phi_1^\dagger \phi_2)^2+\lambda_{13q3}(\phi_1^\dagger \phi_3)^2+\lambda_{1414}(\phi_1^\dagger \phi_4)^2+h.c.
\ee
in addition to (\ref{4HDM-Tsymmetric-pot}).

The phases of the $Z_2 \times Z_2 \times Z_2$ group generators are:
\be 
a_{222}=\frac{\pi}{4}(-3,1,1,1)\quad , \quad a'_{222}=\frac{\pi}{2}(-3,3,1,1)\quad , \quad a''_{222}=\pi(-1,1,0,0)
\ee
\end{itemize}


\chapter{Abelian groups containing anti-unitary transformations in NHDM}\label{chapter6}

In this Chapter we use the strategy to find all Abelian unitary symmetries in NHDM, introduced in Chapter \ref{chapter5}, and extend it to groups that include anti-unitary (generalized CP) transformations. We use the Higgs potential of the general NHDM and the notion of realizable symmetries from Chapter \ref{chapter5}.
First we introduce the group of reparametrization transformations for anti-unitary transformations. Then we describe the strategy to embed an Abelian group of unitary transformation into a larger Abelian group which includes anti-unitary transformations. Next as an example we apply this strategy to 3HDM for all group found in the list \ref{list3HDM}. The results of this Chapter are published in \cite{Abelian2012}.

\section{The group of reparametrization transformations}
                                                                                                                                                                                                                                                                                                                                                                                                                                                                                                                                                                                                                                                                                                                                                                                                                                                                                                                                                                                                                                                                                                                                                                                                                                                                                                                                                                                                                                                                                                                                                                                                                                                                                                                                                                             
To find the symmetries of the potential, we focus on the reparametrization transformations. A reparametrization transformation must be unitary (a Higgs-family transformation) or anti-unitary (a generalized $CP$-transformation): 
\be
U: \quad \phi_a \mapsto U_{ab}\phi_b\qquad \mbox{or} \qquad U_{CP} = U \cdot J:\quad \phi_a \mapsto U_{ab}\phi^\dagger_b 
\ee
with a unitary matrix $U_{ab}$. The transformation $J \equiv CP$ acts on doublets by complex conjugation and satisfies $J^2 = 1$.

The unitary transformations $U$ were studied in detail in Chapter \ref{chapter5}.

Consider now antiunitary transformations $U_{CP} = U\cdot J$, $U\in U(N)$.
One can define an action of $J$ on the group $U(N)$ by $ J\cdot U \cdot J$, where
\be
(J\cdot U \cdot J) \phi = JU (J \phi)= JU(\phi ^*)=J(U\phi^*)=(U\phi^*)^*= U^* \phi \rightarrow J\cdot U \cdot J=U^*  \label{JactiononUN}
\ee
The asterisk denotes complex conjugation.
This action leaves invariant both the overall phase subgroup $U(1)_Y$ and the center of the $SU(N)$ [$J$ maps a unitary transformation onto another unitary transformation in the same space].
Therefore, we can consider only such $U_{CP} = U\cdot J$ that $U \in PSU(N)$, where $PSU(N) = SU(N) / Z_n$ is the group of physically distinct unitary reparametrization transformation described in Chapter \ref{chapter5}.

Once this action is defined, we can represent all distinct reparametrization transformations as a semidirect product
\be
G = PSU(N) \rtimes \Z_2 \label{G}
\ee
where $J$ is the generator of $Z_2$, and $PSU(N)$ a normal subgroup of $G$.

\section{Finding anti-unitary transformations}\label{subsection-antiunitary-general}

As it was described in Chapter \ref{chapter5} all Abelian symmetries of NHDM are subgroups of $PSU(N)$. We studied which subgroups of the maximal torus $T$ (\ref{maximal-torus-PSUN}) can be realizable. Having found the list of realizable subgroups of torus $T$, we can ask whether these groups
can be extended to larger Abelian groups that would include not only unitary but also antiunitary transformations.
Here we describe the strategy that solves this problem.

We defined the action of the $CP$-transformation $J$ on $PSU(N)$ 
given by (\ref{JactiononUN}): $J\cdot U \cdot J = U^*$. This action can be restricted to the maximal torus $T$ (\ref{maximal-torus-PSUN}) and even further, to any subgroup $A \subset T$ (indeed, if $a \in A$, then $a^* \in A$). However, $J$ does not commute with a generic element of $T$ (meaning $J \cdot U \cdot J \neq U$), therefore the group $\langle A, J\rangle$ is not, in general, Abelian.

In order to embed $A$ in a larger Abelian group, we must find
an antiunitary transformation $J' = b\cdot J$ which commutes with any element $a \in A$:
\be
J'a(J')^{-1} = a\quad \Leftrightarrow \quad bJaJb^{-1}=a \quad \Leftrightarrow \quad ba^{-1}b^{-1}=a \quad \Leftrightarrow\quad  b=aba \label{Jprime}
\ee
The last expression here is a linear matrix equation for the matrix $b$ at any given $a$. 
Since $a$ is diagonal, $(aba)_{ij} = a_{ii}b_{ij}a_{jj}$, so whenever $a_{ii}a_{jj} \not = 1$,
one must set $b_{ij}$ to zero. 

If at least one matrix $b$ satisfying (\ref{Jprime}) is found, all the other matrices can be constructed with the help of the last form of this equation.
Suppose $b'=xb$ also satisfies this equation; then, 
\bea 
&&b'a^{-1}b'^{-1}=a  \nonumber\\ 
&&\Rightarrow (xb)a^{-1}(xb)^{-1}=a \Rightarrow xba^{-1}b^{-1}x^{-1}=a \Rightarrow xax^{-1}=a   \nonumber\\ 
&&\Rightarrow xa=ax  \Rightarrow a^{-1}xa=a^{-1}ax \Rightarrow  a^{-1}xa=x 
\eea 
Therefore, $x$ can be any unitary matrix from $PSU(N)$ commuting with $a$,
that is, $x$ centralizes the chosen Abelian group $A$.

\subsection{The strategy}

The strategy for embedding a given Abelian group $A \subset T$ into a larger Abelian group $A_{CP}$
which includes antiunitary transformations proceeds in five steps:
\begin{itemize}
\item
Find one $b$ which solves the matrix equation (\ref{Jprime}) for all $a \in A$.
\item
Find all $x \in PSU(N)$ which commute with each element of $A$, that is, construct the centralizer of $A$ in $PSU(N)$.
\item
Find restrictions placed on $x$'s by the requirement that product of two different antiunitary transformations $(x_1J')$ and $(x_2 J')$ belongs to $A$.
This procedure can result in several distinct groups $A_{CP}$. 
\item
Find which terms among the $A$-symmetric terms in the potential are also symmetric under some of $A_{CP}$;
drop terms which are not.
\item
Check that the resulting collection of terms is not symmetric under a larger group of unitary transformations,
that is, $A$ is still a realizable symmetry group of this collection of terms.
\end{itemize}

As an exercise, let us apply this strategy to the full maximal torus $T$. 
A generic $a \in T$ acting on a doublet $\phi_i$ generates a non-trivial phase rotation 
$\psi_i(\alpha_1,\dots,\alpha_{N-1})$ which can be reconstructed from (\ref{groupsUi}).
The equation (\ref{Jprime}) then becomes
\be
\left(\begin{array}{ccc} 
e^ {i\psi_1} &  &   \\
 & e^{i\psi_2} &  \\
 &  & \cdots
\end{array}\right) b_{ij} \left(\begin{array}{ccc} 
e^ {i\psi_1} &  &   \\
 & e^{i\psi_2} &  \\
 &  & \cdots
\end{array}\right) = b_{ij} \quad \rightarrow \quad 
e^{i(\psi_i+\psi_j)}b_{ij} = b_{ij} 
\ee
Since $\psi_i + \psi_j \not = 0$ for any $i, j$, the only solution to this equation is $b_{ij}=0$
for all $i$ and $j$.
This means that there is no $J'$ that would commute with every element of the torus $T$.
Thus, $T$ cannot be embedded into a larger Abelian group $T_{CP}$ that would include antiunitary transformations.


\section{Examples in 3HDM}

Applying the strategy just introduced, to 3HDM for all groups from the list (\ref{list3HDM}),
we obtain the following additional realizable Abelian groups:
\be
\Z_2^*\,,\quad \Z_2\times \Z_2^*\,, \quad \Z_2\times \Z_2\times \Z_2^*\,, \quad \Z_4^* 
\ee
where the asterisk indicates that the generator of the corresponding cyclic group is an anti-unitary transformation.

Before we explicitly describe embedding of each of the unitary Abelian groups (\ref{list3HDM}) into a larger Abelian group
that contains antiunitary transformations, let's just briefly explain why groups such as $\Z_6^*$, $\Z_8^*$,
$U(1)\times \Z_2^*$ do not appear in this list.

When we search for anti-unitary transformations commuting with the selected subgroup of $T$,
we can indeed construct such groups at the price of imposing certain restrictions on the coefficients of the $T$-symmetric part (\ref{3HDMpotential-Tsymmetric}). This makes the potential symmetric under a larger group of unitary transformations, which is non-Abelian. According to our definition, this means that the groups such as $\Z_6^*$ are not realizable, although they are subgroups of larger realizable non-Abelian symmetry groups.

\subsection{Embedding $U(1)_1\times U(1)_2$} 

We proved in the general case that the maximal torus cannot be extended to a larger Abelian group.
Let us repeat the argument in the specific case of 3HDM.
At the first step of the strategy of embedding a unitary Abelian group into an Abelian group that can contain
antiunitary transformations, we search for a matrix $b\in PSU(3)$ that satisfies $aba=b$ for all $a \in T$ where $T=U(1)_1 \times U(1)_2$ (\ref{3HDM-maximaltorus}).
A generic $a\in T$ is:
\be 
a =  \left(\begin{array}{ccc} 
e^ {-i\alpha - \frac{2i}{3} \beta } &  &   \\
 & e^{i\alpha + \frac{i}{3} \beta} &  \\
 &  & e^{\frac{i}{3} \beta}
\end{array}\right) \label{a}
\ee
Therefore $aba=b$ has the form:
\be
aba =  \left(\begin{array}{ccc} 
e^{-2i\alpha - 4i\beta/3}b_{11} & e^{-i\beta/3}b_{12} & e^{-i\alpha -i \beta/3}b_{13} \\
e^{-i \beta/3}b_{21} & e^{2i\alpha + 2i\beta/3}b_{22} & e^{i\alpha + 2i\beta/3}b_{23} \\
e^{-i\alpha -i \beta/3}b_{31} & e^{i\alpha + 2i\beta/3}b_{32} & e^{2i\beta/3}b_{33} 
\end{array}\right)\, = b \label{aba1}
\ee
Each element $(aba)_{ij}$ is equal to $b_{ij}$ up to a phase rotation $\psi_{ij}$:
\be
\psi_{ij} =  \left(\begin{array}{ccc} 
-2\alpha - 4\beta/3 & -\beta/3 & -\alpha - \beta/3 \\
- \beta/3 & 2\alpha + 2\beta/3 & \alpha + 2\beta/3 \\
-\alpha - \beta/3 & \alpha + 2\beta/3 & 2\beta/3 
\end{array}\right) \label{psi-ij}
\ee
In order for an element $b_{ij}$ to be non-zero, the corresponding phase rotation must be zero or multiple of $2\pi$. It is impossible to find such $b$ for arbitrary $\alpha$ and $\beta$. Therefore, $T = U(1)_1 \times U(1)_2$
cannot be embedded into a larger Abelian group.

\subsection{Embedding $U(1)$} 

There are two distinct types of $U(1)$ groups inside $T$: $U(1)_1$-type and $U(1)_2$-type. Consider first $U(1)_1$.

\subsubsection{$U(1)_1$ type} 

The phase rotations for the group $U(1)_1$ ($\beta=0$, arbitrary $\alpha$) allow $b$ to have non-zero entries $b_{12}$, $b_{21}$ and $b_{33}$ only. With $a$ of the form:
\be 
a =  \left(\begin{array}{ccc} 
e^ {-i\alpha} &  &   \\
 & e^{i\alpha} &  \\
 &  & 1
\end{array}\right)\,\label{a-u1}
\ee
The equation $aba=b$ gives:
\be
aba =  \left(\begin{array}{ccc} 
e^{-2i\alpha}b_{11} & b_{12} & e^{-i\alpha}b_{13} \\
b_{21} & e^{2i\alpha}b_{22} & e^{i\alpha}b_{23} \\
e^{-i\alpha}b_{31} & e^{i\alpha}b_{32} & b_{33} 
\end{array}\right)= b \quad \rightarrow \quad b = \left(\begin{array}{ccc} 
0 & b_{12} & 0 \\
b_{21} & 0 & 0 \\
0 & 0 & b_{33} 
\end{array}\right)\,
\ee

We are free to choose $b_{12} = b_{21} = 1$, $b_{33}=1$ (the fact that $\det b = -1$ instead of $1$ is inessential;
we can always multiply all the transformations by the overall $-1$ factor). Then, the transformation $J'$ which commutes with any $a \in U(1)_1$ is
\be
J' =  \left(\begin{array}{ccc} 
0 & 1 & 0 \\
1 & 0 & 0 \\
0 & 0 & 1 
\end{array}\right)\cdot J\,,\quad
aJ' = J'a \quad \forall a \in U(1)_1 
\label{J'-in-U1}
\ee
Next, we search for transformations $x \in PSU(3)$ such that $xJ'$ also commutes with any element $a \in U(1)_1$;
such transformations form the centralizer of $U(1)_1$.

From the equation $axa^{-1} = x$, analysed with the same technique as phase rotations, 
we find that $x$ must be diagonal. 
\be
axa^{-1} =  \left(\begin{array}{ccc} 
x_{11} & e^{-2i\alpha}x_{12} & e^{-i\alpha}x_{13} \\
e^{2i\alpha}x_{21} & x_{22} & e^{i\alpha}x_{23} \\
e^{i\alpha}x_{31} & e^{-i\alpha}x_{32} & x_{33} 
\end{array}\right)= x \quad \rightarrow \quad x = \left(\begin{array}{ccc} 
x_{11} & 0 & 0 \\
0 & x_{22} & 0 \\
0 & 0 & x_{33} 
\end{array}\right)\,
\ee
We also know that $x \in PSU(3)$, therefore $x$ must be of the following form:
\be 
x =  \left(\begin{array}{ccc} 
e^ {i\gamma _1} &  &   \\
 & e^{i\gamma _2} &  \\
 &  & e^{-i(\gamma _1 + \gamma _2)}
\end{array}\right)\,\label{x}
\ee
We close the group under product, by checking that the production of $xJ'$ and $x'J'$ stays inside $U(1)_1$: 
\be 
xJ' \cdot x'J' = x b \cdot J \cdot x' b \cdot J = x b \cdot x'^* b
\ee
Also:
\begin{small}
\be 
b x'^* b = \left(\begin{array}{ccc} 
0 & 1 & 0 \\
1 & 0 & 0 \\
0 & 0 & 1 
\end{array}\right) \left(\begin{array}{ccc} 
e^ {-i\gamma ' _1} &  &   \\
 & e^{-i\gamma ' _2} &  \\
 &  & e^{i(\gamma ' _1 + \gamma ' _2)} 
\end{array}\right) \left(\begin{array}{ccc} 
0 & 1 & 0 \\
1 & 0 & 0 \\
0 & 0 & 1 
\end{array}\right) = \left(\begin{array}{ccc} 
e^ {-i\gamma ' _2} &  &   \\
 & e^{-i\gamma ' _1} &  \\
 &  & e^{i(\gamma ' _1 + \gamma ' _2)} 
\end{array}\right) \nonumber
\ee
\end{small}
Therefore:
\be 
x b x'^* b = \left(\begin{array}{ccc} 
e^ {i\gamma _1} &  &   \\
 & e^{i\gamma _2} &  \\
 &  & e^{-i(\gamma _1 + \gamma _2)} 
\end{array}\right) 
\left(\begin{array}{ccc} 
e^ {-i\gamma ' _2} &  &   \\
 & e^{-i\gamma ' _1} &  \\
 &  & e^{i(\gamma ' _1 + \gamma ' _2)} 
\end{array}\right)
\ee
So, the vector of phases has the form
\be 
\left( \quad \gamma _1 - \gamma '_2 \quad , \quad \gamma _2 - \gamma '_1  \quad , \quad \gamma '_1 + \gamma '_2 - \gamma _1 - \gamma _2 \quad \right)
\ee
Which should stay inside $U(1)$, therefore is of the form $(-\alpha , \alpha , 0 )$. By introducing $ \xi$ and $ \gamma$, where 
\bea 
&&\gamma '_1 + \gamma '_2 = \gamma _1 + \gamma _2 = 2\xi  \nonumber\\
&&\gamma _1 = \gamma + \xi  \quad , \quad \gamma _2 = \gamma - \xi \nonumber\\
&&\gamma '_1 = \gamma '+ \xi  \quad , \quad \gamma '_2 = \gamma '- \xi  
\eea
One could write $x$ as 
\be 
x =  \left(\begin{array}{ccc} 
e^ {i(\gamma +\xi)} &  &   \\
 & e^{i(\xi - \gamma)} &  \\
 &  & e^{-2i(\xi)}
\end{array}\right)\ = \left(\begin{array}{ccc} 
e^{-i\gamma} & 0 & 0 \\
0 & e^{i\gamma} & 0 \\
0 & 0 & 1 
\end{array}\right)
\left(\begin{array}{ccc} 
e^{i\xi} & 0 & 0 \\
0 & e^{i\xi} & 0 \\
0 & 0 & e^{-2i\xi}
\end{array}\right)\, \label{xgamxi}
\ee
We see that $x$ belongs to a $U(1)$-type group $X_{\xi} $. Here $\xi \in [0,2\pi/3)$ is an arbitrary but fixed parameter specifying which $X_{\xi}$ group we take, while $\gamma$ is the running angle parametrizing the elements of this group. 
At this stage, any choice of $\xi$ is acceptable.
Thus, we have embedded $U(1)_1$ into an Abelian group generated by 
$\langle U(1)_1, J''_{\xi}\rangle \simeq U(1) \times \Z_2^*$;
\be 
X_{\xi} \cdot J'= \underbrace{\left(\begin{array}{ccc} 
e^{-i\gamma} & 0 & 0 \\
0 & e^{i\gamma} & 0 \\
0 & 0 & 1 
\end{array}\right)}_{U(1)}
\underbrace{\left(\begin{array}{ccc} 
e^{i\xi} & 0 & 0 \\
0 & e^{i\xi} & 0 \\
0 & 0 & e^{-2i\xi}
\end{array}\right) \cdot b \cdot J }_{J''_{\xi}}
\ee
Therefore
\be
J''_{\xi} = 
\left(\begin{array}{ccc} 
e^{i\xi} & 0 & 0 \\
0 & e^{i\xi} & 0 \\
0 & 0 & e^{-2i\xi}
\end{array}\right)\cdot J' = 
\left(\begin{array}{ccc} 
0 & e^{i\xi} & 0 \\
e^{i\xi} & 0 & 0 \\
0 & 0 & e^{-2i\xi}
\end{array}\right)\cdot J\,\label{J''xi}
\ee
Let us now see whether this group is realizable and which $\xi$ must be chosen.
The only $U(1)_1$ symmetric terms in the potential are (\ref{U11potential})
\be
\lambda (\phi_1^\dagger \phi_3)(\phi_2^\dagger \phi_3)  + \lambda^* (\phi_3^\dagger \phi_1)(\phi_3^\dagger \phi_2)
\label{extratermsU1}
\ee
Let's see how the sum of two terms transform under $J''_{\xi}$;
is the phase of $\lambda$.
\be 
J'' \phi = \left(\begin{array}{ccc} 
0 & e^{i\xi} & 0 \\
e^{i\xi} & 0 & 0 \\
0 & 0 & e^{-2i\xi}
\end{array} \right) \cdot J \cdot \left(\begin{array}{c}
\phi _1 \\
\phi _2 \\
\phi _3 
\end{array} \right)  = \left( \begin{array}{c}
e^{i\xi} \phi _2^\dagger \\
e^{i\xi}\phi _1^\dagger \\
e^{-2i\xi}\phi _3^\dagger  
\end{array} \right)
\ee
Therefore,
\bea 
&&\lambda (\phi_1^\dagger \phi_3)(\phi_2^\dagger \phi_3) \quad \underrightarrow{\mbox{under J}''} \quad |\lambda| e^{i\psi _ \lambda} e^{-6i\xi} (\phi_3^\dagger \phi_2)(\phi_3^\dagger \phi_1)    \nonumber \\
&&\lambda^* (\phi_3^\dagger \phi_1)(\phi_3^\dagger \phi_2) \quad \underrightarrow{\mbox{under J}''} \quad |\lambda| e^{-i\psi _ \lambda} e^{6i\xi} (\phi_1^\dagger \phi_3)(\phi_2^\dagger \phi_3)\nonumber \\
&& \quad \rightarrow \quad \psi_\lambda - 6\xi = -\psi_\lambda \quad \rightarrow \quad \xi = \psi_\lambda/3  \label{ui-invariant-under-j''}
\eea
Therefore the sum of these two terms are invariant under $J''_{\xi}$ provided $\xi = \psi_\lambda/3$, where $\psi_\lambda$ is the phase of $\lambda$.

This prescription uniquely specifies the group $X_\xi \subset PSU(3)$. 
Recall that a general $U(1)$ symmetric potential also contains the $T$-symmetric terms (\ref{3HDMpotential-Tsymmetric}). In order to guarantee that the $T$-symmetric terms are invariant under $J''$, we must set
\be
m_{1}^2=m_{2}^2\,,\quad \lambda_{11}=\lambda_{22}\,,\quad \lambda_{13} = \lambda_{23}\,,\quad \lambda'_{13} = \lambda'_{23} 
\label{U1-restrictions}
\ee
However upon this reduction of free parameters we observe that the resulting potential is also symmetric under under another $Z_2$ group generated by the unitary transformation
\be
d =  \left(\begin{array}{ccc} 
0 & e^{i\delta} & 0 \\
e^{-i\delta} & 0 & 0 \\
0 & 0 & 1 
\end{array}\right)
\ee
But this transformation does not commute with $U(1)_1$. Therefore, the true symmetry group of such potential is non-Abelian. According to our definition, we conclude that $U(1) \times \Z_2^*$ is not realizable in 3HDM.

\subsubsection{$U(1)_2$ type} 
Consider now $U(1)_2$: $\alpha=0$, arbitrary $\beta$. 
\be 
aba = \left(\begin{array}{ccc} 
e^{- 4i\beta/3}b_{11} & e^{-i\beta/3}b_{12} & e^{-i \beta/3}b_{13} \\
e^{-i \beta/3}b_{21} & e^{2i\beta/3}b_{22} & e^{2i\beta/3}b_{23} \\
e^{-i \beta/3}b_{31} & e^{2i\beta/3}b_{32} & e^{2i\beta/3}b_{33} 
\end{array}\right)\, = b \quad \rightarrow \quad b_{ij} = 0
\ee
No non-trivial matrix $b \in PSU(3)$
can satisfy $aba = b$. Therefore, $U(1)_2$ cannot be embedded into a larger Abelian group with antiunitary transformations.
Note also that selecting various diagonal $U(1)$ subgroups of $T$, for example by setting $\alpha+\beta/3 = 0$, will result just in another version of $U(1)_1$.

\subsection{Embedding $ U(1)\times \Z_2$}

The unitary $U(1)\times \Z_2$ emerges only when the continuous group is of the $U(1)_2$-type. Since $U(1)_2$ cannot be embedded into a larger Abelian group with antiunitary transformations, $U(1)\times \Z_2$ cannot be embedded into a larger realizable Abelian group either.

Therefore in 3HDM all realizable continuous Abelian groups are necessarily unitary and cannot contain antiunitary transformations.

\subsection{Embedding $\Z_4$}

The $\Z_4$ symmetry group with generator $a$ arises as a subgroup of $U(1)_1$. Therefore, $b$ and $J'$ can be chosen
as before, (\ref{J'-in-U1}), but the new conditions for $x$ should now be checked.

The vector of phases in this case is:
\be 
\left( \quad \gamma _1 - \gamma '_2 \quad , \quad \gamma _2 - \gamma '_1  \quad , \quad \gamma '_1 + \gamma '_2 - \gamma _1 - \gamma _2 \quad \right)
\ee
which should be a subgroup of $Z_4$. Therefore it has to agree with $ (-\frac{\pi}{2}k , \frac{\pi}{2}k , 0 )$, which could be satisfied with $\gamma_1 =\frac{\pi}{2}n$ or $\gamma_1 =\frac{\pi}{2}n + \frac{\pi}{4} $. Thus $x$ should have one of the following forms:
\be 
x = \left(\begin{array}{ccc} 
e^{-i\frac{\pi}{2}k +i\xi} & 0 & 0 \\
0 & e^{i\frac{\pi}{2}k +i\xi} & 0 \\
0 & 0 & e^{-2i\xi}
\end{array}\right) \quad \mbox{or} \quad \left(\begin{array}{ccc} 
e^{-i\frac{\pi}{2}k - i\frac{\pi}{4} +i\xi} & 0 & 0 \\
0 & e^{i\frac{\pi}{2}k + i\frac{\pi}{4}+i\xi} & 0 \\
0 & 0 & e^{-2i\xi}
\end{array}\right)
\ee
Correspondingly, two possibilities for embedding the $\Z_4$ group arise: $\Z_4 \times \Z_2^*$ and $\Z_8^*$.

Consider first the $\Z_4 \times \Z_2^*$ embedding.

\subsubsection{$\Z_4 \times \Z_2^*$ group} 
\be 
X_{\xi} \cdot J'= \underbrace{\left(\begin{array}{ccc} 
e^{-i\frac{\pi}{2}k} & 0 & 0 \\
0 & e^{i\frac{\pi}{2}k} & 0 \\
0 & 0 & 1 
\end{array}\right)}_{Z_4}
\underbrace{\left(\begin{array}{ccc} 
e^{i\xi} & 0 & 0 \\
0 & e^{i\xi} & 0 \\
0 & 0 & e^{-2i\xi}
\end{array}\right) \cdot b \cdot J }_{J''_{\xi}}
\ee
The $\Z_2^*$ subgroup is generated by $J''$;
\be
J''_{\xi} = 
\left(\begin{array}{ccc} 
e^{i\xi} & 0 & 0 \\
0 & e^{i\xi} & 0 \\
0 & 0 & e^{-2i\xi}
\end{array}\right)\cdot J' = 
\left(\begin{array}{ccc} 
0 & e^{i\xi} & 0 \\
e^{i\xi} & 0 & 0 \\
0 & 0 & e^{-2i\xi}
\end{array}\right)\cdot J
\ee
The $Z_4$ symmetric terms in the 3HDM potential are (\ref{3hdm-z4-pot});
\be
\lambda_{1323}(\phi_1^\dagger\phi_3)(\phi_2^\dagger\phi_3)+ \lambda^*_{1323} (\phi_3^\dagger \phi_1)(\phi_3^\dagger \phi_2) + \lambda_{1212}(\phi_1^\dagger\phi_2)^2 + \lambda^{*}_{1212} (\phi_2^\dagger\phi_1)^2 
\ee
As in (\ref{ui-invariant-under-j''}), the first two terms are invariant under $J''$ for $\xi =  \psi_{\lambda_{1323}}/3$, where $\psi_{\lambda_{1323}}$ is the phase of $\lambda_{1323}$.

The last two terms are automatically invariant under $J''$;
\be 
\lambda_{1212} (\phi_1^\dagger\phi_2)^2 + \lambda^{*}_{1212} (\phi_2^\dagger\phi_1)^2 \quad \underrightarrow{\mbox{under J}''} \quad |\lambda_{1212}| e^{i\psi _ {\lambda_{1212}}}(\phi_1 ^ \dagger \phi_2)^2 + |\lambda_{1212}| e^{-i\psi _ {\lambda_{1212}}}(\phi_2 ^ \dagger \phi_1)^2 \nonumber
\ee
where $\psi_{\lambda_{1212}}$ is the phase of $\lambda_{1212}$, and therefore they are symmetric under the full $\Z_4 \times \Z_2^*$. 

However, similar to the previous case, restrictions (\ref{U1-restrictions}) on the $T$-symmetric terms (\ref{3HDMpotential-Tsymmetric}), result in a potential which is symmetric under a $\Z_2$ group generated by the unitary transformation
\be
d =  \left(\begin{array}{ccc} 
0 & e^{i\delta} & 0 \\
e^{-i\delta} & 0 & 0 \\
0 & 0 & 1 
\end{array}\right)   \label{dZ2}
\ee
with the condition $\delta = {\psi_{\lambda_{1212}}/2}$ resulting from; 
\bea 
&&\lambda_{1212} (\phi_1^\dagger\phi_2)^2 \quad \underrightarrow{\mbox{under  d}} \quad |\lambda _{1212}| e^{i\psi _ {\lambda_{1212}}}e^{-4i\delta}(\phi_2 ^ \dagger \phi_1)^2  \nonumber\\  
&& \lambda^{*}_{1212} (\phi_2^\dagger\phi_1)^2  \quad \underrightarrow{\mbox{under  d}} \quad |\lambda_{1212}| e^{-i\psi _ {\lambda_{1212}}} e^{4i\delta}(\phi_1 ^ \dagger \phi_2)^2  \nonumber\\ 
&&\quad \rightarrow \quad \psi _{\lambda_{1212}} -4\delta = - \psi _{\lambda_{1212}} \quad \rightarrow \quad \delta = \frac{\psi _{\lambda '}}{2}
\eea

However, $d$ does not commute with $Z_4$ phase rotations, therefore, the resulting symmetry group of the potential is non-Abelian.
Hence, the  $\Z_4 \times \Z_2^*$ symmetry group is not realizable.

\subsubsection{$\Z_8^*$ group}

Consider now the $\Z_8^*$ group generated by
\be
J'''_{\xi} =  J'' \cdot b = \left(\begin{array}{ccc} 
0 & e^{i\xi - i\pi/4} & 0 \\
e^{i\xi+i\pi/4} & 0 & 0 \\
0 & 0 &  e^{-2i\xi} 
\end{array}\right)\cdot J\label{J'''}
\ee
with the property $(J''')^2 = a$, the generator of the $\Z_4$.
\be 
(J''')^2 = \left(\begin{array}{ccc} 
e^{- i\pi/2} & 0 & 0 \\
0 & e^{i\pi/2} & 0 \\
0 & 0 &  1 
\end{array}\right)
\ee
Note that upon action of $J'''$ one of the $Z_4$-symmetric terms (\ref{3hdm-z4-pot}) $(\phi_1^\dagger \phi_2)^2$ changes its sign; 
\be 
(\phi_1^\dagger \phi_2)^2 \quad \underrightarrow{\mbox{under J}'''} \quad e^{i\pi} (\phi_1^\dagger \phi_2)^2 = -(\phi_1^\dagger \phi_2)^2
\ee
Therefore, such a term cannot appear in the potential. This means we are left with only one type of $\Z_4$-symmetric 
term in the potential. According to the discussion in Section \ref{Identify}, rank$X(V) = 1$ and, therefore, the potential
becomes symmetric under a continuous symmetry group. Thus, $\Z_8^*$ is not realizable as a symmetry group of a 3HDM potential.

\subsection{Embedding $\Z_3$}

The $\Z_3$ subgroup of the maximal torus $T$ is generated by the transformation $a$
\be 
a = \left(\begin{array}{ccc} 
e^{-2i\pi/3} & 0 & 0 \\
0 & e^{2i\pi/3} & 0 \\
0 & 0 & 1
\end{array}\right) \label{a-z3}
\ee
This group arises from the $U(1)_1$-type group, therefore the representation for $J'$ in (\ref{J'-in-U1}) is still valid.
Therefore the vector of phases has the form:
\be 
\left( \quad \gamma _1 - \gamma '_2 \quad , \quad \gamma _2 - \gamma '_1  \quad , \quad \gamma '_1 + \gamma '_2 - \gamma _1 - \gamma _2 \quad \right)
\ee
And has to agree with $(-\frac{2\pi}{3}k , \frac{2\pi}{3}k , 0)$. With the same method as before, we see that $J''$ has the form
\be
J''_{\xi} =  \left(\begin{array}{ccc} 
0 & e^{i\xi-i\pi/3} & 0 \\
e^{i\xi+i\pi/3} & 0 & 0 \\
0 & 0 & e^{-2i\xi} 
\end{array}\right)\cdot J\,\label{J'-Z3}
\ee
Since $(J'')^2 = a$ in (\ref{a-z3}), this transformation generates the group $\Z_6^*$.
The $\Z_3$-symmetric terms in the potential are (\ref{z3-3hdm-pot})
\be
\lambda_1 (\phi_1^\dagger \phi_2)(\phi_1^\dagger \phi_3) + 
\lambda_2 (\phi_2^\dagger \phi_3)(\phi_2^\dagger \phi_1) + 
\lambda_3 (\phi_3^\dagger \phi_1)(\phi_3^\dagger \phi_2) + h.c.\label{Z3symmetric}
\ee
with complex $\lambda_i$.
Let's see what conditions have to be satisfied to make these terms symmetric under $\Z_6^*$.
\be 
J'' \phi = \left(\begin{array}{ccc} 
0 & e^{i\xi - i \pi/3} & 0 \\
e^{i\xi +i\pi/3} & 0 & 0 \\
0 & 0 & e^{-2i\xi}
\end{array} \right) \cdot J \cdot \left(\begin{array}{c}
\phi _1 \\
\phi _2 \\
\phi _3 
\end{array} \right)  = \left( \begin{array}{c}
e^{i\xi- i \pi/3} \phi _2^\dagger \\
e^{i\xi+i\pi/3}\phi _1^\dagger \\
e^{-2i\xi}\phi _3^\dagger  
\end{array} \right)
\ee
Therefore
\bea
&&\lambda_1 (\phi_1^\dagger \phi_2)(\phi_1^\dagger \phi_3) \quad \underrightarrow{\mbox{under J}''} \quad |\lambda_1| e^{i \psi_{\lambda_1}} e^{i \pi-3i\xi} (\phi_1^\dagger \phi_2)(\phi_3^\dagger \phi_2) \nonumber\\ 
&&\lambda_2 (\phi_2^\dagger \phi_3)(\phi_2^\dagger \phi_1) \quad \underrightarrow{\mbox{under J}''} \quad |\lambda_2| e^{i \psi_{\lambda_2}} e^{-i \pi-3i\xi} (\phi_3^\dagger \phi_1)(\phi_2^\dagger \phi_1) \nonumber\\ 
&&\lambda_3 (\phi_3^\dagger \phi_1)(\phi_3^\dagger \phi_2) \quad \underrightarrow{\mbox{under J}''} \quad |\lambda_3| e^{i \psi_{\lambda_3}} e^{6i\xi} (\phi_2^\dagger \phi_3)(\phi_1^\dagger \phi_3) 
\eea
For these terms to be $Z_3$ invariant, the following conditions have to be satisfied:
\be 
|\lambda_1|  = |\lambda_2| \quad \mbox{and} \quad \psi_{\lambda_1} + \psi_{\lambda_2} = 3\xi + \pi \label{Z6conditions}
\ee
And since $ \psi_{\lambda_3} = -3\xi$, one can deduce that 
\be 
 \psi_{\lambda_1} + \psi_{\lambda_2} + \psi_{\lambda_3} = \pi
\ee
where $\psi_i$ are phases of $\lambda_i$.

It is, therefore, possible to have a $\Z_6^*$-symmetric potential by requiring (\ref{Z6conditions}) and applying the usual
restrictions (\ref{U1-restrictions}) on the $T$-symmetric part of the potential (\ref{3HDMpotential-Tsymmetric}).

However, just as in the previous cases, the resulting potential becomes symmetric under a larger unitary symmetry group, of the form (\ref{dZ2}), with the conditions:
\be 
|\lambda_1|  = |\lambda_2| \quad \mbox{and} \quad \psi_{\lambda_1} - \psi_{\lambda_2} = 3\delta
\ee
which come from requiring the sum of the following terms to be invariant under $d$:
\bea
&&\lambda_1 (\phi_1^\dagger \phi_2)(\phi_1^\dagger \phi_3) \quad \underrightarrow{\mbox{under d}} \quad |\lambda_1| e^{i \psi_{\lambda_1}} e^{-3i\delta} (\phi_2^\dagger \phi_1)(\phi_2^\dagger \phi_3)  \nonumber\\
&&\lambda_2 (\phi_2^\dagger \phi_3)(\phi_2^\dagger \phi_1) \quad \underrightarrow{\mbox{under d}} \quad |\lambda_2| e^{i \psi_{\lambda_2}} e^{3i\delta} (\phi_1^\dagger \phi_3)(\phi_1^\dagger \phi_2) \nonumber\\
&&\lambda_3 (\phi_3^\dagger \phi_1)(\phi_3^\dagger \phi_2) \quad \underrightarrow{\mbox{under d}} \quad |\lambda_3| e^{i \psi_{\lambda_3}} (\phi_3^\dagger \phi_2)(\phi_3^\dagger \phi_1) 
\eea
Again this $d$ does not commute with $Z_3$, and therefore the group $\Z_6^*$ is not realizable.

\subsection{Embedding $\Z_2\times \Z_2$}

The $\Z_2 \times \Z_2 \subset T$ symmetry can be realized as a group of simultaneous sign flips of any pair of doublets
(or alternatively of independent sign flips of any of the three doublets). 
If the simultaneous sign flip of doublets $i$ and $j$ are denoted as $R_{ij}$,
then the group contains $\{ 1,\,R_{12},\, R_{13},\, R_{23} \} $, with $R_{12}R_{13}=R_{23}$ and any $(R_{ij})^2=1$.
Therefore the generator of $Z_2 \times Z_2$ is:
\be 
a = \left(\begin{array}{ccc} 
\sigma_1 &   &  \\
  & \sigma_2 &   \\
  &   & \sigma_3 
\end{array}\right), \qquad \mbox{where} \quad \sigma_i = \pm 1
\ee
From the equation $ a x a^{-1} = x$, we have:
\be 
\left(\begin{array}{ccc} 
\sigma_1 ^2 x_{11} & \sigma_1 \sigma_2 x_{12}  & \sigma_1 \sigma_3 x_{13} \\
\sigma_2 \sigma_1 x_{21} & \sigma_2 ^2 x_{22} &  \sigma_2 \sigma_3 x_{23} \\
\sigma_3 \sigma_1 x_{31}  & \sigma_3 \sigma_2 x_{32}  & \sigma_3 ^2 x_{33} 
\end{array}\right) = \left(\begin{array}{ccc} 
x_{11} & x_{12}  & x_{13} \\
x_{21}  & x_{22} &  x_{23} \\
x_{31}  & x_{32}  & x_{33} 
\end{array}\right)
\ee
Since the products $\sigma_1 \sigma_2$, $\sigma_1 \sigma_3$ and $\sigma_2 \sigma_3$ can not all be $+1$ at the same time, therefore $x_{ij} =0 , i \neq j$. So, $x$ is a diagonal matrix, and since it has to be unitary, it has the form:
\be 
x = \left(\begin{array}{ccc} 
e^{i\xi_1} & 0 &0 \\
0 & e^{i\xi_2} & 0 \\
0 & 0 & e^{-i\xi_1 - i\xi_2} 
\end{array}\right)
\ee
This group can be embedded into $\Z_2 \times \Z_2 \times \Z_2^*$, where $\Z_2^*$ is generated by 
\be
J'_{\xi_1,\xi_2} =  \left(\begin{array}{ccc} 
e^{i\xi_1} & 0 &0 \\
0 & e^{i\xi_2} & 0 \\
0 & 0 & e^{-i\xi_1 - i\xi_2} 
\end{array}\right)\cdot J \label{J'-Z2Z2}
\ee
Up to here, any $\xi_1$, $\xi_2$ can be used to construct the $\Z_2 \times \Z_2 \times \Z_2^*$ group.
Note that this group can also be written as $\Z_2 \times \Z_2^* \times \Z_2^*$ or $\Z^*_2 \times \Z^*_2 \times \Z_2^*$ by redefining the generators.

The $\Z_2 \times \Z_2$-symmetric potential contains the following terms (\ref{z2Xz2-3hdm-pot});
\be
\lambda_{12} (\phi_1^\dagger \phi_2)^2 + 
\lambda_{23} (\phi_2^\dagger \phi_3)^2 + 
\lambda_{31} (\phi_3^\dagger \phi_1)^2 + h.c.\label{termsZ2Z2} 
\ee
with complex $\lambda_{ij}$. 
Let's see with which conditions these terms are invariant under $\Z_2 \times \Z_2 \times \Z_2^*$-symmetric potential:
\be 
J' \phi = \left(\begin{array}{ccc} 
e^{i\xi_1} & 0 &0 \\
0 & e^{i\xi_2} & 0 \\
0 & 0 & e^{-i\xi_1 - i\xi_2}
\end{array} \right) \cdot J \cdot \left(\begin{array}{c}
\phi _1 \\
\phi _2 \\
\phi _3 
\end{array} \right)  = \left( \begin{array}{c}
e^{i\xi_1} \phi _1^\dagger \\
e^{i\xi_2}\phi _2^\dagger \\
e^{-i\xi_1 - i\xi_2}\phi _3^\dagger  
\end{array} \right)
\ee
Therefore,
\bea
&&\lambda_{12} (\phi_1^\dagger \phi_2)^2 \quad \underrightarrow{\mbox{under J}'} \quad |\lambda_{12}| e^{i \psi_{12}} e^{-2i\xi_1 +2i\xi_2} (\phi_1^\dagger \phi_2)^2  \nonumber\\
&&\lambda_{23} (\phi_2^\dagger \phi_3)^2 \quad \underrightarrow{\mbox{under J}'} \quad | \lambda_{23}| e^{i \psi_{23}} e^{-2i\xi_1 -4i\xi_2} (\phi_2^\dagger \phi_3)^2 \nonumber\\
&&\lambda_{31} (\phi_3^\dagger \phi_1)^2 \quad \underrightarrow{\mbox{under J}'} \quad | \lambda_{31}| e^{i \psi_{31}} e^{4i\xi_1 +2i\xi_2} (\phi_3^\dagger \phi_1)^2 
\eea
So, we require the following conditions:
\bea
\psi_{12} -2\xi_1 +2\xi_2 = -\psi_{12}  \qquad &\quad&  \quad \psi_{12} = \xi_1 - \xi_2  \nonumber\\
\psi_{23} -2\xi_1 -4\xi_2 = -\psi_{23}  \qquad &\rightarrow& \quad \psi_{23} = \xi_1 + 2\xi_2  \nonumber\\  
\psi_{31} +4\xi_1 +2\xi_2 = -\psi_{31} \qquad &\quad&  \quad \psi_{31} = -2\xi_1 -\xi_2 
\eea
Which adds up to $\psi_{12} + \psi_{23} + \psi_{31} = 0 $,  or alternatively, when the product 
$\lambda_{12} \lambda_{23} \lambda_{31}$ is real.

Note that in contrast to the previous cases, the full $T$-symmetric potential (\ref{3HDMpotential-Tsymmetric})
is invariant under $\Z_2 \times \Z_2 \times \Z_2^*$ and no larger group. Therefore this symmetry group is realizable.

\subsection{Embedding $\Z_2$}

The unitary Abelian symmetry group $\Z_2$ can be implemented in a variety of ways, all of which are equivalent.
One can take, for example, the $\Z_2$ group generated by sign flips of the first and second doublets, $R_{12}$.
The usual $CP$-transformation $J$ commutes with $R_{12}$. An analysis similar to the previous case shows that the $\Z_2 \times \Z_2^*$ group generated by $R_{12}$ and $J$
is realizable.

It is also possible to embed the $\Z_2$ group into $\Z_4^*$. Consider the following
transformation
\be
J'_{\xi} =  \left(\begin{array}{ccc} 
0 & e^{i\xi} & 0 \\
-e^{i\xi} & 0 & 0 \\
0 & 0 & e^{-2i\xi} 
\end{array}\right)\cdot J\,,\quad \mbox{with} \quad (J')^2 = R_{12} \label{J'-Z2}
\ee
where $\xi$ is an arbitrary but fixed parameter.
The $Z_2$-symmetric potential consists of the $T$-symmetric terms (\ref{3HDMpotential-Tsymmetric}) with the usual restrictions (\ref{U1-restrictions}) and the specifically $Z_2$-symmetric terms:
\bea
&&\lambda_5(\phi_1^\dagger \phi_2)^2 + \lambda_6(\phi_1^\dagger\phi_2)\left[(\phi_1^\dagger\phi_1)-(\phi_2^\dagger\phi_2)\right]
+ \lambda_7 (\phi_1^\dagger\phi_3)(\phi_2^\dagger\phi_3) \nonumber \\
&&+\lambda_8 (\phi_1^\dagger\phi_3)^2 + \lambda_9 (\phi_2^\dagger\phi_3)^2 + h.c. 
\eea
Let's study how these terms transform under $J'$;
\be 
J' \phi = \left(\begin{array}{ccc} 
0 & e^{i\xi} &0 \\
-e^{i\xi} & 0 & 0 \\
0 & 0 & e^{-2i\xi}
\end{array} \right) \cdot J \cdot \left(\begin{array}{c}
\phi _1 \\
\phi _2 \\
\phi _3 
\end{array} \right)  = \left( \begin{array}{c}
e^{i\xi} \phi _2^\dagger \\
-e^{i\xi}\phi _1^\dagger \\
e^{-2i\xi}\phi _3^\dagger  
\end{array} \right)
\ee
Therefore,
\bea
\lambda_5(\phi_1^\dagger \phi_2)^2  & \underrightarrow{\mbox{under J}'} &  |\lambda_5| e^{i \psi_5} (-\phi_1^\dagger \phi_2)^2 \\ \nonumber
\lambda_6(\phi_1^\dagger\phi_2)\left[(\phi_1^\dagger\phi_1)-(\phi_2^\dagger\phi_2)\right]  & \underrightarrow{\mbox{under J}'} &  |\lambda_6| e^{i \psi_6}  (-\phi_1^\dagger\phi_2)\left[(\phi_2^\dagger\phi_2)-(\phi_1^\dagger\phi_1)\right] \\ \nonumber
\lambda_7 (\phi_1^\dagger\phi_3)(\phi_2^\dagger\phi_3) & \underrightarrow{\mbox{under J}'} & |\lambda_7| e^{i \psi_7} e^{-6i\xi} (\phi_3^\dagger \phi_1)(\phi_3^\dagger \phi_2) \\  \nonumber
\lambda_8 (\phi_1^\dagger\phi_3)^2 & \underrightarrow{\mbox{under J}'} & |\lambda_8|e^{i \psi_8} e^{-6i\xi}(\phi_3^\dagger\phi_2)^2 \\ \nonumber
\lambda_9 (\phi_2^\dagger\phi_3)^2 & \underrightarrow{\mbox{under J}'} & |\lambda_8|e^{i \psi_9} e^{-6i\xi}(-\phi_3^\dagger\phi_1)^2 \\ \nonumber
\eea
To keep these terms invariant, we require:
\bea
&&|\lambda_8| =|\lambda_9| \nonumber\\
&& \psi_8 + \psi_9 = 6\xi \nonumber\\
&& 2\psi_7 = 6\xi 
\eea
Or simply $\psi_8 + \psi_9 = 2\psi_7$. 
By adjusting $\lambda_5$ and $\lambda_6$, one can guarantee that no other unitary symmetry appears as a result of (\ref{U1-restrictions}).
Therefore, this $\Z_4^*$ group is realizable.


\chapter{$Z_p$ scalar dark matter from NHDM}\label{chapter7}

The discrepancy between the masses of large astronomical objects, determined from their gravitational effects, and the masses calculated from the luminous matter they contain, led to the hypothesis of dark matter. 

In 1932, by measuring the orbital velocities of stars in the Milky Way, Jan Oort was the first to find evidence for dark matter, when he calculated that the mass of the galactic plane must be more than the mass of the material that can be seen \cite{Freeman}. In 1933, while examining the Coma galaxy cluster, Fritz Zwicky used the virial theorem to infer the existence of unseen matter, which he referred to as "dark matter" \cite{Zwicky1,Zwicky2}. He calculated the gravitational mass of the galaxies within the cluster and obtained a value at least 400 times greater than expected from their luminosity, which meant that most of the matter must be dark. 

Dark matter is a currently unknown type of matter responsible for 23 percent of the mass-energy content of the universe. It does not emit or absorb light (or any other electromagnetic radiation), and therefore cannot be seen. Dark matter particle must be (almost) stable on cosmological time-scales \cite{Hambye}. Clearly its decay lifetime has to be larger than the age of the Universe ($\sim{10}^{18}$ seconds). 

Despite the astronomical evidence for the existence of dark matter (DM) \cite{DM}, there is still no direct experimental clue of which particle can play the role of a dark matter candidate. None of the Standard Model particles fully explain the anomalies observed. However, in the models beyond SM one can find dark matter candidates, which interestingly also address other particle physics issues such as electroweak symmetry breaking and small neutrino masses. 

Initially DM models contained completely inert dark matter, which would not be destroyed in any reaction, and would only annihilate through $dd^* \to X_{SM}$ (We use the notation $d$: dark matter candidate, $d^*$: dark matter anti-particle). Many models were introduced in which the stability of DM is provided by a conserved $\Z_2$ symmetry. 

The motivation for stabilization of DM using $\Z_2$ symmetry comes from the Standard Model, where stability results from the gauge symmetries and particle content, rather than from an ad-hoc symmetry.

In supersymmetric models the role of the $\Z_2$ symmetry is played by the $R$-parity, while in more phenomenologically oriented models such as the Inert doublet model \cite{inert}, or the minimal singlet model \cite{singlet}, the $\Z_2$ symmetry is imposed by hand when constructing the Lagrangian. In these models, all the SM particles including the SM-like Higgs boson have positive parity, while the dark sector particles are of negative parity. The lightest among these negative parity particles, which we will generically denote as $d$ here, is stable and represents the dark matter candidate.

Although $\Z_2$ symmetric models avoid the decay of $d$ to Standard Model particles, they allow direct two-particle annihilation $dd \to X_{SM}$, where $X_{SM}$ is any set of SM particles.
This certainly changes the kinetics of dark matter evolution in the early Universe and its relic abundance after the freeze-out. 

Trying to avoid the direct two-particle annihilation $dd \to X_{SM}$ , it is natural to explore groups larger than $\Z_2$ and try to stabilize dark matter by conserving their discrete quantum number. This idea was explored in a number of papers, in which both Abelian \cite{MaZ3,ZNdetailed,Z2Z2} and non-Abelian \cite{nonAbelian} finite groups were used. 

One particular class of groups used to stabilize dark matter are cyclic groups $\Z_p$ \cite{MaZ3,ZNdetailed}. With this choice, all fields are characterized by a conserved quantum number which we will call the $\Z_p$ charge $q$ and which is additive modulo $p$. Usually one assigns zero $\Z_p$ charge ($q=0$), to all the SM fields including the SM-like Higgs boson, while the dark matter candidates have non-zero $\Z_p$ charges. This strategy prohibits the direct two-particle annihilation $dd \to X_{SM}$, but opens up the possibility of a new two-particle process such as $dd \to d^* X_{SM}$, which was called semi-annihilation in \cite{semi-annihilation}, or even multi-particle versions of it, "multiple semi-annihilation". Such processes, have different impact on the kinetics of the dark mater abundances in the early Universe \cite{semi-annihilation-evolution}.

In this Chapter we demonstrate that $\Z_p$-stabilized scalar dark matter can easily arise in multi-Higgs-doublet models. 
The advantage of these models is that even with few doublets one can get $\Z_p$ with a rather large $p$. This fact is useful in avoiding semi-annihilation process $dd \to d^* X_{SM}$, which is allowed in a $\Z_3$-symmetric model and can be avoided by using a $\Z_p$ group with $p >3$.

The structure of the Chapter is as follows. In Section~\ref{Z3-symmetric} we study $Z_p$-symmetric models, particularly a 3HDM example in which dark matter candidates are stabilized by the $\Z_3$ symmetry group. In Section~\ref{avoiding} we explore $\Z_p$-symmetric 4HDMs, and the possibility to avoid semi-annihilation processes with large $\Z_p$, in particular we study a $\Z_7$-symmetric potential in detail. The results of this Chapter are published in \cite{Zp-scalar-dm}.

\section{$\Z_3$-symmetric 3HDM}\label{Z3-symmetric}

Before presenting the example of a $\Z_3$-stabilized dark matter in 3HDM, let's list the conditions that an electroweak symmetry breaking (EWSB) model with $\Z_p$-stabilized scalar dark matter must satisfy.
\begin{itemize}
\item
The entire Lagrangian and not only the Higgs potential must be $\Z_p$-symmetric. The simplest way to achieve this is to set the $\Z_p$ charges of all the SM particles to zero and to require that only one Higgs doublet (the SM-like doublet) couples to fermions. The $\Z_p$ charge of this doublet must be zero, and it does not matter which doublet is chosen to be SM-like due to the freedom to simultaneously shift the $\Z_p$ charges of all the doublets.

\item
The $\Z_p$ symmetry must remain after EWSB. This is possible when only the SM-like doublet acquires a non-zero vacuum expectation value. 

\item
If we insist on $\Z_p$-stabilization, that is, we require that not only decays but also 2-, 3-, $\dots$, $(p-1)$-particle annihilation to SM fields are forbidden by quantum numbers, then the dark matter candidates must have $\Z_p$ charge $q$ which is coprime with $p$.

\end{itemize}

To show that multi-Higgs-doublet models can easily satisfy these conditions, we start with the simplest single parametric model of this kind, $\Z_3$-symmetric 3HDM.

A natural way to implement the $\Z_3$ symmetry in 3HDM is to construct a potential symmetric under $\phi_1 \to \phi_2 \to \phi_3 \to \phi_1$. However, we proved in Chapter \ref{chapter5} that upon an appropriate Higgs basis change, this transformation will turn into pure phase rotations of certain doublets. Therefore, we use the phase rotation representations of the cyclic symmetry groups $\phi_i \to e^{i\alpha_i}\phi_i$. Also, to keep the notation short, we will describe any such transformation by providing the $N$-tuple of phases $\alpha_i$.

A scalar potential invariant under a certain group $G$ of phase rotations can be written as a sum $V = V_0 + V_G$, where $V_0$ is invariant under any phase rotation, while $V_G$ is a collection of extra terms which realize the chosen symmetry group. The generic phase rotation invariant part has the form
\be
\label{Tsymmetric6}
V_0 = \sum_i \left[- m_i^2 (\phi_i^\dagger \phi_i) + \lambda_{ii} (\phi_i^\dagger \phi_i)^2\right] 
+ \sum_{ij}\left[\lambda_{ij}(\phi_i^\dagger \phi_i) (\phi_j^\dagger \phi_j) + 
\lambda'_{ij}(\phi_i^\dagger \phi_j) (\phi_j^\dagger \phi_i)\right] 
\ee
while $V_G$ obviously depends on the group. In particular, for the group $\Z_3$ in 3HDM we have
\be
V_{\Z_3} = \lambda_{1}(\phi_3^\dagger\phi_1)(\phi_2^\dagger\phi_1) + 
\lambda_{2}(\phi_1^\dagger\phi_2)(\phi_3^\dagger\phi_2) + 
\lambda_{3}(\phi_2^\dagger\phi_3)(\phi_1^\dagger\phi_3) + h.c. 
\ee
where at least two of the coefficients $\lambda_1,\,\lambda_2,\,\lambda_3$ are non-zero (otherwise, the potential would have a continuous symmetry). This potential is symmetric under the phase rotations
\be 
\phi_1 \to \phi_1 \quad , \quad \phi_2 \to e^{2i\pi/3}\phi_2 \quad , \quad \phi_3 \to e^{4i\pi/3}\phi_3
\ee
The vector of phases in this case has the form 
\be
a = {2\pi \over 3}(0,\, 1,\, 2)\,,\quad a^3 = 1 \label{generatorZ3}
\ee

This notation suggests $\Z_3$ quantum numbers $0, 1$ for the first and second doublets respectively, and quantum number $2$ for the third doublet.
In fact, the assignment of these charges to the three doublets is completely arbitrary, and the group generated by
the generator $a$ with permuted charges has the same action on the potential. In (\ref{generatorZ3}) we simply chose $\phi_1$ to be the SM-like doublet.

Whether this symmetry is conserved or spontaneously broken depends on the pattern of the vacuum expectation values.
If we insist on conservation of the $\Z_3$ symmetry, we require that 
\be
\lr{\phi_1^0} = \frac{v}{\sqrt{2}} \quad , \quad \lr{\phi_2^0} = \lr{\phi_3^0} = 0
\ee

The potential $V_0 + V_{\Z_3}$ can have a $\Z_3$-symmetric global minimum upon appropriate choice of coefficients. Below we show that we do not require any constraints on $\lambda$'s and therefore no fine-tuning is required for this point to be a minimum of the potential. 

First, we note that if $(v_1,\,0,\,0)$ is an extremum of $V_0$, then it is also an extremum of $V_0 + V_{\Z_3}$, 
since the extra terms from $V_{\Z_3}$ contain $\phi_1$ only linearly and quadratically and they do not create any new linear terms in the potential after EWSB.

Therefore, when constructing a model, one can first build $V_0$ with a global minimum at $(v_1,\,0,\,0)$
and then add a sufficiently weak (with no strong couplings) $V_{\Z_3}$ so that this point remains a minimum.

Now, turning to minimization of $V_0$, suppose that we search for the neutral minimum with a generic complex 
v.e.v. pattern $(v_1,\, v_2,\, v_3)$. The potential would have the form;
\be
V_0 = -{m_i}^2 {v_i}^2 + \lambda_{ii} {v_i}^4 + \left( \lambda_{ij}+\lambda'_{ij} \right) {v_i}^2{v_j}^2 
\ee
Introducing $\rho_i = |v_i|^2 \ge 0$, we can rewrite $V_0$ as;
\bea
V_0 &=& -{m_i}^2 \rho_i + \lambda_{ii} {\rho_i}^2 + \left( \lambda_{ij}+\lambda'_{ij} \right) \rho_i \rho_j \nonumber \\
    &=& - M_i \rho_i + {1 \over 2}\Lambda_{ij}\rho_i \rho_j 
\eea
where $M_i = (m_1^2,\,m_2^2,\,m_3^2)$, and $\Lambda_{ij}$ is constructed from $\lambda_{ij}$ and $\lambda_{ij}^\prime$.
One could introduce  $\mu_i = (\Lambda_{ij})^{-1}M_j$, and rewrite the potential as
\be
V_0= {1 \over 2}\Lambda_{ij}(\rho_i -\mu_i)(\rho_j - \mu_j) + const\,.\label{shifted}
\ee
Positivity conditions on the potential guarantee that $\Lambda_{ij}$ is a positive definite matrix, therefore its inverse exists. 

The allowed values of the $\rho_i$, ($\rho_i \ge 0$), populate the first octant in the three-dimensional euclidean space.
Due to (\ref{shifted}), the search for the global minimum (trying to minimize $(\rho_i -\mu_i)$) can be reformulated as the search for the point in the first octant which lies closest to $\mu_i$ in the euclidean metric defined by $\Lambda_{ij}$. Clearly, if $\mu_i$ itself lies inside the first octant ($\mu_1,\, \mu_2,\, \mu_3 \ge 0$), then the global minimum is at $\rho_i = \mu_i$. If $\mu_i$ lies outside the first octant, then the closest point lies either on the edge of the first octant, or at the origin. 

The global minimum is unique by the convexity arguments; since the octant $\rho_i \ge 0$ is convex, there is a unique $\mu_i$ that minimizes $(\rho_i -\mu_i)$.


Having established that the required vacuum pattern is generically possible, we now switch to a simple
version of the model. This is done only to simplify the presentation of the argument;
if needed, the calculations can be repeated for a generic potential.
Namely, we now choose $m_1^2 > 0$, $m_2^2,\,m_3^2 < 0$ and take $\Lambda_{ij} = 2\lambda_0 \delta_{ij}$,
so the the potential becomes;
\bea
V &=& V_0+V_{Z_3}   \\
&=& - m_1^2 (\phi_1^\dagger \phi_1) + |m_2^2| (\phi_2^\dagger \phi_2) + |m_3^2| (\phi_3^\dagger \phi_3)
+ \lambda_0 \left[(\phi_1^\dagger \phi_1)^2+ (\phi_2^\dagger \phi_2)^2 + (\phi_3^\dagger \phi_3)^2\right] \nonumber\\
&& + \lambda_{1}(\phi_3^\dagger\phi_1)(\phi_2^\dagger\phi_1) + 
\lambda_{2}(\phi_1^\dagger\phi_2)(\phi_3^\dagger\phi_2) + 
\lambda_{3}(\phi_2^\dagger\phi_3)(\phi_1^\dagger\phi_3) + h.c. \nonumber
\eea
By construction, its global minimum is at $\lr{\phi_i^0} = (v/\sqrt{2},\,0,\,0)$, where $v^2 = m_1^2/\lambda_0$.

\subsection{Mass matrix}

In order to find the mass matrices, we write the doublets as
\be
\phi_1 = \doublet{G^+}{{1\over\sqrt{2}}(v + h + i G^0)}\,,\quad 
\phi_{2} = \doublet{w_{2}^+}{z_{2}}\,,\quad 
\phi_{3} = \doublet{w_{3}^+}{z_{3}} \label{expand}
\ee
Here $h$ is the SM-like Higgs boson, $G^0$ and $G^+$ are the would-be Goldstone bosons, 
while $w_2^+,\, w_3^+$ and $z_2,\, z_3$ are charged and neutral scalar bosons, respectively.

Note that in contrast to the usual practice, we describe the neutral scalar bosons in the second and third doublet
by complex fields rather than pairs of neutral fields. The reason is that, by construction, the fields corresponding to the real and imaginary pairs of $z$'s have identical masses and coupling constants. In any process that can arise in this model, these two fields are emitted and absorbed simultaneously, so they can be described by a single complex field $z_i$.

Now we assign the $\Z_3$ quantum numbers according to the generator $a=2\pi/3 (0, 1, 2)$. The SM-like Higgs boson gets $\Z_3$-charge $q=0$. Fields $w_2^+$ and $z_2$, from the second doublet, have $\Z_3$-charge $q=1$, while for $w_3^+$ and $z_3$, from the third doublet, $q=2$. The conjugates of fields $w_2^-$ and $z_2^*$ have $q=2$ ( $=-1$ mod 3), while for $w_3^-,\, z_3^*$ $q=1$ ( $=-2$ mod 3). Note that the neutral scalars with similar quantum numbers can mix.

Let's rewrite the potential in terms of the new parameters. The second order potential is:
\bea
V_{2} &=& \left[-m_1^2 v +\lambda_0 v^3 \right]h  \\
  &+& \left[-\frac{m_1^2}{2} + \frac{3\lambda_0}{2} v^2 \right](h^2) \nonumber\\
  &+& \left[-m_1^2 +\lambda_0 v^2\right](G^-G^+) + \left[-\frac{m_1^2}{2} + \frac{\lambda_0v^2}{2} \right](G_0^2)  \nonumber\\
  &+& |m_2^2|(w_2^- w_2^+) + |m_3^2|(w_3^- w_3^+) \nonumber\\
  &+& |m_2^2|(z_2^* z_2) + |m_3^2|(z_3^* z_3) + \left[\frac{\lambda_1v^2}{2} \right](z_2^* z_3^*) + \left[\frac{\lambda_1^*v^2}{2} \right] (z_3 z_2) \nonumber
\eea
With the value $v^2= \frac{m_1^2}{\lambda_0}$, one finds the mass spectrum; The SM-like Higgs boson has mass $m_h^2 = 2m_1^2$, while the masses of the charged scalar bosons are $m_{w_2^\pm}^2 = |m_2^2|$, $m_{w_3^\pm}^2 = |m_3^2|$, and the Goldstone bosons $G^{\pm}, G^0$ obviously have mass zero.

Neutrals with equal $q$ can mix, which indeed happens at $\lambda_1 \not = 0$. We describe the resulting mass eigenstates by complex fields $d$ and $D$ ($m_d < m_D$), both having $q=1$: 
\be
m_{D,d}^2 = {|m_2^2| + |m_3^2| \over 2} \pm {1 \over 2}\sqrt{(|m_2^2| - |m_3^2|)^2 + {|\lambda_1|^2 \over \lambda_0^2}m_1^4} 
\ee
where $d$ and $D$ are rotations of $z_2$ and $z_3^*$;
\bea
&& d = \cos\alpha\, z_2 + \sin\alpha e^{-i\beta}\, z_3^*\,,\quad D = - \sin\alpha e^{i\beta}\, z_2 + \cos\alpha\, z_3^* \nonumber\\
&& \tan2\alpha = {|\lambda_1| \over \lambda_0}{m_1^2 \over |m_2^2|-|m_3^2|}\,,\quad \beta = \mbox{arg}\, \lambda_1 
\eea
Note that within this model we have
\be
m_d < m_{w_2^\pm},\, m_{w_3^\pm} < m_D\,,\quad m_{w_2^\pm}^2 + m_{w_3^\pm}^2 = m_d^2 + m_D^2 
\ee

\subsection{Higher order terms in the $Z_3$ symmetric potential}

The triple and quartic interaction terms arising from $V_0$ and $V_{\Z_3}$ specify the dynamics of the dark matter candidates. After using a unitary gauge transformation to eliminate the Goldstone bosons, higher order terms in the $Z_3$-symmetric potential are:
\bea
V_{3}  &=& \lambda_0 \left[ 2vhG^-G^+ + vh^3 + vh G^0 G^0 \right]   \\
&&+ v \lambda_1 \left[\frac{1}{\sqrt{2}} w_3^-G^+z_2^* + \frac{1}{\sqrt{2}} w_2^-G^+z_3^* +  hz_3^*z_2^* \right]+ h.c. \nonumber\\
&&+ \frac{\lambda_2 v}{\sqrt{2}} \left[ z_2w_3^-w_2^+ + z_2z_2z_3^*\right] +h.c. \nonumber\\
&&+ \frac{\lambda_3 v}{\sqrt{2}} \left[ z_3w_2^-w_3^+ + z_2^*z_3z_3 \right] + h.c. \nonumber\\
V_{4}  &=& 2\lambda_0 [ {(G^-G^+)}^2 + \frac{1}{4}h^4 + \frac{1}{4}{(G^0 G^0)}^2 + h^2G^- G^+ + G^0G^0G^-G^+ \frac{1}{2}h^2{G^0}^2  \nonumber\\
&& \quad + {(w_2^-w_2^+)}^2 + {(z_2^*z_2)}^2 + 2w_2^-w_2^+z_2^*z_2 + {(w_3^-w_3^+)}^2 + {(z_3^*z_3)}^2 + 2w_3^-w_3^+z_3^*z_3 ]  \nonumber\\
&&+ \lambda_1 [ w_3^-G^+G^+w_2^- + \frac{1}{\sqrt{2}}w_3^-G^+hz_2^* + \frac{i}{\sqrt{2}}w_3^-G^+G^0z_2 + \frac{1}{\sqrt{2}}w_2^-G^+hz_3^*  \nonumber\\ 
&& \quad + \frac{1}{2}h^2z_2^*z_3^* + \frac{i}{2}G^0hz_2^*z_3^* - \frac{i}{\sqrt{2}}w_2G^0G^+z_3^* - \frac{i}{2}G^0hz_2^*z_3^* + \frac{1}{2}G^0G^-z_2^*z_3^*] \nonumber\\
&&+ \lambda_2 [ w_3^-G^-w_2^+w_2^+ + w_2^+G^-z_2z_3^* + \frac{1}{\sqrt{2}}hw_3^-w_2^+z_2  + \frac{1}{\sqrt{2}}hz_2z_2z_3^* \nonumber\\
&& \quad + \frac{i}{\sqrt{2}}G^0z_2w_2^+w_3^- + \frac{i}{\sqrt{2}}G^0z_2z_2z_3^*   ] \nonumber\\
&&+ \lambda_3 [ w_2^-G^-w_3^+w_3^+ + \frac{1}{\sqrt{2}}w_2^-w_3^+hz_3 - \frac{i}{\sqrt{2}}w_2^-w_3^+G^0z_3 + z_2^*z_3w_3^+G^- \nonumber\\
&& \quad + \frac{1}{\sqrt{2}}hz_2^*z_3z_3 - \frac{i}{\sqrt{2}}G^0z_3z_3z_2^*   ] + h.c. \nonumber
\eea
The higher order terms in the potential in terms of $d$ and $D$, regardless of the coefficients are of the following form;
\bea
V_{3} &=& \lambda_0 \left[ hhh \right]   \\
&&+ \lambda_1 \left[ hd^*d , hd^*D , hdD^* , hD^*D \right]  \nonumber\\
&&+ \lambda_2 \left[ w_2^+ w_3^- d , w_2^+ w_3^- D , ddd , ddD , dDD , DDD \right] +h.c. \nonumber\\
&&+ \lambda_3 \left[ w_2^- w_3^+ d^* , w_2^- w_3^+ D^* , d^*d^*d^* , d^*d^*D^* , d^*D^*D^* , D^*D^*D^* \right] +h.c. \nonumber\\
V_{4} &=& \lambda_0 [ hhhh , (w_{2,3}^-w_{2,3}^+)^2 , (w_2^-w_2^+ + w_3^-w_3^+)(dd^* , dD^* , d^*D , D^*D), \nonumber\\
&& \qquad ddd^*d^*,  ddD^*D^*,  d^*d^*DD,  DDD^*D^* ] \nonumber \\
&&+ \lambda_1 \left[ hhd^*d , hhd^*D , hhdD^* , hhD^*D \right]  \nonumber \\
&&+ \lambda_2 \left[ hw_2^+ w_3^- d , hw_2^+ w_3^- D , hddd , hddD , hdDD , hDDD \right] +h.c. \nonumber \\
&&+ \lambda_3 \left[ hw_2^- w_3^+ d^* , hw_2^- w_3^+ D^* , hd^*d^*d^* , hd^*d^*D^* , hd^*D^*D^* , hD^*D^*D^* \right] +h.c. \nonumber
\eea
The lightest particle from the second and third doublets is $d$, and it is stabilized against decaying into the SM particles by the $\Z_3$ symmetry. From the kinematically allowed terms in the third order potential, we have decays such as $D \to dh$. Charged Higgs bosons $w_{2,3}^{\pm}$ will decay into $d$ or $d^*$ plus SM particles.

With the $Z_3$ conserved quantum number, one could also write third order terms such as $D \to dZ$, and $D \to w_2^+ W^-$, $D \to w_3^- W^+$, if allowed kinematically. 

If the mass splitting between the $d$ and $D$ is small, then these processes involve virtual $h$, $Z$, etc. which then decay into the SM particles. In this aspect, $D$ decays are similar to weak decays.

In the case of symmetric dark matter, the main process leading to depletion of dark matter after electroweak symmetry breaking is the direct annihilation $dd^* \to X_{SM}$, and the semi-annihilation reaction is only a correction to this process. Still, it might be possible that this correction leads to a sizeable departure of the kinetics of the dark matter in the early Universe and affects the relic abundance at the freeze-out \cite{semi-annihilation-evolution}.

The situation becomes more interesting in models with asymmetric dark matter \cite{asymmetric},
in which an asymmetry between $d$ and $d^*$ is generated at a higher energy scale. It is possible for example that upon electroweak phase transition almost all $d$'s annihilate with $d^*$ into the SM sector, leaving behind a certain concentration of dark matter candidates $d$. 

$d$ can also scatter elastically $dd \to dd$, which can initiate semi-annihilation processes such as $dd \to d^*X_{SM}$ with a subsequent annihilation of $d^*$ with a $d$. This possibility originates from the following quartic terms in the scalar potential 
\be
{1\over \sqrt{2}} h ddd \cos\alpha \sin\alpha e^{-i\beta} (\lambda_2 \cos\alpha + \lambda_3^*\sin\alpha e^{-i\beta}) + h.c.
\ee
Depending on $\lambda$'s, this process can be as efficient as the direct annihilation 
in the usual annihilating dark matter models, or it can be suppressed by the small coupling constants. 

Finally, the same interaction terms also generate the triple annihilation processes $ddd \to h \to X_{SM}$,
whose rate is, however, suppressed at small densities compared to the semi-annihilation.

\section{Avoiding semi-annihilation in 4HDM}\label{avoiding}

The presence of the $hddd$ terms in the interaction Lagrangian in the 3HDM example, which were responsible for the two-particle semi-annihilation process, was due to the $\Z_3$ symmetry group. 

It is possible to avoid such two-particle semi-annihilation processes by employing a $\Z_p$ group with larger $p$ in a model with four Higgs doublets.

As proved in Chapter \ref{chapter5}, one can encode in the 4HDM scalar sector any group $\Z_p$ with $p \le 8$.
Note that in order to avoid continuous symmetry, one must accompany the phase-symmetric part of the potential $V_0$
with at least three distinct terms transforming non-trivially under phase rotations.

The full list of these realizable symmetries, their corresponding terms and their generators are given in Section \ref{4HDM-examples}. In Table \ref{tableDM} we remind the reader of these symmetry groups together with examples of the three interaction terms and the phase rotations that generate the corresponding group. Note that all these groups are realizable, so that if the three terms in each line are written down with non-zero coefficients, there is no phase transformation other than multiple of the generator and the common overall phase shift that leaves them invariant. 

\begin{table}[!htb]
\begin{center}
\begin{tabular}{ c c c }	
\hline
Group & Interaction terms & Generators \\ [3mm]
\hline 
  $\Z_2$ &  $(\phi_1^\dagger \phi_2),\  (\phi_1^\dagger \phi_3),\ (\phi_1^\dagger \phi_4)^2$ & ${2\pi \over 2}(0,\,0,\,0,\,1)$  \\[3mm]
  $\Z_3$ &  $(\phi_3^\dagger \phi_2),\  (\phi_1^\dagger \phi_3)(\phi_4^\dagger \phi_3),\ (\phi_1^\dagger \phi_4)(\phi_1^\dagger \phi_2)$ & ${2\pi \over 3}(0,\,1,\,1,\,2)$  \\[3mm]
  $\Z_4$ &  $(\phi_3^\dagger \phi_2),\  (\phi_1^\dagger \phi_3)(\phi_4^\dagger \phi_3),\ (\phi_1^\dagger \phi_4)^2$ & ${2\pi \over 4}(0,\,1,\,1,\,2)$  \\[3mm]
  $\Z_5$ &  $(\phi_4^\dagger \phi_3)(\phi_2^\dagger \phi_3),\  (\phi_3^\dagger \phi_2)(\phi_1^\dagger \phi_2),\ (\phi_4^\dagger \phi_1)(\phi_3^\dagger \phi_1)$ & ${2\pi \over 5}(0,\,1,\,2,\,3)$  \\[3mm]
  $\Z_6$ &  $(\phi_4^\dagger \phi_3)(\phi_2^\dagger \phi_3),\  (\phi_3^\dagger \phi_2)(\phi_1^\dagger \phi_2),\ (\phi_1^\dagger \phi_4)^2$ & ${2\pi \over 6}(0,\,1,\,2,\,3)$  \\[3mm]
  $\Z_7$ &  $(\phi_4^\dagger \phi_1)(\phi_3^\dagger \phi_1),\  (\phi_4^\dagger \phi_3)(\phi_2^\dagger \phi_3),\ (\phi_4^\dagger \phi_2)(\phi_1^\dagger \phi_2)$ & ${2\pi \over 7}(0,\,2,\,3,\,4)$  \\[3mm]
  $\Z_8$ &  $(\phi_4^\dagger \phi_3)(\phi_2^\dagger \phi_3),\  (\phi_4^\dagger \phi_2)(\phi_1^\dagger \phi_2),\ (\phi_1^\dagger \phi_4)^2$ & ${2\pi \over 8}(0,\,2,\,3,\,4)$  \\ [3mm]
\hline
\end{tabular}
\caption{Cyclic groups realizable as symmetry groups in the scalar sector of 4HDM}
\end{center}
\label{tableDM}
\end{table}

Since the potential is invariant under the common phase shift of all doublets, one can freely add additional equal phases to the generators shown in the third column in Table \ref{tableDM} and, in addition, one can permute the doublets. For example, the last line of this table can be replaced by $(\phi_3^\dagger \phi_2)(\phi_1^\dagger \phi_2),\  (\phi_3^\dagger \phi_1)(\phi_4^\dagger \phi_1),\ (\phi_4^\dagger \phi_3)^2$, $(\phi_4 \rightarrow \phi_3 \rightarrow \phi_2 \rightarrow \phi_1 \rightarrow \phi_4)$, which is symmetric under the $\Z_8$ group generated by phase rotations ${2\pi \over 8}\, (4,\, 0,\, 2,\, 3)\equiv {2\pi \over 8}\, (2,\, -2,\, 0,\, 1)\equiv {2\pi \over 8}\, (2,\, 6,\, 0,\, 1)\equiv {2\pi \over 8}\, (0,\, 1,\, 2,\, 6)$.

The patterns of phase shifts given in this table allow for construction of various dark sectors with different possibilities for dark matter dynamics. Here, we show that there are examples in 4HDM in which semi-annihilation of dark matter candidates is also forbidden. For this purpose we pick a $Z_7$-symmetric example.

\subsection{$\Z_7$ symmetric example}

Consider the $\Z_7$-symmetric 4HDM with the potential;
\bea
V = V_0+V_{Z_7} &=& - m_1^2 (\phi_1^\dagger \phi_1) + |m_2^2| (\phi_2^\dagger \phi_2) + |m_3^2| (\phi_3^\dagger \phi_3)  + |m_4^2| (\phi_4^\dagger \phi_4) \\
&& + \lambda_0 \left[(\phi_1^\dagger \phi_1)^2+ (\phi_2^\dagger \phi_2)^2 + (\phi_3^\dagger \phi_3)^2 + (\phi_4^\dagger \phi_4)^2\right] \nonumber\\
&& + \lambda_{1}(\phi_4^\dagger\phi_1)(\phi_3^\dagger\phi_1) + 
\lambda_{2}(\phi_4^\dagger\phi_2)(\phi_1^\dagger\phi_2) + 
\lambda_{3}(\phi_4^\dagger\phi_3)(\phi_2^\dagger\phi_3) + h.c.\nonumber  
\eea
As before, we assume that only the first doublet couples to fermions $(v/\sqrt{2},\,0,\,0,\,0)$, then we expand the doublets; 
\be
\phi_1 = \doublet{G^+}{{1\over\sqrt{2}}(v + h + i G^0)},
\phi_{2} = \doublet{w_{2}^+}{z_{2}},
\phi_{3} = \doublet{w_{3}^+}{z_{3}}, 
\phi_{4} = \doublet{w_{4}^+}{z_{4}} 
\ee
Similar to the previous example, $h$ is the SM-like Higgs boson, $G^0$ and $G^+$ are the Goldstone bosons, 
while $w_2^+,\, w_3^+, w_4^+$ and $z_2, z_3, z_4$ are charged and neutral scalar bosons, respectively.

Now we assign the $\Z_7$ quantum numbers according to the generator $a=2\pi/7 (0, 2, 3, 4)$. The SM-like Higgs boson gets $\Z_7$-charge $q=0$. Fields $w_2^+$ and $z_2$, from the second doublet, have $\Z_7$-charge $q=2$, while for $w_3^+$ and $z_3$, from the third doublet, $q=3$, and for $w_4^+$ and $z_4$, from the forth doublet, $q=4$. The conjugates of fields$w_2^-$ and $z_2^*$ have $q=5$ ( $=-2$ mod 7), and $w_3^-$ and $z_3^*$ have $q=4$ ( $=-3$ mod 7), while for $w_4^-,\, z_4^*$ $q=3$ ( $=-4$ mod 7). Note that the neutral scalars with similar quantum numbers can mix.

Let's rewrite the potential in terms of the new parameters;
\bea
V_{2} &=& \left[-m_1^2 v +\lambda_0 v^3 \right]h  \\
  &&+ \left[-\frac{m_1^2}{2} + \frac{3\lambda_0}{2} v^2 \right](h^2) \nonumber\\
  &&+ \left[-m_1^2 +\lambda_0 v^2\right](G^-G^+) + \left[-\frac{m_1^2}{2} + \frac{\lambda_0v^2}{2} \right](G_0^2)  \nonumber\\
  &&+ |m_2^2|(w_2^- w_2^+) + |m_3^2|(w_3^- w_3^+) + |m_4^2|(w_4^- w_4^+) + |m_2^2|(z_2^* z_2) \nonumber\\
  &&+ |m_3^2|(z_3^* z_3) + |m_4^2|(z_4^* z_4) + \left[\frac{\lambda_1v^2}{2} \right](z_3^* z_4^*) + \left[\frac{\lambda_1^*v^2}{2} \right] (z_3 z_4) \nonumber
 \eea
With the value $v^2= \frac{m_1^2}{\lambda_0}$, one finds the mass spectrum; The SM-like Higgs boson has mass $m_h^2 = 2m_1^2$, while the masses of the charged scalar bosons are $m_{w_2^\pm}^2 = |m_2^2|$, $m_{w_3^\pm}^2 = |m_3^2|$, $m_{w_4^\pm}^2 = |m_4^2|$ and for the neutral scalar $m_{z_2}^2 = |m_2^2|$.

Neutrals with equal $q$ can mix, which happens at $\lambda_1 \not = 0$. Again we describe the resulting mass eigenstates by complex fields $d$ and $D$ ($m_d < m_D$), both having $q=3$: 
\be
m_{D,d}^2 = {|m_3^2| + |m_4^2| \over 2} \pm {1 \over 2}\sqrt{(|m_3^2| - |m_4^2|)^2 + {|\lambda_1|^2 \over \lambda_0^2}m_1^4} 
\ee
where $d$ and $D$ are rotations of $z_3$ and $z_4^*$;
\bea
&& d = \cos\alpha\, z_3 + \sin\alpha e^{-i\beta}\, z_4^*\,,\quad D = - \sin\alpha e^{i\beta}\, z_3 + \cos\alpha\, z_4^* \nonumber\\
&& \tan2\alpha = {|\lambda_1| \over \lambda_0}{m_1^2 \over |m_3^2|-|m_4^2|}\,,\quad \beta = \mbox{arg}\, \lambda_1 
\eea
Here we have
\be
m_d < m_{w_3^\pm},\, m_{w_4^\pm} < m_D\,,\quad m_{w_3^\pm}^2 + m_{w_4^\pm}^2 = m_d^2 + m_D^2 
\ee

By adjusting free parameters, one can easily make $d$ lighter than $w_2^\pm$ and $z_2$, and therefore the lightest among the particles that transform non-trivially under $\Z_7$. Then, the other particles will either eventually decay to $d$ or $d^*$ plus SM particles or will be stable representing an additional contribution to dark matter. Again, if an asymmetry between $d$ or $d^*$ exists and if the rate of their annihilation is high, 
then after the freeze-out we are left with the gas predominantly made of $d$'s.

\subsection{Higher order terms in the $Z_7$ symmetric potential}

The subsequent microscopic dynamics depend on the interactions between $d$'s and $z_2$'s which follow from the $Z_7$-symmetric part of the potential.
After eliminating the Goldstone bosons and rewriting the potential in terms of $d$ and $D$, the interactions between $d$'s and $z_2$'s in the higher order terms in the potential, regardless of the coefficients are of the following form
\bea
V_{3} &=& \lambda_1 \left[ hd^*d , hd^*D , hdD^* , hD^*D \right]   \\
&&+\lambda_2 \left[ z_2z_2d , z_2z_2D  \right] +h.c. \nonumber \\
V_{4} &=& \lambda_0 \left[ ddd^*d^*,  ddD^*D^*,  d^*d^*DD,  DDD^*D^* \right] \nonumber \\
&&+ \lambda_1 \left[ hhd^*d , hhd^*D , hhdD^* , hhD^*D \right]  \nonumber \\
&&+ \lambda_2 \left[ hz_2z_2d , hz_2z_2D  \right] +h.c. \nonumber \\
&&+ \lambda_3 \left[ z_2^*w_4^- w_3^+ d , z_2^*w_4^- w_3^+ D , z_2^*ddd , z_2^*ddD , z_2^*dDD , z_2^*DDD\right] +h.c. \nonumber
\eea
The relevant terms are $hd z_2 z_2$ from the $\lambda_2$ term and $dddz_2^*$ from the $\lambda_3$ term.
Since $m_d < m_{z_2} $, one- or two-particle processes such as $d\to z_2^* z_2^*$, $dd \to d^*z_2$, and $dd \to z_2z_2z_2$ are 
all kinematically forbidden, shown in Figure (\ref{Int1}):
\begin{figure} [ht]
\centering
\includegraphics[height=2.7cm]{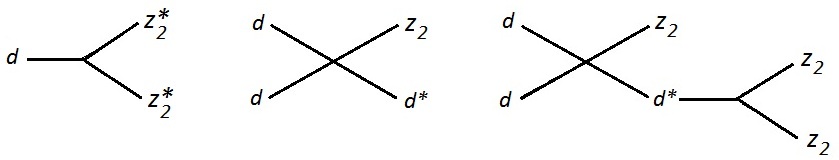}
\caption{Forbidden 1- and 2-particle processes}
\label{Int1}
\end{figure}

Multiple collision kinetics depends on whether $m_{z_2} < 3 m_d$ or not. If $z_2$ is not too heavy, then the "triple semi-annihilation", $ddd \to z_2 X_{SM}$, is kinematically allowed and will create a population of $z_2$ even if it was absent before.  However, $z_2$ will get depleted by the semi-annihilation process $z_2 z_2 \to d^* X_{SM}$.
So, if one starts with a certain concentration of $z_2$, $d$ and their antiparticles, then $z_2$ will die off at a higher rate than $d$'s. 

In stationary conditions, the terminal concentrations will be those equilibrating the rates of the following $6d$ tree-level scattering with intermediate $z_2$'s, $6d \to z_2 z_2 X_{SM}\to d^* X'_{SM}$ and the subsequent annihilation of $d^*$. The net result of this chain will be the "$7d$-burning process", $7d \to X_{SM}$, the bottleneck in this chain being the triple-$d$ process $ddd \to z_2 X_{SM}$, shown in Figure (\ref{Int3}).
\begin{figure} [ht]
\centering
\includegraphics[height=4.5cm]{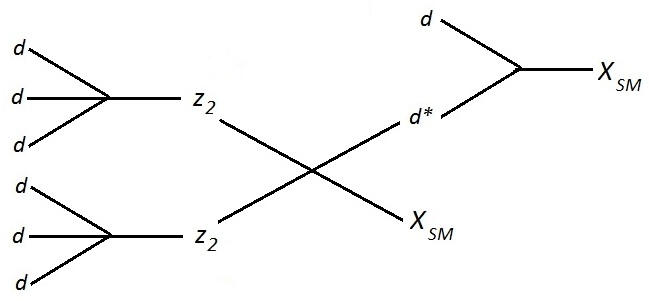}
\caption{The $7d$-burning processes, with triple-$d$ process bottleneck}
\label{Int3}
\end{figure}

On the other hand, if $m_{z_2} > 3 m_d$, then $ddd \to z_2$ is kinematically forbidden, while the inverse process leads to a quick $z_2$ decay. In this case, one can still burn $d$'s via the tree-level process with intermediate virtual $z_2$'s, $dddd \to d^*d^* z_2 z_2 \to d^*d^*d^* X_{SM}$. The net result will be the same $7d$-burning, but the bottleneck process is now the $4d$ collision, whose rate is even more suppressed, shown in Figure (\ref{Int2}).
\begin{figure} [ht]
\centering
\includegraphics[height=6.5cm]{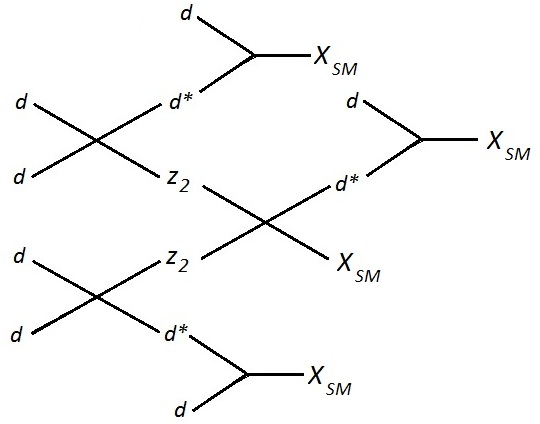}
\caption{The $7d$-burning processes, with $4d$ process bottleneck}
\label{Int2}
\end{figure}

\chapter{Summary}\label{chapter8}

In this thesis we reviewed the Brout-Englert-Higgs mechanism of spontaneously symmetry breaking. This mechanism requires the existence of a scalar boson, the so called Higgs boson. However, there is no reason to restrict ourselves to only one scalar doublet, and there are many motivations for introducing more than one doublet. We briefly studied the two-Higgs-doublet model in a particular case.

One of the interesting topics in multi-Higgs-doublet models, is the symmetries that one can impose on the scalar potential. It is very helpful to study these symmetries in the "orbit space", which was presented here in the general NHDM. We studied the particular case of three-Higgs-doublet model orbit space in detail. 

From the analysis of the orbit space we found certain symmetries that are always broken after EWSB in multi-Higgs-doublet models with more than two doublets. We named them frustrated symmetries because of their resemblance to the phenomenon of frustration in condensed matter physics. These symmetries are not specific to doublets. They can arise when the representation of the electroweak group has lower dimensionality than the horizontal (Higgs family) space, i.e. more than one singlet, more than two doublets, more than three triplets, etc. 

Several examples of frustrated symmetries in the three-Higgs-doublet model were given; $SU(3)$, octahedral, tetrahedral, and $Z_3 \times Z_3$ symmetry. The case of octahedral symmetry was studied further. This symmetry seems to be the largest realizable discrete symmetry that can be imposed on the scalar sector in 3HDM. Another interesting feature of this case is that it is 2HDM-like, meaning that the mass spectrum of the scalar potential with Octahedral symmetry is similar to the one of 2HDM, with a degeneracy in the mass of the charged scalar boson, and three different masses for the neutral scalar bosons. It would be interesting to see what experimental observables could distinguish this model from the actual 2HDM. 

In this thesis we also made a step towards classification of possible symmetries of the scalar sector of the NHDM. Namely, we studied which Abelian groups can be realized as symmetry groups of the NHDM potential. We proved that they can be either subgroups of the maximal torus, or certain finite Abelian groups which are not subgroups of maximal tori. For the subgroups of the maximal torus, we developed an algorithmic strategy that gives the full list of possible realizable Abelian symmetries for any given $N$. We illustrated how the strategy works with two small-$N$ examples. We also proved several statements concerning the order and structure of finite groups which apply for arbitrary $N$. Finally, we conjectured that, in the NHDM, any finite Abelian group with order $\le 2^{N-1}$ is realizable.

We extended the strategy of finding Abelian unitary symmetries, to include Abelian antiunitary symmetries (with generalized $CP$ transformations) in NHDM. We introduced a strategy that could be used for any NHDM, and applied it to the case of 3HDM, where the full list of realizable Abelian antiunitary symmetries for 3HDM was found to be $\Z_2^*,\quad \Z_2\times \Z_2^*, \quad \Z_2\times \Z_2\times \Z_2^*,  \quad \Z_4^*$.

An obvious direction of future research is to understand phenomenological consequences of the symmetries found. This includes, in particular, a study of how the symmetries of the potential are broken, and what is the effect of these symmetries on the physical Higgs boson spectrum.

Having found that $Z_p$ symmetries are realizable in multi-Higgs-doublet models, we showed that these models can naturally accommodate
scalar dark matter candidates protected by the group $\Z_p$. These models do not require any significant fine-tuning and can lead to a variety of forms of microscopic dynamics among the dark matter candidates (allowing or forbidding semi-annihilation, offering different routes to multi-particle annihilation, etc.). 

In particular, we gave explicit examples of dark sectors where the bottleneck process leading to depletion of asymmetric dark matter can be a 2-particle, 3-particle or 4-particle semi-annihilation. We stress that these models do not require any serious fine-tuning. We only ask for the presence of terms invariant under the chosen symmetry group but do not constrain coefficients in front of these terms.

In certain aspects these models resemble the Inert Doublet Model \cite{inert}, but in the other they rely on symmetry patterns that arise only with several doublets. In this respect, such models can be viewed as "multi-inert" doublet models although this name of course does not completely specify the microscopic dynamics. Exploring the observational consequences of each sort of microscopic dynamics is a separate study.


\appendix

\chapter{Definitions}

\section*{Basic definitions}

\begin{itemize}

\item  ${\bf Group}$.
A group G is a set with an operation that combines any two elements a and b to form another element. To qualify as a group, G must satisfy four requirements known as the group axioms:
\\Closure: For all a and b in G, the result of the operation, a 
 b, is also in G.
\\Associativity: For all a, b and c in G, $(a • b) • c = a • (b • c)$.
\\Identity element: There exists an element e in G, such that for every element a in G, $e • a = a • e = a$. 
\\Inverse element: For each a in G, there exists an element $b=a^{-1}$ in G such that $a • b = b • a = e$.

\item  ${\bf Group \quad order}$.
The order of a group G, denoted $|G|$, is the number of its elements.

\item  ${\bf Subgroup}$.
A subset H of G is called a subgroup of G if H also forms a group under the group operation.

\item  ${\bf Proper \quad subgroup}$.
If A is a subgroup of B, but not equal to B, meaning there exists at least one element of B not contained in A, then A is also a proper (or strict) subgroup of B; $A\subsetneq B$.

\item  ${\bf Maximal \quad subgroup}$.
A maximal subgroup H of a group G is a proper subgroup, such that no other proper subgroup K contains H.
A subgroup of a group is termed  maximal Abelian subgroup if it is Abelian and is not properly contained in a bigger Abelian subgroup. 

\item ${\bf Abelian \quad group}$.
An Abelian group, also called a commutative group, is a group in which the result of applying the group operation on two group elements does not depend on their order. Example: the groups of even integers $E=\{0,\pm 2,\pm 4,\pm 6, \cdots \}$ is Abelian under addition operation.

\item  ${\bf Cyclic \quad group}$.
A cyclic group is a group that can be generated by a single element g (the group generator), by applying the group operation on g as many times as needed. A cyclic group of finite order n is usually denoted as $Z_n$, and its generator g satisfies $g^n=I$, where I is the identity element. Cyclic groups are Abelian.

\item ${\bf Unitary \quad group}$
A unitary group of order N, $U(N)$, is the group of $N \times N$ unitary matrices.

\item ${\bf Special \quad unitary \quad group}$
The special unitary group of order N, $SU(N)$, is the group of $N \times N$ unitary matrices with determinant 1.

\item ${\bf Unitary \quad transformation}$
A unitary transformation is a transformation that respects the inner product of two vectors.



\section*{Action}

\item  ${\bf Group \quad action}$.
A group action is a way of describing symmetries of objects using groups. 
A group G (with operation $\cdot$) is said to act on a set X when the operator $\circ : G \times X \to X$ satisfies the following conditions for all elements $x \in X$:
\\Associativity: $(g_1 \cdot g_2) \circ x = g_1 \circ (g_2 \circ x), \forall g_1, g_2 \in G$
\\Identity: $e \circ x = x$
The action of a group on a group is a straightforward generalization of the action of a group on a set.
A group can also act on itself. One of the examples of such action is conjugation.


\item  ${\bf Conjugation}$.
In a group G, two elements a and b are called conjugate if there exists an element g in G with $gag^{-1}= b$. Conjugation preserves the group operation, therefore it's a homomorphism.

\item  ${\bf Conjugacy \quad class}$.
Conjugation partitions the group into conjugacy classes with similar properties. If a is an element of the group G, the conjugacy class of a has the form: $Cl(a) = { gag^{−1} \mid g \in G }$. Every element of the group G belongs to precisely one conjugacy class, and the classes $Cl(a)$ and $Cl(b)$ are equal if and only if a and b are conjugate. 

\item  ${\bf Conjugate \quad subgroup}$.
If H is a subgroup of G, with $h \in H$ and g a fixed element of G which is not a member of H, then the transformation $ghg^{-1}$ generates the subgroup $gHg^{-1}$ which is conjugate to the subgroup H.

\item  ${\bf Homomorphism}$.
A group homomorphism is a map $f \mid G \rightarrow H$ between two groups such that the group operation is preserved $f(g_1 \cdot g_2)=f(g_1) \circ f(g_2)$ for all $g_1,g_2$ in G, where the product on the left-hand side is the operation in G and on the right-hand side in H. A group homomorphism maps the identity element in G to the identity element in H $f(e_G)=e_H$. 

\item  ${\bf Isomorphism}$.
Isomorphism is a homomorphism between two groups with a one to one correspondence between group elements.

\item  ${\bf Automorphism}$.
A group automorphism $f$ is an isomorphism from a group $G$ to itself  $f \mid G \rightarrow G$. $f$ satisfies the following conditions: $f(gh) = f(g)f(h) \mid  g_1,g_2 \in G$.

\item  ${\bf Commutator}$.
The commutator of elements $g_1$ and $g_2$ of a group G is $[g_1,g_2] = {g_1}^{-1} {g_2}^{-1} g_1 g_2$. The commutator $[g_1,g_2]$ is equal to the identity element ${e}$ if and only if $g_1 g_2 = g_2 g_1$, that is, if and only if $g_1$ and $g_2$ commute. 

\section*{Structure}

\item  ${\bf Center \quad of \quad the \quad group}$.
The center of a group G, denoted $Z(G)$, is the set of elements that commute with every element of G: $Z(G) = \{z \in G \mid \forall g\in G, zg = gz \}$. The center is a subgroup of G, which by definition is Abelian.

\item  ${\bf Normal \quad subgroup}$.
N is called a normal subgroup of G if for each $n \in N$ and $g \in G$, the element $gng^{-1}$ is still in N : $ N \triangleleft G\,\,\Leftrightarrow\,\forall\,n\in{N},\forall\,g\in{G}\ , gng^{-1}\in{N}$. 

\item  ${\bf Coset}$.
In a group G, with subgroup H, and g an element of G, one defines $ gH = \{gh \mid h \in H \}$as a left coset of H in G, and $Hg = \{hg \mid h \in H \}$ as a right coset of H in G. 

\item  ${\bf Factor \quad (quotient) \quad group}$.
For a group G and its normal subgroup N, the quotient group of N in G, $ G/N$, is the set of cosets of N in G; $G/N = { aN \mid a \in G }$. The group operation on $G/N$ is the product of subsets, meaning for each $aN$ and $bN$ in $G/N$, the product of $aN$ and $bN$ is $(aN)(bN)$. This operation is closed, because $(aN)(bN)$ is a (left) coset itself: $(aN)(bN) = a(Nb)N = a(bN)N = (ab)NN = (ab)N$. 

\item  ${\bf Generating \quad set}$.
A generating set of a group G is a subset S of G, denoted $ \langle S \rangle$, such that every element of the group can be expressed as the combination (under the group operation) of elements of the subset and their inverses.

\section*{Construction}

\item  ${\bf Direct \quad product}$.
The direct product of two groups G and H, denoted $G \times H$, is defined as follows: The elements of $G \times H$ are ordered pairs $(g, h)$, where $g \in G$ and $h \in H$. That is, the set of elements of $G \times H$ is the Cartesian product of the sets G and H. The operation on $G \times H$ is defined: $(g_1, h_1)\cdot(g_2, h_2)=(g_1 \cdot g_2, h_1 \cdot h_2)$. Both groups G and H are normal subgroups of the group $G \times H$.

\item  ${\bf Semidirect \quad product}$.
A semidirect product is a generalization of the notion of direct product. If N is a normal subgroup of G and H a subgroup of G, we say that G is a semidirect product of N and H, $G = N \rtimes H$ when $G = NH$ and $N \cap H = {e} $. 

Let $Aut(N)$ be the group of all automorphisms of N, and $f$ be a group homomorphism where $f \mid H \rightarrow Aut(N)$ defined by $f(h_n)= hnh^{−1}, \forall h \in H,\forall n \in N$. Then the operation in $G = N \rtimes H$ is determined by the homomorphism $f$, and is as follows; $$ * : (N\rtimes H)\times(N\rtimes H) \rightarrow N\rtimes H$$ defined by $(n_1, h_1)*(n_2, h_2) = (n_1 f({h_1}_{n_2}), h_1h_2)$ for $n_1, n_2 \in N$ and $h_1, h_2 \in H$.

\item  ${\bf Torus}$.
In geometry, a torus is a surface generated by rotating a circle in three dimensional space about an axis coplanar with the circle. In topology, a torus is homomorphic to the Cartesian product of two circles: $S_1 \times S_1$. An n-dimensional torus is the product of n circles. That is: $T^n = \underbrace{S^1 \times S^1 \times \cdots \times S^1}_n$. 
\\In group theory, a torus in a Lie group G is an Abelian Lie subgroup of G. A maximal torus is one which is maximal among such subgroups. Every torus is contained in a maximal torus. For example, the unitary group $U(N)$ contains a maximal torus, the subgroup of all diagonal matrices, of the form: $T = \left\{\mathrm{diag}(e^{i\theta_1},e^{i\theta_2},\dots,e^{i\theta_n}) : \theta_j \in \mathbb R\right\}$. 

\end{itemize}

\chapter{Tetrahedral and octahedral symmetric potentials}

\section*{The potential with tetrahedral symmetry}

The "tetrahedral" potential used in Section \ref{3HDM-examples} as an example of a frustrated symmetry is  defined by the following scalar potential:
\bea
V & = &  - M_0 r_0 + \Lambda_0 r_0^2 + \Lambda_1(r_1^2+r_4^2+r_6^2) + \Lambda_2(r_2^2+r_5^2+r_7^2) 
+ \Lambda_3(r_3^2+r_8^2) \nonumber \\  \nonumber
&& + \Lambda_4(r_1r_2 - r_4r_5 + r_6r_7) \\  \nonumber
\eea
In terms of doublets, this potential has the form:
\bea
V&=& - {M_0 \over \sqrt{3}} (\fd_1 \f_1 + \fd_2 \f_2 + \fd_3 \f_3) \nonumber \\ \nonumber
&&+ {\Lambda_0 \over 3} \left[(\fd_1 \f_1) + (\fd_2 \f_2) + (\fd_3 \f_3)\right]^2 \\  \nonumber
&&+ {\Lambda_1 \over 3} \left[(\fd_1 \f_1)^2 + (\fd_2 \f_2)^2 + (\fd_3 \f_3)^2-(\fd_1 \f_1)(\fd_2 \f_2)-(\fd_2 \f_2)(\fd_3 \f_3)-(\fd_3 \f_3)(\fd_1 \f_1)\right]\\  \nonumber
&&+ \Lambda_2 \left[(\Re \fd_1\f_2)^2 + (\Re \fd_2\f_3)^2 + (\Re \fd_3\f_1)^2\right]  \\ \nonumber
&&+ \Lambda_3 \left[(\Im \fd_1\f_2)^2 + (\Im \fd_2\f_3)^2 + (\Im \fd_3\f_1)^2\right] \\  \nonumber
&&+ \Lambda_4 \left[(\Re \fd_1\f_2) (\Im \fd_1\f_2) + (\Re \fd_2\f_3) (\Im \fd_2\f_3) + (\Re \fd_3\f_1) (\Im \fd_3\f_1)\right]\\ \nonumber
\eea
This potential is symmetric under the following transformations: independent sign flips of the doublets, the cyclic permutation of the three doublets, and their compositions, which we show as;

\be
\begin{tabular}{c | ccc}
& $\phi_1$ & $\phi_2$ & $\phi_3$\\
\hline
$M$ & $-$ &  & \\
$N$ &  & $-$ &  \\
$O$ &  &  & $-$  \\
$P$ & $\phi^{\dagger}_2$ & $\phi^{\dagger}_1$ & $CP$ \\
$Q$ & $\phi^{\dagger}_3$ & $CP$ & $\phi^{\dagger}_1$ \\
$R$ & $CP$ & $\phi^{\dagger}_3$ &  $\phi^{\dagger}_2$ \\
$S$ & $\phi_3$ & $\phi_1$ & $\phi_2$ \\
\end{tabular} \nonumber
\ee
Therefore the symmetry group of the potential in terms of these transformation has 24 elements:
\bea
S &=& \lbrace e, M, N, O, P, Q, R, S, S^2, PM, QM, RM, SM, S^2M, \nonumber \\
&& PN, QN, RN, SN, S^2N, PO, QO, RO, SO, S^2O  \rbrace \nonumber
\eea

The cycle graph of such group is shown in Figure (\ref{A4-cyclic-graph}), which is the cycle graph of a tetrahedron. 

\begin{figure} [ht]
\centering
\includegraphics[height=6cm]{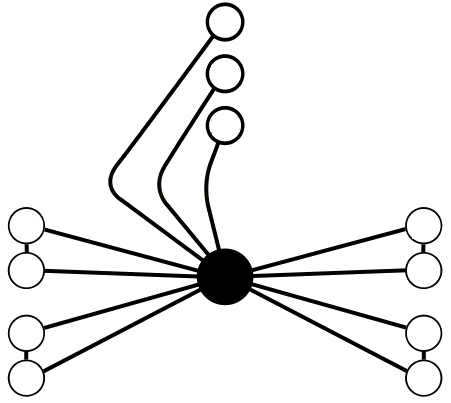}
\caption{Cycle graph of the group $A_4$, the symmetry group of a tetrahedron.}
\label{A4-cyclic-graph}
\end{figure}

\section*{The potential with octahedral symmetry}

The "octahedral" potential used in Section \ref{3HDM-examples} as an example of a frustrated symmetry is  defined by the following scalar potential:
\bea
V&=&-M_{0}r_{0}+\Lambda_{0}r_{0}^2+\Lambda_{1}(r_{1}^2+r_{4}^2+r_{6}^2)+\Lambda_{2}(r_{2}^2+r_{5}^2+r_{7}^2)+\Lambda_{3}(r_{3}^2+r_{8}^2) \nonumber \\ \nonumber
\eea
Written in terms of doublets, the potential has the following form: 
\bea
V&=&-\frac{M_0}{\sqrt{3}}\left(\phi_1^{\dagger}\phi_1+\phi_2^{\dagger}\phi_2+\phi_3^{\dagger}\phi_3\right)+\frac{\Lambda_0}{3}\left(\phi_1^{\dagger}\phi_1+\phi_2^{\dagger}\phi_2+\phi_3^{\dagger}\phi_3\right)^2 \nonumber \\ \nonumber
&&+\Lambda_1\left[(\Re\phi_1^{\dagger}\phi_2)^2+(\Re\phi_2^{\dagger}\phi_3)^2+(\Re\phi_3^{\dagger}\phi_1)^2\right]\\ \nonumber
&&+\Lambda_2\left[(\Im\phi_1^{\dagger}\phi_2)^2+(\Im\phi_2^{\dagger}\phi_3)^2+(\Im\phi_3^{\dagger}\phi_1)^2\right]\\ \nonumber
&&+\frac{\Lambda_3}{3}\left[(\phi_1^{\dagger}\phi_1)^2+(\phi_2^{\dagger}\phi_2)^2+(\phi_3^{\dagger}\phi_3)^2-(\phi_1^{\dagger}\phi_1)(\phi_2^{\dagger}\phi_2)-(\phi_2^{\dagger}\phi_2)(\phi_3^{\dagger}\phi_3)-(\phi_3^{\dagger}\phi_3)(\phi_1^{\dagger}\phi_1)\right] \\ \nonumber
\eea
This potential is symmetric under the following transformations: independent sign flips of the doublets, the cyclic permutation of the three doublets, the $CP$-transformation and their compositions, which we show as;
\be
\begin{tabular}{c | ccc}
& $\phi_1$ & $\phi_2$ & $\phi_3$\\
\hline
$M$ & $-$ &  & \\
$N$ &  & $-$ &  \\
$O$ &  &  & $-$  \\
$P$ & $\phi_2$ & $\phi_1$ & \\
$Q$ & $\phi_3$ &  & $\phi_1$ \\
$R$ &  & $\phi_3$ &  $\phi_2$ \\
$S$ & $\phi_3$ & $\phi_1$ & $\phi_2$ \\
$T$ & $CP$ & $CP$ & $CP$  \\
\end{tabular} \nonumber
\ee
Therefore the symmetry group of the potential in terms of these transformation has 48 elements:
\bea
S = && \lbrace e, M, N, O, P, Q, R, S, S^2, T, MT, NT, OT, PT, QT, RT, ST, S^2T, \nonumber \\  
&& PM, QM, RM, SM, S^2M, PN, QN, RN, SN, S^2N, PO, QO, RO, SO, S^2O,  \nonumber \\  
&& PMT, QMT, RMT, SMT, S^2MT, PNT, QNT, RNT, SNT, S^2NT, \nonumber \\
&& POT, QOT, ROT, SOT, S^2OT   \rbrace \nonumber 
\eea

The cycle graph of such group is shown in Figure (\ref{S4-cyclic-graph}), which is the cycle graph of an octahedron.

\begin{figure} [ht]
\centering
\includegraphics[height=8.5cm]{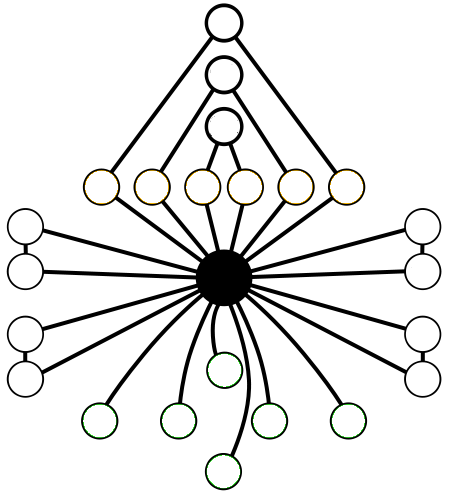}
\caption{Cycle graph of the group $S_4$, the symmetry group of an octahedron.}
\label{S4-cyclic-graph}
\end{figure}

\chapter{Detailed study of the symmetric minimum $(v,v,v)$}\label{detailed}

A remarkable phenomenological feature of a 3HDM with Octahedral symmetry is that it is 2HDM-like. After symmetry breaking, due to remaining symmetry, this model exhibits certain degeneracy in the mass spectrum of the physical Higgs bosons, which precisely mimics the typical Higgs spectrum of 2HDM. It would be interesting to see what phenomenological consequences distinguish these two models. Here we pick the symmetric the minimum point, $(v,v,v)$, and study the potential in more detail.

\section*{Mass matrix eigenstates}

We study the potential at the vicinity if this minimum:
$$\phi_a=\left(
\begin{array}{c}
         w_a^{+}  \\
         \frac{v+h_a+\eta_a}{\sqrt{2}}
\end{array}
\right)$$ 
with $$v^2=\frac{M_0}{\sqrt{3}(\Lambda_0+\Lambda_1)} $$ the mass eigenstates are:
\bea
n_1^+=\frac{1}{\sqrt{2}}(w_2^+-w_3^+) \quad &;& \quad m^2_{n_1}=-\frac{3}{2}\Lambda_1v^2  \nonumber \\  
n_2^+=\frac{1}{\sqrt{3}}(w_1^++w_2^++w_3^+) \quad &;& \quad m^2_{n_2}=0  \nonumber \\ 
n_3^+=\frac{1}{\sqrt{6}}(-2w_1^++w_2^++w_3^+)\quad &;& \quad  m^2_{n_3}=-\frac{3}{2}\Lambda_1v^2 \nonumber \\ 
p_1=\frac{1}{\sqrt{2}}(h_2-h_3) \quad &;& \quad  m^2_{p_1}=\frac{-\Lambda_1+\Lambda_3}{2}v^2  \nonumber \\ 
 p_2=\frac{1}{\sqrt{3}}(h_1+h_2+h_3) \quad &;& \quad  m^2_{p_2}=(\Lambda_0+\Lambda_1)v^2 \nonumber \\ 
p_3=\frac{1}{\sqrt{6}}(-2h_1+h_2+h_3)\quad &;& \quad  m^2_{p_3}=\frac{-\Lambda_1+\Lambda_3}{2}v^2 \nonumber \\ 
q_1=\frac{1}{\sqrt{2}}(\eta_2-\eta_3) \quad &;& \quad  m^2_{nq_1}=\frac{3}{4}(-\Lambda_1+\Lambda_2)v^2 \nonumber \\ 
q_2=\frac{1}{\sqrt{3}}(\eta_1+\eta_2+\eta_3) \quad &;& \quad  m^2_{q_2}=0 \nonumber \\ 
q_3=\frac{1}{\sqrt{6}}(-2\eta_1+\eta_2+\eta_3) \quad &;& \quad  m^2_{q_3 }=\frac{3}{4}(-\Lambda_1+\Lambda_2)v^2 \nonumber 
\eea

The three Nambu-Goldstone bosons that give masses to the $W^+$, $W^-$ and $Z$ are $n_2^+$, $n_2^-$ and $q_2$.

\section*{Decay rates}

From the third order terms in the potential we have the decays of the Higgs bosons:
\begin{itemize}
\item $p_2 \rightarrow n_1^+ n_1^- \quad$ and  $\quad p_2 \rightarrow n_3^+ n_3^-\quad$  with vertex coefficient $\frac{\sqrt{3}}{3}(2\Lambda_0-\Lambda_{1})v$
\item $p_2 \rightarrow p_1 p_1 \quad $  with vertex coefficient $\quad \frac{\sqrt{3}}{3}(\Lambda_0+\Lambda_{3})v$
\item $p_2 \rightarrow p_3 p_3 \quad$  with vertex coefficient $\quad \frac{\sqrt{3}}{6}(\Lambda_1+2\Lambda_{3})v$
\item $p_2 \rightarrow q_1 q_1 \quad$ and  $\quad p_2 \rightarrow q_3 q_3 \quad $  with vertex coefficient $ \quad \frac{\sqrt{3}}{2}(\Lambda_0-\frac{\Lambda_{1}}{2})v+\frac{\sqrt{3}}{2}\Lambda_2v$
\item $p_3 \rightarrow n_1^+ n_1^- \quad $ with vertex coefficient $\quad \frac{\sqrt{6}}{6}(\Lambda_1-\Lambda_{3})v$
\item $p_3 \rightarrow n_3^+ n_3^- \quad $ with vertex coefficient $ \quad -\frac{\sqrt{6}}{6}(\Lambda_1-\Lambda_{3})v$
\item $p_3 \rightarrow q_1 q_1 \quad $ and  $\quad p_3 \rightarrow q_3 q_3 \quad$  with vertex coefficient $\quad -\frac{\sqrt{6}}{12}(-\Lambda_1+\Lambda_{3})v$
\item $p_1 \rightarrow n_1^+ n_3^-\quad $ and  $\quad p_1 \rightarrow n_1^- n_3^+ \quad $  with vertex coefficient $ \quad \frac{\sqrt{6}}{6}(\Lambda_3-\Lambda_{1})v$
\item $p_1 \rightarrow q_1 q_3 \quad $  with vertex coefficient $ \quad \frac{\sqrt{6}}{6}(\Lambda_3-\Lambda_{1})v$
\end{itemize}

One could easily calculate that the particles with the same mass $p_1, p_3$ and $q_1,q_3$, have the same decay rate. Following the straightforward method of writing the kinetic term, to calculate the couplings to the gauge bosons results in:

\begin{itemize}
\item $p_2 \rightarrow W^+ W^-\quad$ with vertex coefficient $\quad \frac{\sqrt{3}}{2}g_2^2v$ 
\item $p_2 \rightarrow Z Z \quad $ with vertex coefficient $ \quad\frac{\sqrt{3}}{4}(g_1^2+g_2^2)v$ 
\end{itemize}

Let's assume $\Lambda_0 >> \Lambda_2 > \Lambda_3 > |\Lambda_1| >0$. With this choice the mass spectrum will have the form $m_{p_2} >> m_{q_1,3}> m_{p_1,3} > m_{n_1,3}> m_{n_2}, m_{q_2} = 0$. Therefore the decays $p_3 \rightarrow q_{1,3} q_{1,3}$ from the list of interactions above will be forbidden. It appears that we would have stable particles in this case, $q_{1,3} $ and $n_{1,3} $, which could be dark matter candidate.


\end{document}